\documentclass[acmsmall]{acmart}

\acmJournal{PACMPL}
\acmVolume{1}
\acmNumber{CONF} 
\acmArticle{1}
\acmYear{2018}
\acmMonth{1}
\acmDOI{} 
\startPage{1}

\setcopyright{none}

\usepackage{booktabs}   
\usepackage{subcaption}

\usepackage{epsfig}
\usepackage{color}
\usepackage{amsfonts,amssymb}
\usepackage{verbatim}
\usepackage{times}
\usepackage[linesnumbered,ruled,procnumbered,noend]{algorithm2e}
\usepackage{stmaryrd}

\usepackage{thmtools,mathtools}
\usepackage{thm-restate}

\usepackage[version=0.96]{pgf}
\usepackage{tikz}
\usetikzlibrary{arrows,shapes,snakes,automata,backgrounds,petri}
\usepackage{xcolor}
\usepackage{paralist}
\usepackage{xifthen,hhline}
\usepackage{mathpartir}

\usepackage{wrapfig}

\def\sectionautorefname{Sec.}
\def\figureautorefname{Fig.}

\usepackage{listings}

\usepackage[margin,draft]{fixme}
\fxsetup{theme=color,mode=multiuser}
\FXRegisterAuthor{gp}{agp}{GP}
\FXRegisterAuthor{cw}{acw}{CW}
\FXRegisterAuthor{ce}{ace}{CE}

\DeclareSymbolFont{largesymbolsA}{U}{txexa}{m}{n}
\SetSymbolFont{largesymbolsA}{bold}{U}{txexa}{bx}{n}
\DeclareFontSubstitution{U}{txexa}{m}{n}

\newcommand{\textred}[1]{\color{red}#1\color{black}}

\newif\ifshort
\shortfalse

\newcommand{\bigfrac}[2]{\frac{\raisebox{1ex}{$#1$}}{\raisebox{-1.5ex}{$#2$}}}

\newcommand{\forget}[1]{}

\newcommand{\crdtlin}{RA-linearizability}
\newcommand{\crdtlinearization}{RA-linearization}
\newcommand{\crdtlinearizable}{RA-linearizable}
\newcommand{\CRDTLin}{Replication-Aware Linearizability}
\newcommand{\CRDTLinshort}{RA\text{-}linearizability}

\newcommand{\aobj}{\ensuremath{\mathtt{o}}}
\newcommand{\ahist}{\ensuremath{{h}}}
\newcommand{\objs}{\ensuremath{\mathbb{O}}}
\newcommand{\acts}{\ensuremath{\mathbb{A}}}
\newcommand{\aact}{\ensuremath{a}}
\newcommand{\arep}{\ensuremath{\mathtt{r}}}
\newcommand{\reps}{\ensuremath{\mathbb{R}}}
\newcommand{\amethod}{\ensuremath{\mathsf{m}}}

\newcommand{\tsof}{\ensuremath{\mathsf{ts}}}

\newcommand{\methods}{\ensuremath{\mathbb{M}}}
\newcommand{\histories}{\ensuremath{\mathbb{H}\mathsf{ist}}}
\newcommand{\traces}{\ensuremath{\mathbb{T}\mathsf{r}}}
\newcommand{\datadomain}{\ensuremath{\mathbb{D}}}
\newcommand{\timestampdomain}{\ensuremath{\mathbb{T}}}

\newcommand{\powerset}[1]{\ensuremath{\mathcal{P}(#1)}}
\newcommand{\astate}{\ensuremath{\sigma}}
\newcommand{\abstate}{\ensuremath{\phi}}
\newcommand{\abstates}{\ensuremath{\Phi}}
\newcommand{\states}{\ensuremath{\Sigma}}

\newcommand{\aop}{\ensuremath{\mathsf{op}}}

\newcommand{\argv}{\ensuremath{a}}
\newcommand{\retv}{\ensuremath{b}}
\newcommand{\alabellong}[3][\amethod]{\ensuremath{#1(#2) \Rightarrow #3}}
\newcommand{\alabelshort}[2][\amethod]{\ensuremath{#1(#2)}}
\newcommand{\alabellongind}[4][\amethod]{\ensuremath{#1(#2) \overset{#4}{\Rightarrow}
    #3}}
\newcommand{\alabelobjind}[4][{\aobj.\amethod}]{\ensuremath{{#1}(#2) \overset{#4}{\Rightarrow}
    #3}}
\newcommand{\src}[2]{\ensuremath{\mathsf{gen}_{#1}(#2)}}
\newcommand{\dwn}[2]{\ensuremath{\mathsf{eff}_{#1}(#2)}}

\newcommand{\alabel}{\ensuremath{\ell}}
\newcommand{\refmap}{\ensuremath{\mathsf{abs}}}
\newcommand{\firstrep}{\ensuremath{\mathsf{qry}}}
\newcommand{\secondrep}{\ensuremath{\mathsf{upd}}}
\newcommand{\specOrSet}{\ensuremath{\Spec(\text{OR\text{-}Set})}}
\newcommand{\specTwoPSet}{\ensuremath{\Spec(\text{2P\text{-}Set})}}
\newcommand{\specLWWSet}{\ensuremath{\Spec(\text{Set})}}
\newcommand{\specReg}{\ensuremath{\Spec(\text{Reg})}}
\newcommand{\specMVReg}{\ensuremath{\Spec(\text{MV\text{-}Reg})}}
\newcommand{\specCounter}{\ensuremath{\Spec(\text{Counter})}}
\newcommand{\specRGA}{\ensuremath{\Spec(\text{RGA})}}
\newcommand{\specWooki}{\ensuremath{\Spec(\text{Wooki})}}
\newcommand{\specAddatOne}{\ensuremath{\Spec(\text{addAt1})}}
\newcommand{\specAddatTwo}{\ensuremath{\Spec(\text{addAt2})}}
\newcommand{\specAddatThree}{\ensuremath{\Spec(\text{addAt3})}}
\newcommand{\labels}{\ensuremath{\mathbb{L}}}

\newcommand{\alabelset}{\ensuremath{\mathsf{L}}}

\newcommand{\avisord}{\ensuremath{\mathsf{vis}}}
\newcommand{\atsord}[1]{\ensuremath{\prec_{#1}}}
\newcommand{\aseqord}{\ensuremath{\mathsf{seq}}}
\newcommand{\alinord}{\ensuremath{\mathsf{lin}}}

\newcommand{\apre}{\ensuremath{\mathsf{pre}}}
\newcommand{\comp}{\ensuremath{\otimes}}

\newcommand{\atrace}{\ensuremath{\mathit{tr}}}
\newcommand{\hist}[1]{\ensuremath{\mathit{h}({#1)}}}
\newcommand{\Spec}{\ensuremath{\mathsf{Spec}}}
\newcommand{\updates}{\ensuremath{\mathsf{Updates}}}
\newcommand{\queries}{\ensuremath{\mathsf{Queries}}}

\newcommand{\queryupdates}{\ensuremath{\mathsf{Query\text{-}Updates}}}

\newcommand{\effector}{\ensuremath{\delta}}
\newcommand{\semop}[2][\aop]{\ensuremath{\llbracket #1
    \rrbracket}\ifthenelse{\isempty{#2}}{}{(#2)}}
\newcommand{\localstates}{\ensuremath{\mathsf{LC}}}
\newcommand{\globalstates}{\ensuremath{\mathsf{GC}}}
\newcommand{\gstates}{\ensuremath{\mathsf{G}}}

\newcommand{\labeldom}[1]{\ensuremath{\mathsf{labels}(#1)}}
\newcommand{\downstreams}{\ensuremath{\mathsf{DS}}}
\newcommand{\msgs}{\ensuremath{\mathsf{Ms}}}

\newcommand{\specarrow}[1]{\xhookrightarrow{#1}}

\newcommand{\ats}{\ensuremath{ts}} 
 
\newcommand{\atsource}{\ensuremath{\theta}} 
\newcommand{\aglobalstate}{\ensuremath{\mathtt{gc}}}

\newcommand{\crdtimp}{crdt-implementation}

\newcommand{\gconfres}{\ensuremath{\mathsf{gcr}}} 
\newcommand{\igconfres}{\ensuremath{\mathsf{igcr}}}

\begin{document}

\title[]{\CRDTLin{}}

\author{Constantin Enea}
\affiliation{
  \institution{University Paris Diderot}            
  \country{France}                    
}
\email{cenea@irif.fr}          

\author{Suha Orhun Mutluergil}
\affiliation{
  \institution{University Paris Diderot}            
  \country{France}                    
}
\email{mutluergil@irif.fr}

\author{Gustavo Petri}
\affiliation{
  \institution{ARM Research}           
  \country{United Kingdom}                   
}
\email{gustavo.petri@arm.com}         

\author{Chao Wang}
\affiliation{
  \institution{University Paris Diderot}           
  \country{France}                   
}
\email{wangch@irif.fr}         

\begin{abstract}

Geo-distributed systems often replicate data at multiple locations to achieve availability and performance despite network partitions. These systems must accept updates at any replica and propagate these updates asynchronously to every other replica. Conflict-Free Replicated Data Types (CRDTs) provide a principled approach to the problem of ensuring that replicas are eventually consistent despite the asynchronous delivery of updates.

We address the problem of specifying and verifying CRDTs, introducing a new correctness criterion called Replication-Aware Linearizability. This criterion is inspired by linearizability, the de-facto correctness criterion for (shared-memory) concurrent data structures. We argue that this criterion is both simple to understand, and it fits most known implementations of CRDTs. We provide a proof methodology to show that a CRDT satisfies replication-aware linearizability which we apply on a wide range of implementations. Finally, we show that our criterion can be leveraged to reason modularly about the composition of CRDTs.
\end{abstract}

\maketitle

\section{Introduction}
\label{sec:introduction}

Conflict-Free Replicated Data Types (CRDTs)~\cite{ShapiroPBZ11} have
recently been proposed to address the problem of availability of a
distributed application under network partitions.
CRDTs represent a methodological attempt to alleviate the problem of
retaining some data-Consistency and Availability under network
Partitions (CAP), famously known to be an impossible combination of
requirements by the CAP theorem of~\citet{GilbertL02}.
CRDTs are data types designed to favor availability over consistency
by replicating the type instances across multiple nodes of a
network, and allowing them to temporarily have different
views.
However, CRDTs guarantee that the different states of the 
nodes will \emph{eventually} converge to a state common to all
nodes~\cite{ShapiroPBZ11,Burckhardt14}.
This \emph{convergence property} is intrinsic to the data
type design and in general no synchronization is needed, hence
achieving availability. 

\noindent
{\bf Availability vs. Consistency.}
To illustrate the problem we consider the implementation of a
list-like CRDT object, the Replicated Growable Array (RGA) -- due
to~\citet{RohJKL11}\footnote{We use a variation of code extracted
  from~\cite{AttiyaBGMYZ16}.} --, used for text-editing applications. 
RGA supports three operations:
\begin{inparaenum}
\item \lstinline|addAfter(a,b)| which adds the character
  \lstinline|b| -- the concrete type is inconsequential here --
  immediately after the occurrence of the character \lstinline|a|
  assumed to be present in the list,\footnote{We assume elements are unique, implemented with timestamps.}
\item \lstinline|remove(a)| which removes \lstinline|a|
  assumed to be present in the list, and
\item \lstinline|read()| which returns the list contents.
\end{inparaenum}

To make the system available under partitions, RGA allows each of
the nodes to have a copy of the list instance.
We will call each of the nodes holding a copy a \emph{replica}.
\ifshort
\else
Then the question is, how can we maintain the consistency of the
different copies of the list given that the data could be at any point
in time be modified or read by any of the replicas?
A naive approach would synchronize all the replicas on each
operation, hence maintaining coherence, but rendering the system
unavailable if any one replica goes off-line.

Instead, 
\fi
RGA allows any of the replicas to modify the \emph{local}
copy of the list immediately -- and hence return control to the client
-- and lazily propagate the updates to the other replicas.
For instance, assuming that we have an initial list containing the
sequence $\mathtt{a \cdot b \cdot e \cdot f}$~\footnote{We use
  $s_0 \cdot s_1$ to denote the composition of sequences $s_0$ and
  $s_1$.}
and two replicas, $\arep_1$ and $\arep_2$, if $\arep_1$ inserts the
letter \lstinline|c| after \lstinline|b| (calling
\lstinline|addAfter(b,c)|), while $\arep_2$ concurrently inserts the
letter \lstinline|d| after \lstinline|b| (\lstinline|addAfter(b,d)|)
the replicas will have the states $\mathtt{a \cdot b \cdot c \cdot e
  \cdot f}$ and $\mathtt{a \cdot b \cdot d \cdot e \cdot f}$
respectively.
We have solved the availability problem, but we have introduced
inconsistent states.
This problem is only exacerbated by adding more replicas.

\noindent
{\bf Convergence.}
To restore the replicas to a consistent state, CRDTs
guarantee that under conflicting operations -- that is, operations
that could lead to different states -- there is a systematic way to
\emph{detect conflicts}, and there is a strategy followed by all
replicas to \emph{deterministically resolve conflicts}.

In the case of RGA, the implementation adds metadata to each
item of the list identifying the originating replica as well as
timestamp of the operation in that replica.\footnote{We ignore here conflicts due to \lstinline|remove|. They
  are discussed in~\autoref{sec:overview}.}
This metadata is enough to detect when conflicts have occurred.
Generally there are a number of assumptions that are necessary for the
metadata to detect conflicts (for instance that timestamps increase
monotonically with time)
which we shall discuss in the following sections.
Then, for RGA it is enough to know whether two \lstinline|addAfter|
operations have conflicted by simply comparing the replica identifiers
and their timestamps.
In fact, this is a sound over-approximation of conflict since two
concurrent \lstinline|addAfter| operations have a real conflict only
if their first arguments are the same (e.g. the element \lstinline|b|
in the example aforementioned).
In such case, the strategy to resolve the conflict will always choose
to order first the character added with the highest timestamp in the
resulting list, and in the particular case where the timestamps should
be the same, an arbitrary order among replicas will be used.
In the example above, and assuming that the character
\lstinline|c| was added with timestamp $t_1$ and the character
\lstinline|d| was added with timestamp $t_2$, if $t_2 < t_1$ (for some
order $\leq$ between timestamps), the list will converge to $\mathtt{a
  \cdot b \cdot c \cdot d \cdot e\cdot f}$. 
We obtain the same result if $t_1 = t_2$ and assume that we have a replica
order $<_r$, we have $\arep_2 <_r \arep_1$.
\ifshort
\else
In any other case we obtain $\mathtt{a \cdot b \cdot d \cdot c \cdot e\cdot f}$.
\fi
Using an arbitrary order among replica identifiers is common in
CRDT implementations to break ties among elements with equal
timestamps.
We will generally assume that metadata provides a strict ordering and
ignore the details.

If the effects of all operations are
delivered to all replicas eventually, the replicas will converge to
the same state -- assuming a quiescent period of time where no new
operations are performed.
This allows to eventually recover the consistency of the data type without
giving away availability.

\noindent
{\bf Specifications.}
The simplicity of the list data type
\ifshort
\else
with the API that we have described above
\fi
allows for a somewhat simple conflict resolution
strategy.
\ifshort
\else
Any strategy ordering conflicting concurrent insertions in a
deterministic way will work.
\fi
However, this is not true for many other CRDT implementations.
It is therefore critical to provide the programmer with a clear, and
precise, specification of the allowed behaviors of the data
type under conflicts.
Unfortunately this is not an easy task.
Many times the programmer has no option but to read the implementation
to understand how the metadata is used to resolve conflicts, for
instance by reading the algorithms by~\citet{ShapiroPBZ11} (a case
where the algorithms are particularly well documented).
Recently~\citet{BurckhardtGYZ14, Burckhardt14} have
developed a formal framework where CRDTs and other weakly
consistent systems can be specified.
However, we consider that reading these specifications is far from
trivial for the average programmer, let alone writing new
specifications.
Evidently, having a formal specification is a necessary step towards
the verification of the implementations of CRDTs.

\noindent
{\bf Simpler specifications, \emph{not simplistic} specifications.}
It is important to remark at this point that while it is our goal to
make the specification of CRDTs simpler, we believe that it is impossible to make
them coincide with their sequential data type counterparts.
Most CRDTs will exhibit, due to concurrency and consistency
relaxations, behaviors that are not possible in the sequential version
of the type they represent.
A notable instance is the Multi-Valued-Register (MVR), which resolves
conflicts arising from concurrent updates to the register by storing
multiple values.
Hence, a subsequent read operation to the register might return a set
of values rather than a single value.
This is certainly a behavior that is not possible for a
``traditional'' register, and in fact, one that the programmer must be
aware of.
Our goal is to accurately specify the behaviors of the CRDT, meaning
that often times, different implementations of the same underlying
data type (say a register) will have different specifications if their
conflict resolution allows for different behaviours, for instance the
Last-Writer-Wins (LWW) and the MVR registers which will be mentioned later.

\noindent
{\bf Paper Contributions.}
Inspired by linearizability~\cite{HerlihyW90} 
we propose a \emph{new consistency criterion
  for CRDTs}, which we call \emph{\CRDTLin{}} (\CRDTLinshort{}).
\CRDTLinshort{} both simplifies CRDT specifications, and allows us to
give correctness proof strategies for these specifications.
To satisfy \CRDTLinshort{} a data type must be so that
any execution of a client interacting with an instance of the data
type
\begin{inparaenum}
\item should result in a state that can be obtained as a sequence (or
  linearization) of its updates -- where we assume that all updates
  are executed sequentially-- and
\item any operation reading the state of the data type instance should
  be justified by executing a \emph{sub-sequence} of the above
  mentioned sequence of updates.
\end{inparaenum}
For instance, for the RGA example, the state of the final
list (when all updates are delivered) should be reachable by considering a sequence where all
\lstinline|addAfter|  operations are executed sequentially.
\ifshort
\else
\footnote{We
  will come back to RGA to add \lstinline|remove|
  in \sectionautorefname~\ref{sec:overview}.}
 \fi
\ifshort
\else
This definition shares some similarities with that of~\citet{PerrinMJ14}. 
We address the main differences in \sectionautorefname~\ref{sec:rel-work}.
\fi

Equipped with this criterion we show that many existing CRDTs are
\crdtlinearizable{}.
We provide both, their 
specification, and proofs
showing that implementations respect the specification.
We provide two different proof methodologies based on the structure of
the conflict-resolution mechanism implemented by the CRDT.
We categorize CRDT implementations into classes according to their
conflict-resolution strategy.
Encouragingly, most of the CRDTs by~\citet{ShapiroPBZ11} can be proved
RA-linearizable.

Given that our criterion is inspired by linearizability, we consider
if it also preserves the same compositionality properties, i.e.
whether the composition of a set of RA-linearizable objects is also
RA-linearizable.
While we show that this is not true in general, we show that 
compositionality can be achieved when we concentrate to specific
classes of conflict resolution as described above.

Finally, we have mechanized our methodologies to prove \CRDTLinshort{}.
We use the verification tool Boogie~\cite{BarnettCDJL05}  to
encode our specifications, CRDTs, and prove the correctness of the
implementations (proof scripts are available at~\cite{boogie-proofs}).
\ifshort
\else
To the best of our knowledge, the only other works that mechanize 
correctness proofs of CRDT implementations are~\cite{GomesKMB17,ZellerBP14}, which are 
frameworks directly carrying the proofs at a semantical level in
Isabelle/HOL, and concentrating on proving Strong Eventual
Consistency (SEC).

Finally, notice that the CISE tool~\cite{NajafzadehGYFS16} does not
actually prove the data types, but rather invariants on top of them.
\fi

\vspace{-1mm}
\section{Overview}
\label{sec:overview}

\begin{figure}[t]
\begin{lstlisting}[basicstyle=\ttfamily\scriptsize,caption={Replicated Growable Array (RGA) pseudo-code.},captionpos=b,label={lst:rga}]
  payload Ti-Tree N, Set Tomb
  initial N = @|$\emptyset$|@, Tomb = @|$\emptyset$|@
  addAfter(a,b) :
    generator :
      precondition : a = @|$\circ$|@ or (a != @|$\circ$|@ and (_,_,a) @|$\in$|@ N and a @|$\not\in$|@ Tomb)
      let t@|$_{\mathtt{b}}$|@ = getTimestamp()
    effector(a, t@|$_{\mathtt{b}}$|@, b) :
      N = N @|$\cup$|@ {(a, t@|$_{\mathtt{b}}$|@, b)}
  remove(a) :
    generator :
      precondition : (_,_,a) @|$\in$|@ N and a @|$\notin$|@ Tomb and a @|$\neq \circ$|@
    effector(a) :
      Tomb = Tomb @|$\cup$|@ {a}
  read() :
    let ret-list = traverse(N, Tomb)
    return ret-list
\end{lstlisting}
\end{figure}

\begin{wrapfigure}{r}{0.20\textwidth}
\vspace{-5mm}
  \hspace{-7mm}
  \includegraphics[scale=.67]{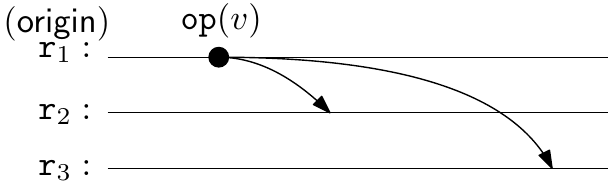}
   \vspace{-5mm}
  \caption{System Model.}
  \label{fig:sys-mod}
\vspace{-4mm}
\end{wrapfigure}
We give an informal description of our system model, and illustrate
our contribution with two compelling CRDT implementations
from~\cite{AttiyaBGMYZ16,ShapiroPBZ11}.

\ifshort
\else
We assume that the system is comprised of multiple nodes in a network.
\fi
We consider the implementation of CRDTs,
and we focus on the behaviors of an \emph{an instance} of the
data type, generically called an \emph{object}.
We assume that objects are replicated among several \emph{replicas}. 
\figureautorefname~\ref{fig:sys-mod} shows the execution of an
operation \lstinline|op(v)| evolving as follows:
\begin{inparaenum}[(i)]
\ifshort
\else
\item Firstly, a client issues a call to the object, connects
  to any one replica and performs the operation in that replica. We
  call that replica the \emph{origin}.
\fi
\item If the operation reads and updates the
  object state, the reading action is only performed at the origin.
  This part of the operation is called the \emph{generator}
  (cf.~\cite{ShapiroPBZ11}).
  Then, if the operation modifies the state --
  e.g. \lstinline|addAfter| for RGA -- an update is generated to be
  executed in every replica.
  This part of the operation shall be called the \emph{effector}.
  We assume that effectors are executed immediately at the origin.
  This is represented by the dot at the origin replica in
  \figureautorefname~\ref{fig:sys-mod}.
\item Finally, the effector is delivered to each replica, and their
  states are updated consequently, represented by the target of the
  arrows. 
\end{inparaenum}
This model corresponds to \emph{operation-based} CRDTs.
Our results also apply to \emph{state-based} CRDTs, where replicas
exchange states instead of operations (Sec.~\ref{sec:mec}).

\subsection{RGA CRDT Implementation}
\label{sec:rga-crdt-impl}

Listing~\ref{lst:rga} presents the code of RGA in a style
following that of \citet{ShapiroPBZ11} (a version of the RGA
introduced in~\cite{AttiyaBGMYZ16}). 

The keyword \lstinline|payload| declares the state 
  used to represent the object:
  \ifshort
  \else
  (akin to fields of a class
  in an object oriented language).
  \fi
  a variable \lstinline|N| of type \lstinline|Ti-Tree|, and a variable \lstinline|Tomb| of type
  \lstinline|Set|.
  \ifshort
  \else
  We then find the definitions of the operations: \lstinline|addAfter|,
  \lstinline|remove| and \lstinline|read|.
  \fi
  The effectful operations \lstinline|addAfter| and
  \lstinline|remove| have two labels marked in red:
  \lstinline|generator| and \lstinline|effector|, corresponding to the
  reading and updating part of the operations as described above.
  Notice that the effector can use as arguments values produced by the
  generator. 
  The \lstinline|precondition| annotation indicates facts
  that are assumed about the state prior to the execution.
  
Reconsidering \figureautorefname~\ref{fig:sys-mod} the
source of the arrows represents the execution of a
\lstinline|generator| jointly with the \lstinline|effector| at
replica $\arep_1$, and the target of the arrows represents the
delivery and execution of the effector at replicas $\arep_2$ and
$\arep_3$. 

  \ifshort
  \else
As it is common to many CRDT implementations, RGA replicas will use
timestamps to keep track of causality between updates, effectively
capturing when two updates are concurrent.
Moreover, they will keep the information relating the causal order
in which elements are added to the list.
Provided with this causality information -- or lack thereof--, the
timestamps will be used to resolve conflicts in a deterministic way.
\fi
Each replica maintains a \emph{Timestamp Tree} (\lstinline|Ti-Tree|) containing in every tree
node a pair with: the element added to the list (for instance the
character \lstinline|b|), and a timestamp associated to it
(\lstinline|t|$_{\mathtt{b}}$) used to resolve conflicts.
We will encode the tree as a set of triples (corresponding to
nodes) of the form \mbox{(\lstinline|a|, \lstinline|t|$_{\mathtt{b}}$, \lstinline|b|)}
representing an element \lstinline|b| in the
tree with timestamp \lstinline|t|$_{\mathtt{b}}$ and whose parent is
item \lstinline|a| also present in the tree.
The tree-ness property is ensured by construction.

The \lstinline|generator| portion of \lstinline|addAfter(a,b)| has
a precondition requiring \lstinline|a| to exist in the tree before the
insertion of \lstinline|b| 
(the data structure is initialized with a
preexisting element $\circ$).
The generator then samples a timestamp \lstinline|t|$_{\mathtt{b}}$
for \lstinline|b| which is assumed to be larger than any
timestamp presently in the \lstinline|Ti-Tree|
\lstinline|N| of the origin replica.\footnote{Also, \lstinline|t|$_{\mathtt{b}}$ cannot
be sampled by another replica (as we discussed
in \sectionautorefname~\ref{sec:introduction} this can be ensured by tagging the timestamps with replica identifiers).}
The \lstinline|effector| portion of \lstinline|addAfter(a,b)| adds the
triple \lstinline|(a,t|$_{\mathtt{b}}$\lstinline|,b)| in the
replica's own copy of \lstinline|N|.
This ensures that the tree structure is consistent with the causality
of insertions in the data structure.
A client of the object will only ever attempt to add an
element after another element which it has already seen as mandated by
the \lstinline|addAfter| API.
Hence, the parent node of any node was inserted before it, and is
causally related to it.
Similarly, nodes that are not related to each other on any path of
the tree (eg. siblings) are not causally related.
An example of such a tree is shown in the left most box
of \figureautorefname~\ref{fig:rga-trace}: 
elements \lstinline|c| and \lstinline|b| were concurrently added after
\lstinline|a|, and \lstinline|a| was added first after the
initial element $\circ$.

From a \lstinline|Ti-tree|, we can
obtain a list by traversing the tree in pre-order, with the
proviso that siblings are ordered according to their timestamps with
the \emph{highest timestamp visited first}.
The leftmost box in \figureautorefname~\ref{fig:rga-trace} shows a tree that
results in the list $\mathtt{a \cdot b \cdot c}$ assuming the
timestamp order $\mathtt{t_a < t_c < t_b}$.

\begin{figure*}[t]
  \centering
  \vspace{-2mm}
  \includegraphics[width=.68\textwidth]{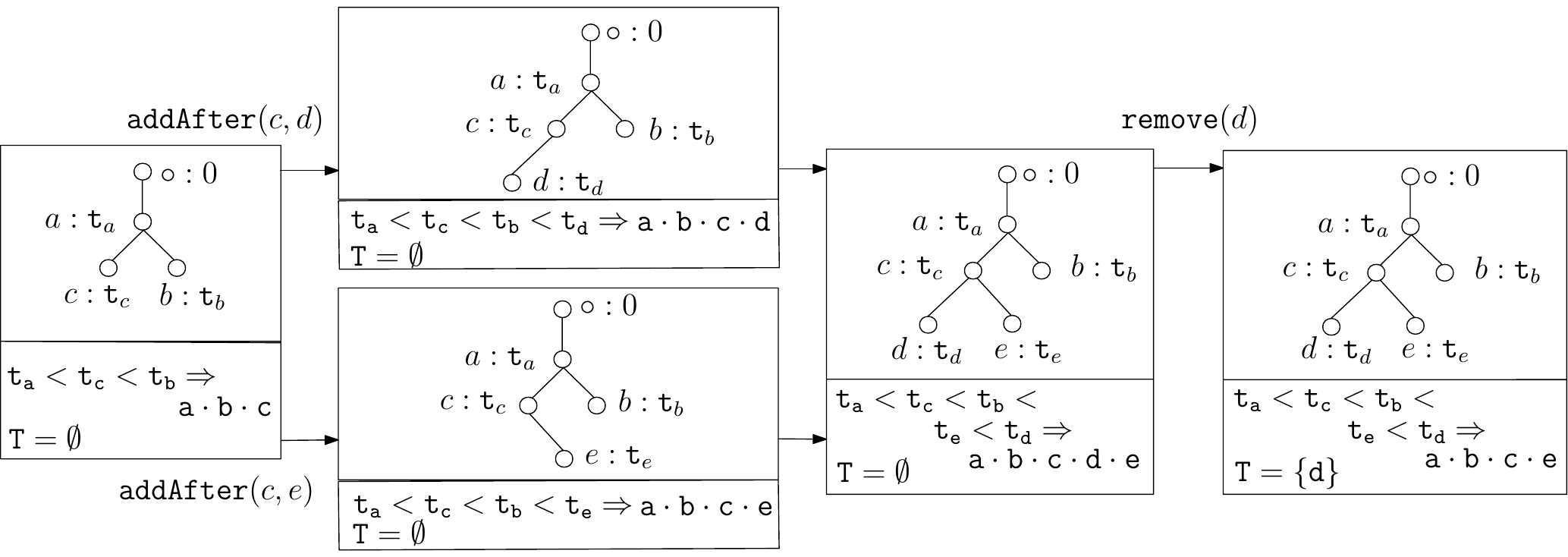}
  \vspace{-2mm}
  \caption{Example of RGA conflict resolution.}
  \label{fig:rga-trace}
  \vspace{-3mm}
\end{figure*}

\ifshort
\else
Consider now two concurrent operations: \lstinline|addAfter(c,d)| and
\lstinline|addAfter(c,e)| executing in two different replicas starting
both with the state depicted on the left of
\figureautorefname~\ref{fig:rga-trace}.
Following Listing~\ref{lst:rga} we obtain the trees in the second
column of \figureautorefname~\ref{fig:rga-trace} where we assume that
the top tree is the result of \lstinline|addAfter(c,d)| in one of the
replicas, and the bottom tree is the result of executing
\lstinline|addAfter(c,e)| in the other.
Then, the two trees result in different lists in each
replica before the operations are mutually propagated.
In the third column of \figureautorefname~\ref{fig:rga-trace} we
obtain the result of propagating the operations between the replicas
-- indeed the propagation of any of the operations to any of the
replicas yields the same result, ensured by the commutativity of
CRDTs.
It is clear now that the result of the list is $\mathtt{a \cdot b
  \cdot c \cdot d \cdot e}$.
\fi

We have so far ignored \lstinline|remove|.
Consider the case where a replica executes \lstinline|addAfter(a,b)|
on a replica while another one executes \lstinline|remove(a)|.
\ifshort
\else
If the effector of \lstinline|remove(a)| reaches every replica after
the effector of \lstinline|addAfter(a,b)| there is no problem since
the semantics is clear: the element \lstinline|a| is removed after
the element \lstinline|b| has been added.
However, if the operations reach some replica in opposite order
(recall that they are concurrent) there is a problem since
the precondition of the effector of \lstinline|addAfter(a,b)|
requires that the element \lstinline|a| be present in the
\lstinline|Ti-tree| of the replica.
\fi
To avoid this kind of conflict, rendering the operations
commutative, RGA does not really remove
elements from the \lstinline|Ti-tree|.
Instead, an additional data structure called a \emph{tombstone} is used to
keep track of elements that have been conceptually erased and should
not be considered when reading the list.
Here, the marking of tombstones is a set
\lstinline|Tomb| of elements.
The last column of \figureautorefname~\ref{fig:rga-trace} shows the result of
a \lstinline|remove| operation.

The method \lstinline|read| performs the pre-order
traversal explained before, where all elements in the tombstone
\lstinline|Tomb| are omitted.
In each of the boxes of \figureautorefname~\ref{fig:rga-trace} the list shown
represents the result of a \lstinline|read| operation in the state
depicted.

\begin{figure}
  \vspace{-1mm}
  \includegraphics[scale=.4]{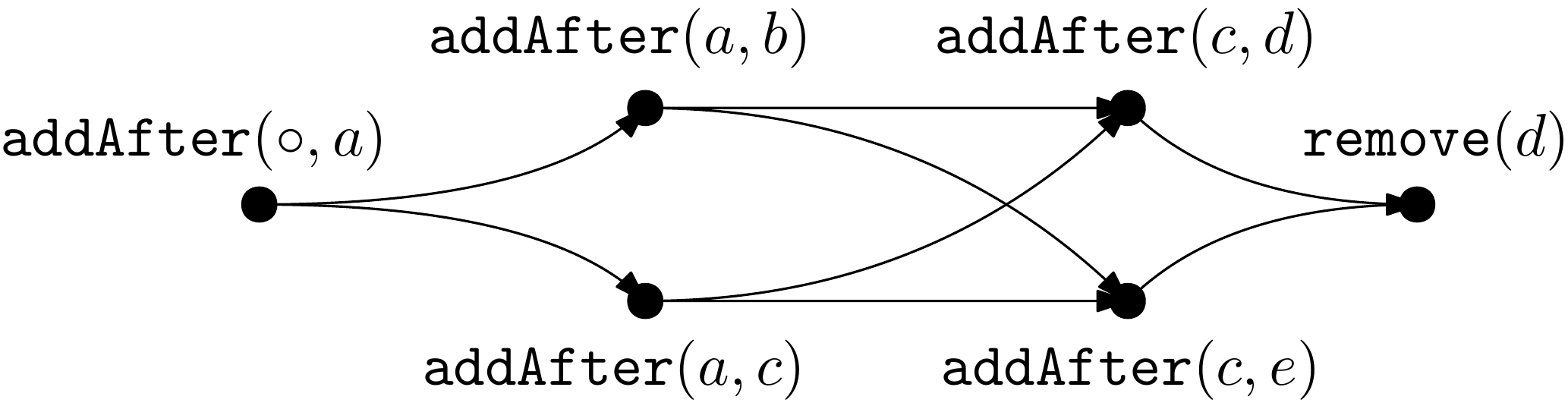}
    \vspace{-3mm}
  \caption{A history for the RGA object.}
  \label{fig:rga-history}
    \vspace{-5mm}
\end{figure}

\smallskip
\noindent
{\bf Operations, histories and linearizability.}
We consider an abstract view of executions of a
CRDT object called a \emph{history}.
Informally a history is a set of operations with a partial order
representing the ordering constraints imposed on the execution of each
operation.
We represent the execution of an operation with a label of the
form $\alabellong{\argv}{\retv}$ representing a call to method
$\amethod$ with arguments $\argv$ and returning the value $\retv$.
When the values are unimportant we shall use the meta-variable
$\alabel$ to denote a label.
The partial order mentioned above represents the \emph{visibility}
relation among operations.
We say that an operation with label
$\alabel_1$ is visible to an operation with label $\alabel_2$ if at
the time when $\alabel_2$ was executed at the origin replica, the
effects of $\alabel_1$ had been applied in the state of the replica
executing $\alabel_2$.
A history is a pair $(\alabelset, \prec)$ containing a
set of labels $\alabelset$ and a visibility relation $\prec$ between labels.
The history 
of the execution in \figureautorefname~\ref{fig:rga-trace}
is presented in \figureautorefname~\ref{fig:rga-history}.
Each node represents a label 
and arrows represent that the operation at the source of the arrow
is visible to the operation at the target.
Since we assume that visibility is transitive we ignore
redundant arrows.

A similar notion of history is used in the context of \emph{linearizability}~\cite{HerlihyW90}.
The only difference is that the order $\prec$ relates two operations 
the first of which returns before the
other one started. A history $(\alabelset, \prec)$ is called linearizable
if there exists a \emph{sequential} history $(\alabelset, \prec_{\mathsf{seq}})$ ($\prec_{\mathsf{seq}}$ is a total order), 
called \emph{linearization},
s.t. 
$(\alabelset, \prec_{\mathsf{seq}})$ is a valid execution, and
$ \prec\ \subseteq\ \prec_{\mathsf{seq}}$.

CRDTs are not linearizable since operations are propagated lazily, so two replicas can see
non-coinciding sets of operations.
We relax linearizability to adapt it to CRDTs as follows:
\begin{inparaenum}
\item we require that the sequential history be consistent with the
  visibility relation among operations instead of the returns-before
  order, and
\item operations that only read the state of the object are allowed
to see a \emph{sub-sequence} of the linearization, instead of
the whole prefix as in the case of linearizability.
\end{inparaenum}
(We will discuss an additional relaxation in
\sectionautorefname~\ref{sec:or-set-crdt}).

\noindent
{\bf Intuition of RGA \CRDTLinshort{}.}
To simplify, consider the linearization of two concurrent
operations adding after a common element: \lstinline|addAfter(a,b)|
and \lstinline|addAfter(a,c)|.
This example corresponds to the history shown in the first three nodes
of \figureautorefname~\ref{fig:rga-history} from left to right.
Because these operations are concurrent they are not related by
visibility so our criterion allows for any ordering
among them.
Let us show that these operations can always be ordered in a way
that the result of future reads will match this ordering.
From the previous explanation we know that the order between
\lstinline|b| and \lstinline|c| in the resulting list will be
determined by their corresponding timestamps
(\lstinline|t|$_{\mathtt{b}}$ and \lstinline|t|$_{\mathtt{c}}$).
Assuming that the ordering is that given in the tree of the first
column of \figureautorefname~\ref{fig:rga-trace}, we know that we can
order the operations as \lstinline|addAfter(a,c)| followed by
\lstinline|addAfter(a,b)| which when executed sequentially obviously
results in $\mathtt{a \cdot b \cdot c}$ as shown.
The timestamp metadata of RGA gives us a strategy to build the
operation sequence that corresponds to a sequential specification.
A concrete linearization of these operations is:\\[2pt]
\centerline{
  \(
  \begin{array}{l}
     \alabelshort[\mathtt{addAfter}]{\circ,a}\ \cdot\
     \alabelshort[\mathtt{addAfter}]{a,c}\ \cdot\
     \alabelshort[\mathtt{addAfter}]{a,b}
  \end{array}
  \)
}\\

\vspace{-4mm}
Unfortunately this simple linearization strategy is not always applicable.
Consider now a similar case where after issuing the
$\mathtt{addAfter}$ operations the replicas attempt to immediately
read the state.
As explained in \figureautorefname~\ref{fig:rga-trace}, a possible behavior is
that the first replica returns $\mathtt{\circ \cdot a \cdot
  b}$ while the second returns $\mathtt{\circ \cdot a \cdot
  c}$.
If we consider the linearization given above, the result $\mathtt{\circ \cdot a \cdot b}$ is not
possible, since $\mathtt{c}$ was added before $\mathtt{b}$ was added.
\ifshort
\else
The problem here is that the replica executing this read has not yet
seen the effect of $\mathtt{addAfter(a, c)}$.
\fi
To overcome this problem we allow methods that read the state to see a
\emph{sub-sequence} of the global linearization.
Thus, we can consider the sequence\\[3pt]
\centerline{
  \(
  \begin{array}{l}
       \alabelshort[\mathtt{addAfter}]{\circ,a}\ \cdot\
     {\color{red} \alabelshort[\mathtt{addAfter}]{a,c}}\ \cdot\
     \alabellong[\mathtt{read}]{}{(\circ \cdot\ a \cdot c)}\ \cdot\ \\
     \hspace{3cm}\alabelshort[\mathtt{addAfter}]{a,b}\ \cdot\ 
     \alabellong[\mathtt{read}]{}{(\circ \cdot\ a \cdot b)}
  \end{array}
  \)
}\\[2pt]
where the last $\mathtt{read}$ ignores the red label
{\color{red} $\mathtt{addAfter}(a, c)$}.
These are only two cases of conflicting concurrent
operations, in \sectionautorefname~\ref{sec:proofs} we show that all operations
can be ordered such that they correspond to a sequential
execution thereof.

\vspace{-1.5mm}
\subsection{OR-Set CRDT Implementation}
\label{sec:or-set-crdt}

\begin{figure}[!t]
  \centering
\begin{lstlisting}[caption={
Pseudo-code of the OR-Set CRDT.},basicstyle=\ttfamily\scriptsize,captionpos=b,label={lst:or-set}]
  payload Set S
  initial S = @|$\emptyset$|@
  add(a) :
    generator :
      let k = getUniqueIdentifier()
      return k
    effector(a, k) :
      S = S @|$\cup$|@ {(a, k)}
  remove(a) :
    generator :
      let R = {(a,k) | (a,k) @|$\in$|@ S}
      return R
    effector(R) :
      S = S @|$\setminus$|@ R
  read() :
    let A = {a : @|$\exists$|@ k. (a,k) @|$\in$|@ S}
    return A
\end{lstlisting}
\end{figure}

The Observed-Remove Set (OR-Set)~\cite{ShapiroPBZ11} 
implements a set with operations: \lstinline|add(a)|,
\lstinline|remove(a)|, \lstinline|read()|. The code of OR-Set is shown in Listing~\ref{lst:or-set}
(we assume return values for \lstinline|add(a)| and \lstinline|remove(a)| for technical reasons).

\begin{wrapfigure}{r}{.26\textwidth}
  \vspace{-1mm}
  \hspace{-7mm}
  \includegraphics[width=0.27 \textwidth]{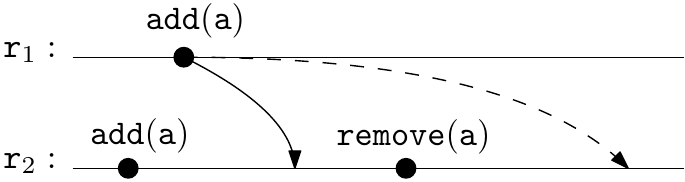}
  \vspace{-2mm}
  \caption{Interleaving-based Set.}
  \label{fig:or-set-simple}
  \vspace{-3mm}
\end{wrapfigure}
Although the meaning of these methods is self-evident from their names,
 the results of conflicting concurrent operations is not evident.
Consider for example the case where two replicas add a certain element
\lstinline|a| and then one of them removes that element.
If we consider an interleaving based execution of these operations
there are two options depending on the interleaving:
\begin{inparaenum}[i)]
\item If \lstinline|remove(a)| is the last operation
  then the expected set is empty, since the two consecutive
  \lstinline|add(a)| are idempotent, and the
  \lstinline|remove| would remove the only occurrence of
  \lstinline|a|. This interleaving is the one depicted with solid
  arrows in \figureautorefname~\ref{fig:or-set-simple}.
\item\label{or-set-ex2} On the other hand, if the operation \lstinline|add(a)| of the
  non-removing process comes last, as depicted with the dashed arrows
  in \figureautorefname~\ref{fig:or-set-simple}, the final set could contain the
  element \lstinline|a|.
\end{inparaenum}
As we have explained before, the operations can
arrive in different orders to different replicas.
To guarantee convergence, OR-Set must ensure that regardless of the
ordering, the resulting set will be the same.
To that end, OR-Set \lstinline|add| operations will tag each added
element with a unique identifier.
Then, a remove operation will only remove the element-identifier pairs
which has already seen.
For instance, in the case (\ref{or-set-ex2}) above, the remove of \lstinline|a| will
only remove the element that has been previously added by the same
replica, since this item has been observed by the \lstinline|remove|
operation -- and thus its identifier is known to it. The
concurrent \lstinline|add(a)| operation will have an identifier that
has not been observed by the \lstinline|remove| 
Therefore the item will not be removed, even if the
effectors of the two adds are performed in a replica before the effector
of the remove.

\noindent
{\bf Intuition of OR-Set Linearizability.}
It is easy to find examples where the implementation of OR-Set can
produce executions that cannot be justified by the standard definition of
linearizability (even with the relaxations discussed in \sectionautorefname~\ref{sec:rga-crdt-impl})
assuming a standard Set specification.
\figureautorefname~\ref{fig:or-set-not-lin} shows one such example.
Clearly any linearization of the visibility relation in this execution should order
the \lstinline|add| and \lstinline|remove| updates before the \lstinline|read| queries,
and the linearization of the updates should end with a
\lstinline|remove|.
Therefore, the final set returned by each of the two \lstinline|read| queries should have
at most one element (the \lstinline|read| queries see all the updates
in the execution), contrary to their return value in this execution.

This execution shows that the \lstinline|remove| operation behaves as
both a query (observing a certain number of adds of the element to be
removed) and an update (by removing said observed elements).
To cope with such cases, we will consider in our definition that
query-update operations can be split into a query part corresponding
to the generator, which only reads the state -- and hence is allowed
to see a sub-sequence of the linearization of updates -- and an
update part corresponding to the effector which will use the results
of the prior query.
For instance, \lstinline|remove| will be split into a query part \lstinline|readIds| where
only the elements visible at the time of the remove are
selected, and an update part \lstinline|remove| where only those elements selected are
erased. Any identifier not in the set returned by \lstinline|readIds| will remain in the set after 
the update part of \lstinline|remove|.
Evidently, this requires some mechanism for ``marking'' the adds that
are concerned.
We will consider that each add has a unique identifier.
\figureautorefname~\ref{fig:or-set-lk-rem} shows this rewriting.
The result of the rewriting admits a linearization consistent with the
specification of Set, as explained above.

\begin{figure*}[t]
  \begin{subfigure}{.8\linewidth}
    \centering
    \includegraphics[width=0.65 \textwidth]{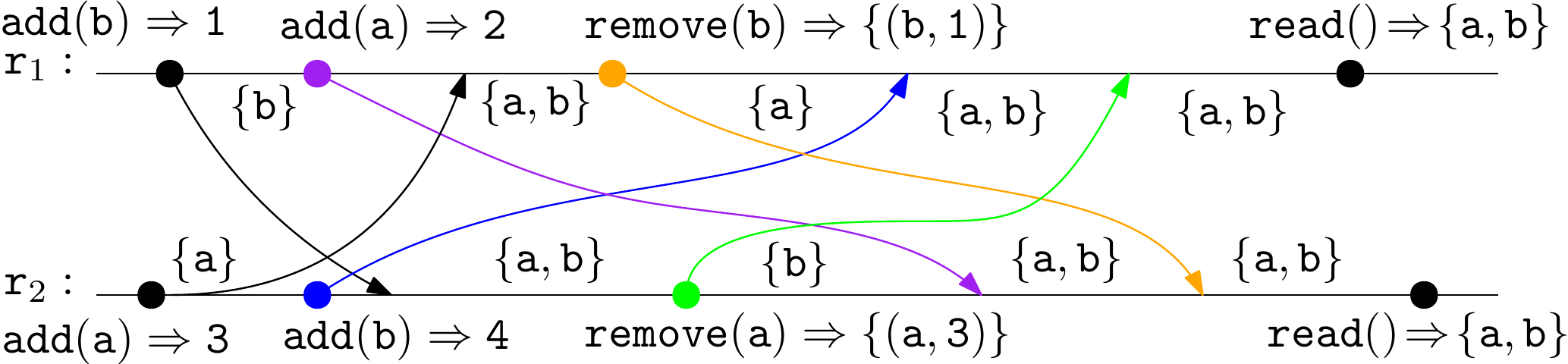}
    \caption{OR-Set non-linearizable execution. Each line represents operations issued to the same replica.}
    \label{fig:or-set-not-lin}
  \end{subfigure}

  \vspace{2mm}
  \begin{subfigure}{.8\linewidth}
    \centering
    \includegraphics[width=0.85 \textwidth]{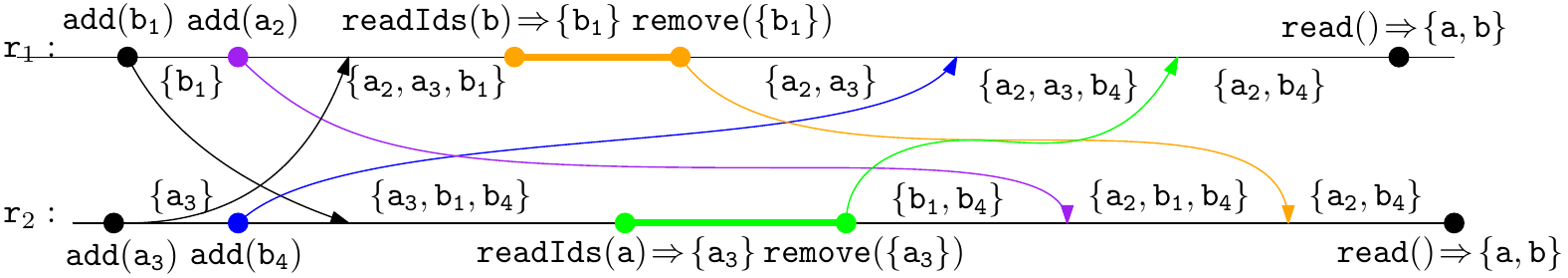}
    \caption{Label rewriting of an OR-Set execution. Pairs $(a,k)$ of an element $a$ and identifier $k$ are written as $a_k$.}
    \label{fig:or-set-lk-rem}
  \end{subfigure}
  \vspace{-2mm}
  \caption{OR-Set Linearizability vs. \CRDTLinshort{}.}
\end{figure*}

\section{\CRDTLin{}}
\label{sec:distributed-lin}

In this section we formalize the intuitions developed in \sectionautorefname~\ref{sec:overview}. We define the semantics of CRDT objects (\S~\ref{ssec:semantics}), specifications (\S~\ref{subsec:sequential specification}), and our notion of RA-linearizability (\S~\ref{subsec:definition of distributed linearizability}). For lack of space, our formalization focuses only on operation-based CRDTs. However, the notion of RA-linearizability applies to state-based CRDTs as well (see Section~\ref{sec:mec}).

\subsection{The Semantics of CRDT objects}\label{ssec:semantics}

\ifshort
\else
To formalize the semantics of CRDT objects and our correctness
criterion we use several semantic domains summarized
in \figureautorefname~\ref{fig:sem-dom}.
We let $\aobj \in \objs$ be a CRDT object in the set of objects
$\objs$.
Similarly, $\arep \in \reps$ is a replica in the set of replicas
$\reps$.
We consider a set of method names $\amethod \in \methods$, and that
each method has a number of arguments and a return value sampled from
a data domain $\datadomain$.
We assume that the domain contains a special value $\bot \in
\datadomain$ used to represent the absence of a value (for instance
the return type of procedures). 
Furthermore, we ignore typing issues which should be
addressed by an underlying programming language.
Also, some methods, e.g., the method {\tt addAfter} of the RGA object, generate timestamps from a
totally-ordered domain $\timestampdomain$. 
\fi
We will use operation labels of the form
$\alabelobjind{\argv}{\retv}{i,\ats}$ to represent the call of a
method $\amethod \in \methods$ of object $\aobj\in \objs$, with
argument $\argv \in \datadomain$, resulting in the value $\retv \in
\datadomain$, and generating the timestamp $\ats$. 
Since there might be multiple calls to the same method with the same
arguments and result, labels are tagged with a unique identifier $i$.
We may omit the object $\aobj$, the identifier $i$, the timestamp $\ats$, or the return
value $\retv$ when they are not important.
The order relation on $\timestampdomain$ is denoted by $<$.
Abusing notations, we assume that the set $\timestampdomain$
contains a distinguished minimal element $\bot$ which we shall use for
operations that do not generate a timestamp such as the method {\tt remove}
of RGA.
The timestamp $\ats$ of a label $\alabel=\alabelobjind{\argv}{\retv}{i,\ats}$
is denoted $\tsof(\alabel)$.
\ifshort
\else
We will ignore identifiers when unambiguous.
\fi
The set of all operation labels is denoted by $\labels$.

Given a CRDT object $\aobj$, its semantics is defined as a labeled transition
system (LTS) $\llbracket \aobj \rrbracket =
(\globalstates,\acts,\aglobalstate_0,\rightarrow)$, where $\globalstates$ is a set of
global configurations, $\acts$ is the set of transition labels called \emph{actions},
$\aglobalstate_0$ is the initial configuration, and
$\rightarrow\subseteq \globalstates \times \acts\times \globalstates$ is the
transition relation.
\ifshort
\else
We use the action $\src{\arep}{\alabel}$ to label the generator of $\alabel$ when executed at replica $\arep$,
and $\dwn{\arep}{\alabel}$ to label the effector of $\alabel$ when executed at $\arep$.
For readability, we use $\aglobalstate \xrightarrow{\aact} \aglobalstate'$
to denote a transition $(\aglobalstate,\aact,\aglobalstate')\in\,\rightarrow$.
\fi

Our semantics assumes the following two
properties of the propagation of effectors:
\begin{inparaenum}[(i)]
\item the effector of each operation is applied exactly once at
  each replica, and
\item if the effector of operation $\alabel_1$ is applied at the
  origin replica of $\alabel_2$ before $\alabel_2$ happens, then for every
  replica $\arep$, the effector of $\alabel_2$ will be applied only
  after the effector of $\alabel_1$ has already been applied.
\end{inparaenum}
These are commonly referred to as \emph{causal delivery}.
We assume causal delivery because our formalization focuses on operation-based CRDTs.
However, the notion of RA-linearizability and the compositionality results in Section~\ref{sec:compositionality}
apply to state-based CRDTs as well, even if the network infrastructure doesn't satisfy causal delivery.

A global configuration $(\gstates, \avisord, \downstreams)$ is a
``snapshot'' of the system that records all the operations that have
been executed.
$\gstates \in [\reps \rightarrow \localstates]$~\footnote{We use $[A\rightarrow B]$ to denote the set of total functions from
  $A$ to $B$.} stores the local
configuration of each replica ($\localstates$ denotes the set of local configurations).
A local configuration $(\alabelset, \astate)$ contains the state
$\astate$ of a replica and the set $\alabelset$ of labels of
operations that originate at this replica, or whose
effectors have been executed (or applied) at this replica.
When $\alabel\in \alabelset$, we say that $\alabel$ is \emph{visible}
to the replica or that the replica \emph{sees} $\alabel$.
The set of replica states $\astate$ is denoted by $\Sigma$.
The relation $\avisord\subseteq \powerset{\labels \times \labels}$ is
the \emph{visibility} relation between operations, i.e.,
$(\alabel_1,\alabel_2)\in \avisord$, where $\alabel_2$ is an operation
originated at a replica $\arep$, if the effector of $\alabel_1$ was
executed at $\arep$ before $\alabel_2$ was executed.
When $(\alabel_1,\alabel_2)\in \avisord$, we say that $\alabel_1$ is
\emph{visible} to $\alabel_2$, or that $\alabel_2$ \emph{sees}
$\alabel_1$.
As it will be clear from the definition of the transition relation,
$\avisord$ is a \emph{strict partial order}. 
Finally, $\downstreams\in [\labels \rightarrow \Delta]$ associates to
each operation label $\alabel \in \alabelset$ an effector
$\effector\in [\states \rightarrow \states]$, which is the replica
state transformer generated when the operation was executed at the
origin replica ($\Delta$ denotes the set of effectors).
For some fixed initial replica state $\astate_0$, the initial global configuration is defined by $\aglobalstate_0 = (\gstates_0, \emptyset, \emptyset) \in \globalstates$, where $\gstates_0$ maps each replica $\arep$ into $(\emptyset, \astate_0)$.

\begin{figure}
  \centering
  \(
  \small
  \begin{array}{cc}
  \begin{array}[t]{rcll}
    \aobj & \in  & \objs & \text{CRDT Objects} \\
    \arep & \in & \reps & \text{Replicas} \\
    \amethod & \in & \methods & \text{Methods}\\
  \end{array}
    &
  \begin{array}[t]{rcll}
    \argv, \retv & \in & \datadomain & \text{Data} \\
    \ats & \in & \timestampdomain & \text{Timestamps} \\
    \alabelset & \subseteq & \labels & \text{Label Set}\\
  \end{array}
  \end{array}
  \)
  \centerline{
    \(
  \small
    \begin{array}[t]{cccc}
    \alabel \equiv \alabelobjind{\argv}{\retv}{i,\ats} & \in & \labels & \text{Operation Label}
    \end{array}
      \)
    } 
\vspace{-6mm}
  \caption{Semantic Domains.}
  \label{fig:sem-dom}
\vspace{-4mm}
\end{figure}

\begin{figure}
  \footnotesize
\[
  \inferrule
  {\text{\sc Operation}\ \hspace{30pt} \gstates(\arep) = (\alabelset,
    \astate) \\ \atsource(\sigma,\amethod,\argv) =
    (\retv,\effector,\ats) \\  \effector(\astate) = \astate' \\
    \alabel = \alabelobjind{\argv}{\retv}{(i,\ats)} \\
    \mathit{unique}(i) \hspace{30pt} \phantom{ }\\
  \ats\neq\bot\implies (\,\forall \alabel'\in\alabelset.\ \tsof(\alabel') < \ats\,) \\
  \forall \alabel'\in \labeldom{\avisord}.\ \tsof(\alabel') \neq \ats}
  {(\gstates, \avisord, \downstreams) \xrightarrow{\src{\arep}{\alabel}} (\gstates[\arep \leftarrow (\alabelset \cup \{\alabel\}, \astate')], 
    \avisord \cup (\alabelset \times \{\alabel\}),
    \downstreams[\alabel \leftarrow \effector])}
\]
\[
  \inferrule
  {
    \text{\sc Effector}\ \hspace{60pt}
    \gstates(\arep) = (\alabelset, \astate)  \hspace{30pt} \phantom{ }\\ \alabel \in \mathsf{min}_{\avisord}(\labeldom{\avisord} \setminus \alabelset) \\
    \downstreams(\alabel)= \delta \\ \delta(\astate) = \astate'}
  {(\gstates, \avisord, \downstreams) \xrightarrow{\dwn{\arep}{\alabel}} (\gstates[\arep \leftarrow (\alabelset \cup \{\alabel\}, \astate')], \avisord, \downstreams)}
\]
  \vspace{-3mm}
\caption{
  Operational Semantics of CRDTs.
  $C[a \leftarrow b]$ denotes the in-place update of
  element $a$ of the domain of $C$ with value $b$;
  $\mathit{unique}(i)$ to ensure that $i$ is a unique identifier;
  and $\labeldom{\avisord}=\{\alabel: \exists \alabel'.\
  (\alabel,\alabel')\in \avisord \lor (\alabel',\alabel)\in \avisord\}$.
}
\label{fig:crdt-opsem}
  \vspace{-3mm}
\end{figure}

The transition relation between global configurations is defined in
\figureautorefname~\ref{fig:crdt-opsem}.
The first rule describes a replica $\arep$ in state $\astate$
executing an invocation of method $\amethod$ with argument $\argv$.
We use a function $\atsource$ to represent the behavior of the generators of all methods collectively (the code under the \lstinline|generator| labels),
i.e., $\atsource(\sigma,\amethod,\argv)$ stands for applying the generator of $\amethod$ with argument $\argv$ on the
replica state $\sigma$. 
\ifshort
\else
Therefore, this transition applies the suitable generator,
which results in a
return value $\retv$, an effector state transformer $\effector$
to be applied on all replicas, and possibly, a timestamp $\ats$.
We have $\ats=\bot$ for methods that don't generate
timestamps. 
\fi
We assume that timestamps are consistent with the visibility relation
$\avisord$, i.e., the timestamp $\ats$ generated by $\atsource$ is
strictly larger than
all the timestamps of operations visible to $\arep$, and that each timestamp can be generated only once.
\ifshort
\else
The association between the label $\alabel$ corresponding to this
invocation and the effector $\effector$ is recorded in the
$\downstreams$ component of the new global configuration.
We say that the effector $\effector$ is \emph{produced} by the operation $\alabel$.
The local configuration $(\alabelset,\sigma)$ of $\arep$ is changed by
applying the effector $\effector$ on the state $\sigma$, resulting
in a new state $\sigma'$, and adding $\alabel$ to the set of labels
$\alabelset$.
Finally, the visibility relation $\avisord$ is changed to record the
fact that the effectors of all operations in $\alabelset$ have been
applied before $\alabel$. 
\fi
This transition is labeled by $\src{\arep}{\alabel}$ where
$\alabel$ is the label of this invocation. We may ignore the index $\arep$ when it is not important.

The second rule describes a replica $\arep$ in state $\astate$
executing the effector $\effector$ that corresponds to an operation
$\alabel$ originated in a different replica.
\ifshort
\else
\footnote{This rule
  implies that we could simplify the first rule by not performing the
  effector immediately, but in general we assume no interleavings of
  operations within a single replica.}
The rule requires that $\effector$ is an effector of a label that has
not yet been applied at $\arep$ (i.e., its corresponding label is not
in the $\alabelset$ component of $\arep$'s configuration) and
moreover, that it is a minimal one with respect to the order
$\avisord$ among such effectors, i.e., there exists no
$\alabel'\not\in \alabelset$ such that $(\alabel',\alabel)\in
\avisord$.
This transition results in
modifying the state of $\arep$ to $\effector(\sigma)$ and adding
$\alabel$ to the set of operations whose effectors have been
executed by $\arep$.
Note that these transition rules preserve the fact that $\avisord$
is a strict partial order.
\fi
This transition is labeled by $\dwn{\arep}{\alabel}$.

\ifshort
\else
In what follows it will be useful to distinguish query (or pure) methods
which do not modify the state from state-modifying methods, which we
shall call updates (or effectful).
\fi
We say that a method $\amethod\in\methods$ 
is a \emph{query} if it always results (by applying the generator) in an identity effector $\effector$ (i.e.
$\delta(\sigma)=\sigma$ for all replica states
$\sigma$).
We shall call an \emph{update} any method $\amethod$ which is not a
query -- that is, whose effectors are not the identity function -- and
whose resulting effector and return value do not depend on the
initial state $\astate$ of the origin replica.
\ifshort
\else
That is, its behavior is fully determined by its arguments.
\fi
More formally, assuming a functional equivalence relation $\equiv$
between effectors that relates any two effectors that have the same
effect (modulo the values of timestamps or unique identifiers)
$\amethod$ is called an update when
$\atsource(\sigma,\amethod,\argv)|_2 \equiv
\atsource(\sigma',\amethod,\argv)|_2$, for every $\argv\in\datadomain$
and two states $\sigma,\sigma'\in\Sigma$ (for a tuple $x$, $x|_k$
denotes the projection of $x$ on the $k$-th component).
A method $\amethod$ which is not a query nor an update is called a
\emph{query-update}.
\ifshort
\else
(It generates an effector which is not the
identity function, and whose effect depends on the local state of the
replica at which the invocation of $\amethod$ originated.)
\fi
For instance, the methods \lstinline|addAfter| and
\lstinline|remove| of RGA, and \lstinline|add| of OR-Set, are updates, 
the method \lstinline|remove| of OR-Set is a query-update, and the \lstinline|read| methods of both the RGA and the OR-Set
are queries.
We denote by $\queries$, $\updates$, and $\queryupdates$, the sets of
operation labels $\alabelobjind{\argv}{\retv}{i,\ats}$ where $\amethod$
is a query, an update, or query-update respectively.
\ifshort
\else
We shall call them query, update, and query-update labels,respectively. 
\fi

An execution of the object $\aobj$ is a sequence of transitions $\aglobalstate_0\xrightarrow{\aact_0}\aglobalstate_1\xrightarrow{\aact_1}\ldots$.
A \emph{trace} $\atrace$ is the sequence of actions $\aact_0 \cdot \aact_1\ldots$ labeling the transitions of an execution.
The set of traces of an object $\aobj$ is denoted by $\traces(\aobj)$.
A \emph{history} is a pair $(\alabelset,\avisord)$ where
$\avisord\subseteq \alabelset\times\alabelset$ is an acyclic relation over the set of labels $\alabelset$.
Given an execution $e$ ending in a global configuration $(\gstates,
\avisord, \downstreams)$, the \emph{history} of $e$, denoted by $\hist{e}$, is the pair
$(\labeldom{\avisord}, \avisord)$. Note that the relation $\avisord$ is a strict partial order in this case.
\ifshort
\else
We will later allow a more general notion of history in order to deal
with object compositions (see \sectionautorefname~\ref{sec:compositionality}).
\fi
Also, the history of a trace $\atrace$, denoted by $\hist{\atrace}$,
is the history of the execution that corresponds to $\atrace$.
The set of histories $\histories(\aobj)$ of an object $\aobj$ is the
set of histories $h$ of an execution $e$ of $\aobj$. A pictorial representation
of an execution (trace) can be found in \figureautorefname~\ref{fig:or-set-not-lin} while an example of a history
can be found in \figureautorefname~\ref{fig:rga-history}.

\subsection{Sequential Specifications}
\label{subsec:sequential specification}

RA-linearizability provides an explanation for concurrent executions of CRDT objects in the form of linearizations, which can be constrained using standard sequential specifications.

\begin{definition}[Sequential Specification]
  \label{definition:sequential specification} A \emph{sequential
  specification} (specification, for short) $\Spec$ is a set of tuples $(\alabelset, \aseqord)$, where
  $\alabelset$ is a set of labels and
  $\aseqord$ is a sequence including all the labels in $\alabelset$.
\end{definition}

To describe sequential specifications in a succinct way we will
provide an operational description.
To that end, we will associate to specifications a notion of abstract
state, which we shall generally denote by $\abstate$ and its domain
shall be denoted by $\abstates$.
Then, to each valid label $\alabel$ we will associate a transition
relation $\abstate \specarrow{\alabel}  \abstate'$ which, given an
abstract state $\abstate$ and
provided that the label $\alabel$ can be applied in $\abstate$, produces a new
abstract state $\abstate'$.
In the specific case where the label $\alabel$ assumes a certain
precondition $\apre$ over the initial abstract state $\abstate$ we
will use Hoare-style preconditions and write
\(\big(\abstate\ |\ \apre(\abstate)\big) \specarrow{\alabel}
  \abstate'\).
In this way, a sequential specification is the set
of labels that are accepted by the successive application of the
transition relation starting from some given initial state
$\abstate_{0}$.

\ifshort
\else
To illustrate the definition we provide the sequential specification
of a very simple counter object, as well as the RGA and OR-Set objects
described before.
\begin{example}[Sequential Specification of a Counter]
  \label{definition:sequential specification of counter} In this case
  the state domain is $\abstates = \mathbb{Z}$, that is the state will
  be an integer, and the transitions are given as follows:\\[3pt]
\centerline{
\(
\small
  k \specarrow{\alabellong[\mathsf{inc}]{}{}} k+1\qquad\qquad
  k \specarrow{\alabellong[\mathsf{dec}]{}{}} k-1\qquad\qquad
  k \specarrow{\alabellong[\mathsf{read}]{}{k}} k
\)
}
\end{example}
\fi

\begin{example}[Sequential Specification of RGA]
  \label{definition:sequential specification of rga}
  Each abstract state $\abstate = (l,T)$ contains a sequence $l$ of
  elements of a given type and a set $T$ of elements appearing in the
  list.
  The element $l$ is the list of all input values, whether already
  removed or not; while $T$ stores the removed values and is used as
  \emph{tombstone set}.
  The sequential specification $\specRGA$ of list with add-after interface is
  defined by:
  \vspace{-1mm}
  \[\small
    \begin{array}{rcl}
      \big(\ (l_1 \cdot b \cdot l_2,T\big)\ |\ a\text{ fresh}\ \big)
      & \specarrow{\alabelshort[\mathtt{addAfter}]{b,a}}
      & (l_1 \cdot b \cdot a \cdot l_2,T)\\
      \big(\ (l,T)\ |\ b \in l\ \text{and}\ b \neq \circ\big)
      & \specarrow{\alabelshort[\mathtt{remove}]{b}}
      & (l,T \cup \{b\})\\
      (l,T)
      & \specarrow{\alabellong[\mathtt{read}]{}{(l/T)}{}}
      & (l,T)\\
 \end{array}
  \vspace{-1mm}
\]
where we denote by $l/T$ the list resulting from removing all elements
of $T$ from $l$.
The method $\alabelshort[\mathtt{addAfter}]{b,a}$ puts $a$ immediately
  after $b$ in $l$, assuming that each value is put into list at
  most once.
  \ifshort
\else
  Method $\alabelshort[\mathtt{remove}]{b}$ adds $b$ into $T$,
  hence removing $b$ from the list for subsequent calls to the
  $\mathtt{read}$ method.
  \fi
  Finally $\alabellong[\mathtt{read}]{}{s}{}$ returns the list content
  excluding any element appearing in $T$.
  Assume that the initial value of list is $(\circ,\emptyset)$, and
  $\circ$ is never removed.
  We will sometimes ignore the value $\circ$ from the return of
  $\ensuremath{\tt read}$.
\end{example}

\begin{example}[Sequential Specification of OR-Set]
\label{definition:sequential specification of or-set}
As explained in \figureautorefname~\ref{fig:or-set-lk-rem}, the fact
that the OR-Set $\mathtt{remove}$ method is a query-update induces a rewriting of the operation labels in a
history. This rewriting introduces update operations $\alabelshort[{\tt add}]{\mathit{a,id}}$, for some identifier $\mathit{id}$, instead of simply $\alabelshort[{\tt add}]{\mathit{a}}$,
and $\alabelshort[\mathtt{remove}]{S}$, for some set $S$ of pairs element-identifier, instead of $\alabelshort[{\tt remove}]{\mathit{a}}$,
and a new query operation $\alabelshort[{\tt readIds}]{\mathit{a}}$ that returns a set of pairs element-identifier. These operations are specified as follows.
The abstract state $\abstate$ is a set of tuples $(a,id)$, where $a$
is a data and $id$ is a identifier. The sequential specification
$\specOrSet$ of OR-Set is given by the transitions:\\[2pt]
\(
\small
\begin{array}{rcll}
  \abstate
  & \kern-5pt \specarrow{\alabellong[\mathtt{readIds}]{a}{S}{}} \kern-5pt
  & \abstate & 
  \kern-40pt [S = \{ (a,id)\ \vert\ (a,id) \in \abstate\}]\\
  \abstate &
             \specarrow{\alabelshort[\mathtt{remove}]{S}}
  & \abstate \setminus S \\[3pt]
  (\ \abstate\ |\ (a, \mathit{id})\ \not\in \abstate\ )
  & \specarrow{ \alabelshort[{\tt
    add}]{\mathit{a,id}} } \kern-5pt
  & \abstate \cup \{ (a,\mathit{id}) \}\\[5pt]
  \abstate
  & \specarrow{\alabellong[\mathtt{read}]{a}{ A }{}}
  & \abstate
             & 
               \kern-40pt [A = \{ a\ \vert\ \exists\ \mathit{id}, (a,\mathit{id}) \in \abstate \}]
\end{array}
\)\\[3pt]
Here $\alabellong[{\tt readIds}]{\mathit{a}}{\mathit{S}}{}$ returns the set of pairs
with data $a$, $\alabelshort[{\tt remove}]{\mathit{S}}$ removes $S$ from the
abstract state, $\alabelshort[{\tt add}]{\mathit{a,id}}$ puts $\{ (a,id) \}$
into the abstract state, and $\alabellong[{\tt read}]{}{\mathit{A}}{}$ returns
the value of the OR-Set.
\end{example}

\ifshort
\else
This definition of specification of an object will be extended to a
set of objects in \sectionautorefname~\ref{sec:compositionality}.
Another important aspect of specifications is whether they are
deterministic or not.
For instance the Wooki CRDT~\cite{DBLP:conf/wise/WeissUM07} is a
list-like object that provides a method $\mathtt{addBetween(a, b, c)}$
which inserts $\mathtt{b}$ between $\mathtt{a}$ and $\mathtt{c}$.
In contrast with RGA, where the method $\mathtt{addAfter(a, b)}$
adds the element $\mathtt{b}$ immediately after $\mathtt{a}$, in Wooki
there are many possible positions where to insert $\mathtt{b}$ if $\mathtt{a}$ and $\mathtt{c}$ are not
adjacent.
In this case, to allow for any deterministic resolution mechanism our
specifications shall be non-deterministic.
This non-determinism in the specification has to be deterministically
resolved by the implementations to ensure convergence.
\fi

\subsection{Definition of \CRDTLin{}}
\label{subsec:definition of distributed linearizability}

We now provide the definition of \crdtlin{} which characterizes histories of CRDT objects.
To simplify the presentation, we consider first the case where all the labels in the history are
either queries or updates (query-updates are considered later).
The intuition of \crdtlin{} is that there is a \emph{global} sequence
(or linearization) of the update operations in an execution which can
produce the state of \emph{each} replica when \emph{all} the updates are visible to them.
\ifshort
\else
In intermediate steps, any replica state should be the result of applying a sub-sequence of updates
of this global sequence. This is because replicas may see a subset of the updates
performed up to some moment.
Therefore, each
\fi
 query should be justified by considering the
sub-sequence of the global sequence restricted to the updates that are
visible to that query.
To be precise:
\begin{definition}
  \label{definition:ralinearizability1} A history $h =
  (\alabelset,\avisord)$ with $\alabelset\subseteq \queries\uplus\updates$ is \crdtlinearizable{} w.r.t. a
   sequential specification
  \Spec{}, if there exists a sequence
  $(\alabelset, \aseqord)$ 
\ifshort
\else
  -- where we remark that the set of labels
  are identical -- 
  \fi
  such that: 
  \begin{enumerate}[(i)]
  \item \aseqord{} is consistent with  \avisord{}, that is: $\avisord
    \cup \aseqord$ is acyclic,
  \item the projection of $\aseqord$ to \emph{updates} is
    admitted by $\Spec$, i.e.
    $\aseqord\!\downarrow_{\updates} \in \Spec$, where we denote by
    $\aseqord\downarrow_{S}$ the restriction of the order $\aseqord$ to
    the set $S$, and
  \item\label{it:query} for each query $\alabel_{\mathsf{qr}}\in \alabelset$, the sub-sequence of updates visible to $\alabel_{\mathsf{qr}}$ together with $\alabel_{\mathsf{qr}}$ is itself admitted by $\Spec$, i.e., $\aseqord\!\downarrow_{\avisord^{-1}(\alabel_{\mathsf{qr}})\cap \updates}\!\cdot\
    \alabel_{\mathsf{qr}} \in \Spec$.
\end{enumerate}
We say that $(\alabelset, \aseqord)$ is an \emph{\crdtlinearization{}} of $h$ w.r.t. $\Spec{}$.
\end{definition}

\ifshort
\else
In a nutshell, this definition requires that for a given
history, there exists a specification sequence such
that
\begin{inparaenum}[(i)]
\item the set of labels are the same and the order in the sequence is
  consistent with the visibility order of the history, that
\item when restricted to update operations -- that is all the updates --, the sequence belongs to
  the specification, and that
\item every query operation can be justified by the specification based only
  on the updates that precede it in the sequence and that are visible
  to it.
\end{inparaenum}
With this definition in mind it is not hard to check that the
\fi
sequences of operations provided in
\sectionautorefname~\ref{sec:rga-crdt-impl} and~\ref{sec:or-set-crdt}
are \crdtlinearization{}s. 
\ifshort
\else
Another example constructing a linearization step by step is shown
under the label ``timestamp-order linearizations''
of~\figureautorefname~\ref{fig:a history of RGA and its
  RA-linearization}.
\fi

We now consider the case where histories include query-updates.
In such case, we apply
Definition~\ref{definition:ralinearizability1} on a rewriting of the 
original history where each query-update is decomposed into a label
representing the generator and another label representing the effector.
\ifshort
\else
As shown in \figureautorefname~\ref{fig:or-set-lk-rem} this rewriting
may introduce new labels of operations that have been added to the
specification of the data type to provide specifications with no
query-update operations.
\fi
A mapping $\gamma:\labels\rightarrow \labels^{\leq 2}$, where $\labels^{\leq 2}$ is the set of labels and pairs of labels in $\labels$, is called a \emph{query-update rewriting}.
We assume that every query or update label is mapped by $\gamma$ to a
singleton and that the $\gamma$ image of such a label preserves its
status, i.e., $\gamma(\alabel)$ is a query, resp., update, whenever
$\alabel$ is a query, resp., update. Also, query-updates labels
$\alabel$ are mapped to pairs $\gamma(\alabel)=(\alabel_1,\alabel_2)$
where $\alabel_1$ is a query and $\alabel_2$ is an update. These
assumptions are important when applying Definition~\ref{definition:ralinearizability1} on the rewriting of a history, since this definition relies on a partitioning of the labels into queries and updates.
For a history $h=(\alabelset,\avisord)$, its $\gamma$-rewriting is a
history $\gamma(h)=(\alabelset',\avisord')$ where
\begin{itemize}
\item $\alabelset'$ is obtained by replacing each label $\alabel$ in
  $\alabelset$ with $\gamma(\alabel)$ (a label may be replaced by two
  labels),
\item whenever a (query-update) label $\alabel$ is mapped by $\gamma$
  to a pair $(\alabel_1,\alabel_2)$, we have that the query is
  ordered before the update, formally $(\alabel_1,\alabel_2)\in \avisord'$,
\item $\avisord'$ preserves the order between labels which are
  mapped to singletons, and
  for any query-update label $\alabel$ mapped to a pair
 $(\alabel_1,\alabel_2)$, the query $\alabel_1$ sees exactly the same
 set of operations as $\alabel$ and any operation which saw $\alabel$
 must see $\alabel_2$.
 Formally, whenever $(\alabel,\alabel')\in\avisord$ we have that
 $(\secondrep(\gamma(\alabel)),
 \firstrep(\gamma(\alabel')))\in\avisord'$, where for a label $\alabel$,
 $\firstrep(\gamma(\alabel))$ (resp., $\secondrep(\gamma(\alabel))$), is
 $\gamma(\alabel)$ when $\gamma(\alabel)$ is a singleton, or its first (resp.,
 second) component when $\gamma(\alabel)$ is a pair.
\end{itemize}

\begin{example}[Query-Update Rewriting of OR-Set]
\label{ex:qur-orset}
As shown in \figureautorefname~\ref{fig:or-set-lk-rem}, the query-update rewriting for OR-Set 
is defined by:  $\gamma( \alabellong[\mathtt{add}]{a}{k}) = \alabelshort[\mathtt{add}]{a,k}$, 
$\gamma( \alabellongind[{\tt read}]{}{A}{} ) = \alabellongind[{\tt read}]{}{A}{}$, and
$\gamma( \alabellong[\mathtt{remove}]{a}{R} ) = (\alabellong[\mathtt{readIds}]{a}{R}{}, \alabelshort[\mathtt{remove}]{R})$.
\end{example}

The following
extends Definition~\ref{definition:ralinearizability1} to arbitrary histories
using the rewriting defined above.

\begin{definition}[\CRDTLin{}]
  \label{definition:distributed linearizability} A history $h =
  (\alabelset,\avisord)$ is \crdtlinearizable{} w.r.t.
  \Spec{}, if there exists a query-update rewriting $\gamma$ s.t. $\gamma(h)$ is \crdtlinearizable{} w.r.t. \Spec{}.
  An RA-linearization w.r.t. $\Spec{}$ of $\gamma(\ahist)$ is called an RA-linearization w.r.t. $\Spec{}$ and $\gamma$ of $\ahist$.
\end{definition}

A set $H$ of histories is called \crdtlinearizable{} w.r.t. 
$\Spec$ when each $h\in H$ is
\crdtlinearizable{} w.r.t.
$\Spec$.
A data type implementation is \crdtlinearizable{} w.r.t.
$\Spec$ if for any object $\aobj$ of the data type, 
$\histories(\aobj)$ is linearizable w.r.t. $\Spec$.

\paragraph{Reasoning with specifications.}

To illustrate the benefit of using \CRDTLinshort{} let us consider a
simple system where two replicas execute a sequence of operations
on a shared OR-Set object:
\\[2pt]
\centerline{
  \footnotesize
\(\mathtt{add(a); rem(a); X = read()\ \|\ add(a); Y = read()}\)
}\\[2pt]
We are interested in checking that the following post-condition holds after the execution
of these operations:\\[2pt]
\centerline{
  \footnotesize
  \(\mathtt{a} \in \mathtt{X} \Rightarrow \mathtt{a} \in \mathtt{Y}\)
}\\[2pt]
Rewriting the program according to the specification of OR-Set discussed
before, we obtain the following, where the variable $\mathsf{R}$
represents the set of value timestamp pairs observed by the
$\mathsf{readIds}$ operation as defined by the rewriting:\\[2pt]
\centerline{
  \footnotesize
  \(\begin{array}{lcl}
      \left[\begin{array}{l}
              \mathtt{add(a,i_1);}\\
              \mathtt{readIds(a)\Rightarrow R;}\\
              \mathtt{rem(R);}\\
              \mathtt{X = read();}\\
              {\color{blue} \{\mathtt{a} \in \mathtt{X} \Rightarrow (a,i_2) \notin \mathtt{R} \}}
            \end{array}\right]
      & \Big\| & 
                 \left[
                 \begin{array}{l}
                   \mathtt{add(a,i_2);} \\
                   \mathtt{Y = read();} \\
                   {\color{blue} \{(a,i_2) \notin \mathtt{R} \Rightarrow \mathtt{a} \in \mathtt{Y}\}}
                 \end{array}\right]\\
      \multicolumn{3}{c}{\mathsf{Post\text{-}condition}:\color{blue} \{\mathtt{a} \in
      \mathtt{X} \Rightarrow \mathtt{a} \in \mathtt{Y}\}}\\[1mm]
    \end{array}
    \)
}
\noindent 
Since OR-Set is RA-linearizable w.r.t. the specification in
Example~\ref{definition:sequential specification of or-set} (proved in
Section~\ref{subsec:time order of execution as linearization}), the
possible values of ${\tt X}$ and ${\tt Y}$ can be computed by
enumerating their RA-linearizations.
The post-condition follows from the conjunction of the assertions in
each replica.
Let us consider the validation of the assertion of right hand side
with the following RA-linearization:
\begin{center}
  \includegraphics[width=0.45\textwidth]{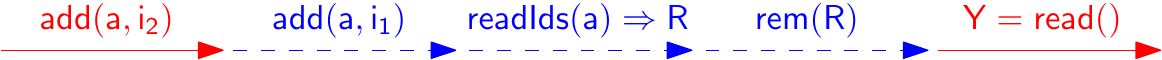}
\end{center}
We have in red color and with solid arrows the operations of the right
hand side replica, and in blue with dashed arrows the left ones.
Let us consider the sub-sequence of the linearization that is visible
to the last operation ($\ensuremath{\mathsf{Y = read()}}$).
Since the first operation ($\ensuremath{\mathsf{add(a, i_2)}}$)
is issued on the same replica, it must be visible to it.
Let us now consider different cases for the operations of the other
replica that are visible to the read:
\begin{inparaenum}[(a)]
\item if the remove operation $\ensuremath{\mathsf{rem(R)}}$ is not
  visible to it, then the assertion is trivially true, because
  $\ensuremath{\mathsf{(a, i_2)}}$ is in the resulting set according
  to the specification, and therefore the consequent of the assertion
  is valid. Assume from now on that $\ensuremath{\mathsf{rem(R)}}$ is
  visible to it, there are two cases
\item if $\ensuremath{\mathsf{(a, i_2)}}$ does not belong
  to $\ensuremath{\mathsf{R}}$ the consequent of the assertion is
  valid, since the addition of $\ensuremath{\mathsf{(a, i_2)}}$ is
  necessarily visible to the read operation, and we conclude as before, 
\item on the other hand, if $\ensuremath{(a, i_2) \in \mathsf{R}}$ we
  have that the antecedent of the implication is falsified, and
  therefore the assertion is also valid. 
\end{inparaenum}

Here we have considered only one RA-linearization, but it is not hard to
see that commuting the operations of the different replicas renders
the same argument.  
Importantly, this reasoning was done entirely at the level of the
RA-linearizations (i.e. the specification) of the data type.

For the assertion on the left hand side replica, since
visibility includes the order between operations issued on the same replica, we get that $\mathsf{add(a,i_1)}$ is ordered before $\mathsf{readIds(a)\Rightarrow R}$ in 
every RA-linearization. Since $\mathsf{add(a,i_1)}$ is also visible to $\mathsf{readIds(a)\Rightarrow R}$, we get that $(a,i_1) \in \mathsf{R}$.
Similarly, every RA-linearization will order $\mathsf{rem(R)}$ before the $\mathsf{read()}$ on the left replica, which implies that if $\mathsf{a} \in \mathsf{X}$, then
$\mathsf{(a,i_2) \notin \mathsf{R}}$. Assuming the contrary, i.e., $\mathsf{(a,i_2) \in \mathsf{R}}$, implies that $\mathsf{R=\{(a,i_1),(a,i_2)}\}$ and since $\mathsf{rem(R)}$  
is visible and linearized before $\mathsf{X = read()}$, we get that $\mathsf{a} \not\in \mathsf{X}$.

\section{Proving \CRDTLin{}}
\label{sec:proofs}

We describe a methodology for proving that CRDT objects are
\crdtlinearizable{} which relies on two properties: 
(1) the effectors of any two concurrent operations (i.e., not visible to each other) commute, 
which is inherent to CRDT objects, and 
(2) the existence of a refinement mapping~\cite{AbadiL91,DBLP:journals/iandc/LynchV95} 
showing that each effector produced by an operation
  $\alabel$, respectively each query $\alabel$, is simulated by the
  execution of $\alabel$ (or its counterpart through a query-update rewriting $\gamma$) 
  in the specification $\Spec$. This methodology is used in two forms depending on how the linearization
  is defined along an execution, which may affect the precise definition of the refinement mapping. 
  \ifshort
  \else
  We illustrate these two variations using OR-Set and RGA as examples.
  \fi

\subsection{Execution-Order Linearizations}
\label{subsec:time order of execution as linearization}

We first consider the case of CRDT objects, e.g., OR-Set, for which the order in which operations are executed at the origin replica 
defines a valid RA-linearization. We say that such objects admit \emph{execution-order linearizations}. We start by formalizing the two properties we use to prove RA-linearizability.

Given a history $\ahist{}=(\alabelset,\avisord)$, we say that two operations $\alabel_1$ and $\alabel_2$ are \emph{concurrent}, denoted $\alabel_1 \bowtie_{\avisord}\alabel_2$, when $(\alabel_1,\alabel_2)\not\in\avisord$ and $(\alabel_2,\alabel_1)\not\in\avisord$. In general, CRDTs implicitly require that the effectors of concurrent operations commute:
\begin{itemize}
\item[$\mathsf{Commutativity}$:] for every trace $\atrace$ with $\hist{\atrace}=(\alabelset,\avisord)$, and every two operations $\alabel_1,\alabel_2\in \alabelset$, if $\alabel_1 \bowtie_{\avisord}\alabel_2$, then
\begin{align*}
\forall \sigma\in\Sigma.\ \delta_{\alabel_1}(\delta_{\alabel_2}(\sigma))=\delta_{\alabel_2}(\delta_{\alabel_1}(\sigma))
\end{align*}
where $\delta_{\alabel_1}$ and $\delta_{\alabel_2}$ are the effectors of $\alabel_1$ and resp., $\alabel_2$.
\end{itemize}

\begin{example}
For OR-Set, two {\tt add}, resp., {\tt remove}, effectors commute because they both add, resp., remove, element-id pairs, while an {\tt add} and a {\tt remove} effector commute when they are concurrent because the element-id pairs removed by the {\tt remove} effector are different from the pair added by the {\tt add} effector (since the {\tt add} is not visible to {\tt remove}).
\end{example}

$\mathsf{Commutativity}$ implies that for every linearization $\alinord$ of the operations in an execution, which is consistent with the visibility relation, every replica state $\sigma$ in that execution can be obtained by applying the delivered effectors in the order defined by $\alinord$ (between the operations corresponding to those effectors). Indeed, by the causal delivery assumption, the order in which effectors are applied at a given replica is also consistent with visibility. Therefore, the only differences between the order in which effectors were applied to obtain $\sigma$ in that execution and the linearization order $\alinord$ involve effectors of concurrent operations, which commute.

\begin{lemma}\label{lem:replica_states}
Let $\rho$ be an execution of an object $\aobj$ satisfying $\mathsf{Commutativity}$, $\ahist=(\alabelset,\avisord)$ the history of $\rho$, and $(\alabelset, \aseqord)$ a linearization of the operations in $\alabelset$ (possibly, rewritten using a query-update rewriting $\gamma$), consistent with $\avisord$. 
For each local configuration $(\alabelset_\arep,\astate_\arep)$ in $\rho$, 
\begin{align*}
\sigma_\arep = \delta_{\alabel_n}(\ldots(\delta_{\alabel_1}(\sigma_0))\ldots)
\end{align*}
where $\delta_\alabel$ denotes the effector of operation $\alabel$, $\sigma_0$ is the initial replica state, and $\aseqord\downarrow_{\alabelset_\arep}=\alabel_1\ldots \alabel_n$.
\end{lemma}

In order to relate the CRDT object with its specification we use refinement mappings, which are ``local'' in the sense that they characterize the evolution of a single replica in isolation. A \emph{refinement mapping} $\refmap$ associates replica states with states of the
specification, such that any update or query applied on a replica state $\sigma$ can be mimicked by the corresponding operation in the specification
starting from $\refmap(\sigma)$. Moreover, the resulting states in the two steps must be again related by $\refmap$. Formally, given a query-update rewriting $\gamma$, we define $\mathsf{Refinement}$ as the existence of a mapping $\refmap$ such that:
\begin{itemize}
\item[Simulating effectors:] For every effector $\delta$ corresponding
  to a (query-)update operation $\alabel$, and every state $\sigma \in \Sigma$, 
  \[
    \sigma'=\delta(\sigma)\implies \refmap(\sigma)\specarrow{\secondrep(\gamma(\alabel))}\refmap(\sigma')
  \]
\noindent where $\specarrow{}$ is the transition function of $\Spec$. 
\item[Simulating generators:] For every query $\amethod$, and every $\sigma \in \Sigma$, 
  \[
    \atsource(\sigma,\amethod,\argv)= (\retv,\_,\_) \implies \refmap(\sigma)\specarrow{\alabel}\refmap(\sigma)
  \]
\noindent where $\alabel=\alabellong{\argv}{\retv}$. Recall that
$\atsource(\sigma,\amethod,\argv)$ stands for applying the generator
of $\amethod$ with argument $\argv$ on the state $\sigma$.
Also, for every query-update $\amethod$, and $\sigma\in\Sigma$,
\[
  \atsource(\sigma,\amethod,\argv)= (\retv,\_,\_) \implies \refmap(\sigma)\specarrow{\firstrep(\gamma(\alabel))}\refmap(\sigma).
\]
\end{itemize}

\begin{example}
Consider the OR-Set object, its specification in Example~\ref{definition:sequential specification of or-set}, and the query-update rewriting in Example~\ref{ex:qur-orset}. Also, let $\refmap$ be a refinement mapping defined as the identity function. The effector of an $\alabellong[{\tt add}]{a}{k}$ operation, rewritten by $\gamma$ to $\alabelshort[{\tt add}]{a,k}$, and the $\alabelshort[{\tt add}]{a,k}$ operation of the specification have the same effect. Similarly, the effector of a query-update $\alabellong[{\tt remove}]{a}{R}$ operation, rewritten by $\gamma$ to $ (\alabellong[{\tt readIds}]{a}{R}{}, \alabelshort[{\tt remove}]{a,R})$, and the $\alabelshort[{\tt remove}]{a,R}$ operation of the specification have the same effect. Applying the query operation $\alabelshort[{\tt read}]{}$ on a state $\sigma$ results in the same return value ${\tt A}$ as applying the same query in the context of the specification on the state $\refmap(\sigma)=\sigma$. Finally, for the query-update $\alabellong[{\tt remove}]{a}{R}$, executing its generator in a state $\sigma$ results in the same return value ${\tt R}$ as executing the query $\alabellong[{\tt readIds}]{a}{R}{}$ introduced by the query-update rewriting in the specification state $\refmap(\sigma)=\sigma$.
\end{example}

Next, we show that any object $\aobj$ satisfying $\mathsf{Commutativity}$ and $\mathsf{Refinement}$ is RA-linearizable. Given a history $\ahist=(\alabelset,\avisord)$ of a trace $\atrace$, the \emph{execution-order linearization} of $\ahist$ is the sequence $(\gamma(\alabelset),\aseqord)$ such that $\gamma(\alabel_1)$ occurs before $\gamma(\alabel_2)$ in $\alinord$ iff $\src{}{\alabel_1}$ occurs before $\src{}{\alabel_2}$ in $\atrace$, for every two labels $\alabel_1,\alabel_2\in\alabelset$. An object $\aobj$ \emph{admits} execution-order linearizations if for any history $\ahist=(\alabelset,\avisord)$ of a trace $\atrace$, the execution-order linearization is an RA-linearization of $\ahist$ w.r.t. $\Spec$ and $\gamma$.

\begin{theorem}\label{th:exec_order_lin}
Any object that satisfies $\mathsf{Commutativity}$ and $\mathsf{Refinement}$ admits execution-order linearizations. 
\end{theorem}
\ifshort
\else
\begin{proof}(Sketch) 
Clearly, any execution-order linearization is consistent with visibility.
Then, we have to argue that queries can be explained by applying the updates visible to them in linearization order (item~(\ref{it:query}) in  Definition~\ref{definition:ralinearizability1}).
More precisely, we have to show that for each query $\alabel_{\mathsf{qr}}$, the sequence $\aseqord'\cdot \alabel_{\mathsf{qr}}$ where $\aseqord'$ is the projection of $\aseqord$ on the set of updates
visible to $\alabel_{\mathsf{qr}}$ is admitted by the specification. First, by Lemma~\ref{lem:replica_states}, the state $\sigma$ of the replica where $\alabel_{\mathsf{qr}}$ is applied is obtained by applying the effectors of the operations visible to $\alabel_{\mathsf{qr}}$ in the linearization order. Then, by $\mathsf{Refinement}$, every effector is simulated by the corresponding operation in the context of the specification. This implies that $\refmap(\sigma_0)\specarrow{\aseqord'}\refmap(\sigma)$, where $\sigma_0$ is the initial replica state. The query $\alabel_{\mathsf{qr}}$ is also simulated by the same operation in the context of the specification, which implies that $\refmap(\sigma)\specarrow{\alabel_{\mathsf{qr}}}\refmap(\sigma)$. These two facts imply that $\refmap(\sigma_0)\specarrow{\aseqord'\cdot \alabel_{\mathsf{qr}}}\refmap(\sigma)$ which means that $\aseqord'\cdot \alabel_{\mathsf{qr}}$ is admitted by the specification.

Finally, the projection of $\aseqord$ on the updates is admitted by the specification since (1) any trace $\atrace$ can be extended with a query operation $\alabel$ that sees all the (query-)updates in $\atrace$, and (2) the validity of $\alabel$ w.r.t. the specification (shown above) implies in particular that the sequence of updates is admitted by the specification.
\end{proof}

We now discuss the issue of the non-deterministic specifications of Wooki, where an ${\tt addBetween(a,b,c)}$ operation inserts the element ${\tt b}$ at a random position between ${\tt a}$ and ${\tt c}$ (when ${\tt a}$ and ${\tt c}$ are not adjacent). The specification of the query ${\tt read}$ is however deterministic returning the whole list stored in the state (excluding tombstone elements). Although the specification of these objects is non-deterministic, the proof of Lemma~\ref{lem:replica_states} and $\mathsf{Refinement}$ imply that the objects are \emph{convergent}, in the sense that any two queries seeing the same set of updates return the same value. Indeed, by Lemma~\ref{lem:replica_states}, the replica states where such queries are applied are the same (assuming that effectors are deterministic, which is the case in all the CRDTs we are aware of) and the existence of a refinement mapping implies that the specification states corresponding to these replica states are also the same. The fact that the queries are deterministic concludes the proof of convergence.
\fi

\subsection{Timestamp-Order Linearizations}
\label{subsec:time-stamp order as linearizabtion}

\begin{figure}[t]
  \centering
  \includegraphics[width=6.5cm]{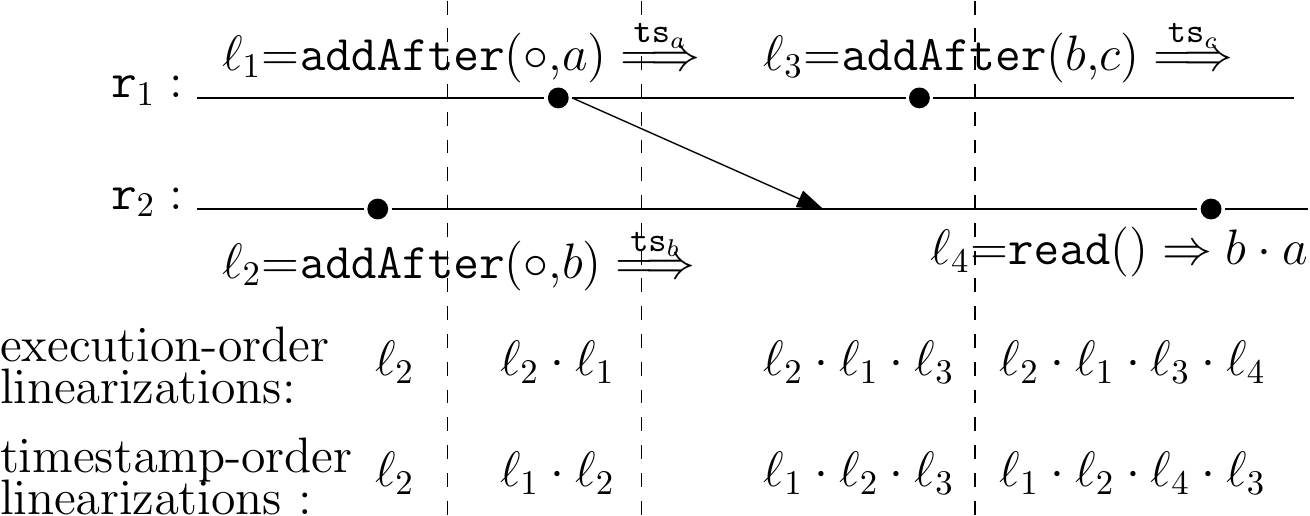}
\vspace{-2mm}
  \caption{Execution-order and timestamp-order linearizations for RGA. Here $\ats_a<\ats_b<\ats_c$.}
  \label{fig:a history of RGA and its RA-linearization}
  \vspace{-3mm}
\end{figure}

CRDT objects such as RGA in Listing~\ref{lst:rga}, that use timestamps for conflict resolution, may not admit execution-order linearizations. 
For instance, \figureautorefname~\ref{fig:a history of RGA and its RA-linearization} shows an execution of RGA where two replicas $\arep_1$ and $\arep_2$ execute two ${\tt addAfter}$ invocations, and an ${\tt addAfter}$ invocation followed by a ${\tt read}$ invocation, respectively. An execution-order linearization which by definition, is consistent with the order in which the operations are applied at the origin replica, will order $\alabelshort[{\tt addAfter}]{\circ,b}$ before $\alabelshort[{\tt addAfter}]{\circ,a}$. The result of applying these two operations in this order in the specification $\specRGA$ (defined in Example~\ref{definition:sequential specification of rga}) is the list $a\cdot b$. However, if the timestamp $\ats_a$ of $a$ is smaller than the timestamp $\ats_b$ of $b$, a ${\tt read}$ that sees these two operations will return the list $b\cdot a$, which is different than the one obtained 
\ifshort
\else
by applying the same operations in the context of $\specRGA$ in linearization order. Such a sequence is not a valid RA-linearization (w.r.t. $\specRGA$).
\fi
Therefore, we consider a variation of the proof methodology described in \sectionautorefname~\ref{subsec:time order of execution as linearization} where the linearizations are additionally \emph{consistent with the order of timestamps} generated by the operations.
\ifshort
\else
We describe this instantiation using the RGA as an example (showing that it is RA-linearizable w.r.t. $\specRGA$).
More precisely, we consider linearizations where the operations that
generate a timestamp, i.e., ${\tt addAfter}$, are ordered in the
same order as their timestamps. 
\fi
For instance, in the execution of \figureautorefname~\ref{fig:a history of RGA and its RA-linearization}, $\alabelshort[{\tt addAfter}]{\circ,a}$ will be ordered before $\alabelshort[{\tt addAfter}]{\circ,b}$ because $\ats_a$ is smaller than $\ats_b$ (irrespective of the order between the generators). 
Moreover, to extend the notion of timestamp ordering to operations
$\alabel$ that don't generate timestamps, i.e., invocations of ${\tt
  remove}$ and ${\tt read}$, we consider a ``virtual'' timestamp which
is defined as the \emph{maximal} timestamp of any operation visible to
$\alabel$ (or $\bot$ if no operation is visible to $\alabel$), and
require that the linearization is consistent with the order between
both ``real''
\ifshort
\else
\footnote{That is, timestamps generated by the operation
  itself.} 
\fi
  and ``virtual'' timestamps. For instance, the ``virtual'' timestamp of the {\tt read}
 in \figureautorefname~\ref{fig:a history of RGA and its
  RA-linearization} is $\ats_b$ because it sees $\alabelshort[{\tt addAfter}]{\circ,a}$ and
$\alabelshort[{\tt addAfter}]{\circ,b}$. Then, a valid RA-linearization will order the
{\tt read} operation before the other $\alabelshort[{\tt
  addAfter}]{b,c}$ operation, since the timestamp $\ats_c$ of the
latter is bigger than the ``virtual'' timestamp $\ats_b$ of the {\tt
  read}. The operations that have the same timestamp (which is possible due to ``virtual'' timestamps)
\ifshort
\else
  \footnote{Among
  operations that have the same timestamp $\ats$, there is exactly one
  operation generating $\ats$, the rest of the operations have $\ats$ as a ``virtual'' timestamp (i.e., they don't generate timestamps and the maximal timestamp they see is $\ats$).} 
\fi
  are ordered as they execute at the origin replica. For instance, the {\tt read} with ``virtual'' timestamp $\ats_b$ is ordered after $\alabelshort[{\tt addAfter}]{\circ,b}$ that has the same timestamp $\ats_b$ since it executes later at the origin replica.
  
Formally, for a history $\ahist=(\alabelset,\avisord)$, we define the timestamp $\tsof_\ahist(\alabel)$ of a label $\alabel$ in the context of the history $h$ to be $\tsof_\ahist(\alabel)=\tsof(\alabel)$ if $\tsof(\alabel)\neq\bot$ and $\tsof_\ahist(\alabel)=\mathsf{max}\, \{\tsof(\alabel'):(\alabel',\alabel)\in \avisord \}$, otherwise.
Given a history $\ahist=(\alabelset,\avisord)$ of a trace $\atrace$, the \emph{timestamp-order linearization} of $\ahist$ is the sequence $(\alabelset,\aseqord)$ such that $\gamma(\alabel_1)$ occurs before $\alabel_2$ in $\alinord$ iff
$\tsof_\ahist(\alabel_1) < \tsof_\ahist(\alabel_2)$ or $\src{}{\alabel_1}$ occurs before $\src{}{\alabel_2}$ in $\atrace$, for every two labels $\alabel_1,\alabel_2\in\alabelset$.
An object $\aobj$ \emph{admits} timestamp-order linearizations if for any history $\ahist=(\alabelset,\avisord)$ of a trace $\atrace$, the timestamp-order linearization is an RA-linearization of $\ahist$ w.r.t. $\Spec$.~\footnote{For simplicity, we ignore query-update rewritings. The CRDTs with timestamp-order linearizations we investigated don't require such rewritings.}

Proving admittance of timestamp-order linearizations relies on $\mathsf{Commutativity}$ and a slight variation of $\mathsf{Refinement}$ where intuitively, an effector generating a timestamp $\ats$ has to be simulated by a specification operation only when it is applied on a state $\sigma$ that doesn't ``store'' a greater timestamp than $\ats$ (other effectors are treated as before). Formally, the set $\tsof(\sigma)$ of timestamps in a state $\sigma$ contains all the timestamps $\ats$ generated by effectors applied to obtain $\sigma$. For RGA, the set of timestamps in a state $\sigma$ is the set of all timestamps stored in its timestamp tree. We define $\mathsf{Refinement}_{\tsof{}}$ by modifying the ``Simulating effectors'' part of $\mathsf{Refinement}$ as follows:
\begin{itemize}
\item[Simulating effectors:] For every effector $\delta$ of an operation $\alabel$, 
\begin{align*}
\forall \sigma\in\Sigma.\ \tsof(\alabel) \not< \tsof(\sigma)\land \sigma'=\delta(\sigma)\implies \refmap(\sigma)\specarrow{\alabel}\refmap(\sigma')
\end{align*}
\ifshort
\else
where $\tsof(\alabel) \not< \tsof(\sigma)$ means that $\tsof(\alabel)$ is not smaller than any timestamp in $\tsof(\sigma)$.
\fi
\end{itemize}

\begin{example}
Let us consider the RGA object, its specification in Example~\ref{definition:sequential specification of rga}, and 
a refinement mapping $\refmap$ which relates a replica state ${\tt (N,Tomb)}$ with a specification state $(l,T)$ where the sequence $l$ is given by the function ${\tt traverse}$ in ${\tt read}$ queries when ignoring tombstones, i.e., $l={\tt traverse(N, \emptyset)}$, and $T={\tt Tomb}$. It is obvious that ${\tt remove}$ effectors and ${\tt read}$ queries are simulated by the corresponding specification operations. Effectors of $\alabellongind[{\tt addAfter}]{a,b}{}{\ats_b}{}$ operations are simulated by the specification operation $\alabelshort[{\tt addAfter}]{a,b}$ only when $\ats_{\tt b}$ is greater than all the timestamps stored in the replica state where it applies.  Thus, let $({\tt N},{\tt Tomb})$ be a replica state such that $\ats < \ats_{\tt b}$ for every $\ats$ with $(\_,\ats,\_)\in {\tt N}$. The result of applying the effector $\effector$ of $\alabellongind[{\tt addAfter}]{a,b}{}{\ats_b}{}$ is to add ${\tt b}$ as a child of ${\tt a}$. Then, applying ${\tt traverse}$ on the new state will result in a sequence where ${\tt b}$ is placed just after ${\tt a}$ because it has the highest timestamp among the children of ${\tt a}$. 
\ifshort
\else
(and all the nodes in the tree ${\tt N}$). 
\fi
This corresponds exactly to the sequence obtained by applying the operation $\alabelshort[{\tt addAfter}]{a,b}$ in the context of the specification.
\end{example}

The proof of an object $\aobj$ admitting timestamp-order linearizations if it satisfies $\mathsf{Commutativity}$ and $\mathsf{Refinement}_{\tsof{}}$  is similar to the one of Theorem~\ref{th:exec_order_lin}. 
\ifshort
\else
Intuitively, although $\mathsf{Refinement}_{\tsof{}}$ is weaker than $\mathsf{Refinement}$, the fact that the timestamp-order linearizations are consistent with the order between the timestamps generated by the operations, allows to show that any sequence of effectors consistent with such a linearization can be simulated by a sequence of specification operations.
\fi

\begin{theorem}\label{th:timestamp_order_lin}
Any object that satisfies $\mathsf{Commutativity}$ and $\mathsf{Refinement}_{\tsof{}}$ admits timestamp-order linearizations. 
\end{theorem}

We remark that the API of a CRDT can impact on whether it is
RA-linearizable. For instance, a slight variation of the RGA in
Listing~\ref{lst:rga} with the same state, but with an interface with
a method \lstinline|addAt(a,k)| to insert an element \lstinline|a| at
an index \lstinline|k|, introduced in~\cite{AttiyaBGMYZ16}, would not
be RA-linearizable w.r.t. an appropriate sequential specification (see~\cite{arxiv}).

\section{Compositionality of RA-Linearizability}
\label{sec:compositionality}

\begin{figure}[t]
  \centering
  \includegraphics[scale=.7]{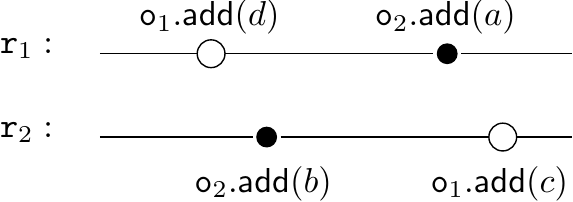}
  \vspace{-2mm}
\caption{A history of two OR-Sets. 
Each operation is visible only at the origin, so visibility is given by the horizontal lines.}
  \label{fig:negative_composition}
\vspace{-3mm}
\end{figure}

We investigate the issue of whether the composition of a set of objects satisfying RA-linearizability is also \crdtlinearizable{}. While this is not true in general, we show that the composition of objects that admit execution-order or timestamp-order linearizations is RA-linearizable under the assumption that they share the same timestamp generator.

\subsection{Object Compositions and RA-Linearizability}\label{ssec:comp_intro}

Given two objects $\aobj_1$ and $\aobj_2$, the semantics of their composition $\aobj_1\comp \aobj_2$ is the standard product of the LTSs corresponding to $\aobj_1$ and $\aobj_2$, respectively. 
\ifshort
\else
Formally, given $\llbracket \aobj_1 \rrbracket =(\globalstates_1,\acts_1,\aglobalstate_0^1,\rightarrow_1)$ and $\llbracket \aobj_2 \rrbracket =(\globalstates_2,\acts_2,\aglobalstate_0^2,\rightarrow_2)$, we define $\llbracket \aobj_1\comp \aobj_2 \rrbracket =(\globalstates_1\times \globalstates_2,\acts_1\cup \acts_2,(\aglobalstate_0^1,\aglobalstate_0^2,\emptyset),\rightarrow_{1,2})$ where\\[2pt]
\centerline{
\(
\begin{array}{rcl}
  \rightarrow_{1,2} & = & \{ ((\aglobalstate_1,\aglobalstate_2),\aact_1,(\aglobalstate_1',\aglobalstate_2)): (\aglobalstate_1,\aact_1,\aglobalstate_1')\in\rightarrow_1\}\\
                    & \cup & \{ ((\aglobalstate_1,\aglobalstate_2),\aact_2,(\aglobalstate_1,\aglobalstate_2')): (\aglobalstate_2,\aact_2,\aglobalstate_2')\in\rightarrow_2\}
\end{array}
\)
}\\[2pt]
\fi
 The history of a trace $\atrace$ of $\aobj_1\comp \aobj_2$ records a ``global'' visibility relation between the operations in the trace, i.e., which operations of $\aobj_1$ or $\aobj_2$ are visible when issuing an operation of $\aobj_1$, and similarly, for operations of $\aobj_2$. Formally, $\hist{\atrace}=(\alabelset, \avisord)$ where $\alabelset$ is the set of labels occurring in $\atrace$, and $(\alabel_1,\alabel_2)\in\avisord$ if there exists a replica $\arep$ such that $\dwn{\arep}{\alabel_1}$ occurs before $\src{\arep}{\alabel_2}$ in the trace $\atrace$. In general, $\avisord$ may not be a partial order since the causal delivery assumption holds only among operations of the same object. The set of histories $\histories(\aobj_1\comp \aobj_2)$ of the composition $\aobj_1\comp \aobj_2$ is the set of histories $h$ of a trace $\atrace$ of $\aobj_1\comp \aobj_2$.

For two specifications $\Spec_1$ and $\Spec_2$ of two objects $\aobj_1$ and $\aobj_2$, respectively, the composition $\Spec_1\comp \Spec_2$ is the set of interleavings of sequences in $\Spec_1$ and $\Spec_2$, respectively. 
\ifshort
\else
More precisely, $\Spec_1\comp \Spec_2$ is the set of sequences $(\alabelset,\aseqord)$ such that their projection on labels of $\aobj_1$, resp., $\aobj_2$, is admitted by $\Spec_1$, resp., $\Spec_2$. 
\fi
We say that the composition $\aobj_1\comp \aobj_2$ is \emph{\crdtlinearizable{}} if every history of $\aobj_1\comp \aobj_2$ is \crdtlinearizable{} w.r.t. $\Spec_1\comp \Spec_2$. The extension 
to a set of objects is defined as usual.

Linearizability~\cite{HerlihyW90} ensures that 
\ifshort
\else
the composition of a set of linearizable objects is also linearizable. More precisely, it ensures that 
\fi
for every history, any per-object linearizations, concerning the operations of a single object, can be combined into a global linearization, concerning all the operations in the history. 
\ifshort
\else
By combining linearizations, we mean constructing a global linearization whose projections on the operations of a single object are exactly the per-object linearizations considered in the beginning.
\fi
However, this is not true for our notion of RA-linearizability. A counterexample is given in \figureautorefname~\ref{fig:negative_composition}. 
\ifshort
\else
The operations of $\aobj_1$ are represented using blank circles and the operations of $\aobj_2$ using filled circles. 
\fi
The operations of $\aobj_1$ can be linearized to $\aobj_1.\alabelshort[{\tt add}]{c}\cdot \aobj_1.\alabelshort[{\tt add}]{d}$ 
\ifshort
\else
(this is a valid \crdtlinearization{} since any sequence of ${\tt add}$ operations is admitted by $\specOrSet$) 
\fi
while the operations of $\aobj_2$ can be linearized to $\aobj_2.\alabelshort[{\tt add}]{a}\cdot \aobj_2.\alabelshort[{\tt add}]{b}$. There is no \crdtlinearization{} of this history whose projections on each of the two objects correspond to these per-object linearizations. 
\ifshort
\else
Trying to construct a linearization where $\aobj_2.\alabelshort[{\tt add}]{a}$ occurs before $\aobj_2.\alabelshort[{\tt add}]{b}$ will imply that $\aobj_1.\alabelshort[{\tt add}]{d}$ must occur before $\aobj_1.\alabelshort[{\tt add}]{c}$ (since it must be consistent with the visibility relation), which contradicts the linearization of $\aobj_1$, and similarly for the other case, when trying to construct a global linearization consistent with the linearization of $\aobj_2$'s operations. 
A reader knowledgeable of the literature on linearizability may notice that this discrepancy between standard linearizability and RA-linearizability comes from the fact that the partial order defining a history in the case of standard linearizability is actually an interval order\footnote{A partial order $R$ is an interval-order if $\{(a,b), (c,d)\} \subseteq R$ implies that $(a,d) \in R$ or $(c,b) \in R$, for every $a,b,c,d$.}, while in the case of RA-linearizability it is an arbitrary partial order.
\fi

\subsection{Composition: Execution-Order Linearizability  }

Although not all per-object \crdtlinearization{s} can be combined into global \crdtlinearization{s}, this may still be true in some cases. For the history in \autoref{fig:negative_composition}, the operations of $\aobj_1$ can also be linearized to $\aobj_1.\alabelshort[{\tt add}]{d}\cdot \aobj_1.\alabelshort[{\tt add}]{c}$ which enables a global linearization $\aobj_1.\alabelshort[{\tt add}]{d}\cdot \aobj_2.\alabelshort[{\tt add}]{a}\cdot \aobj_2.\alabelshort[{\tt add}]{b}\cdot \aobj_1.\alabelshort[{\tt add}]{c}$ whose projection on each object is consistent with the per-object linearization (we take the same linearization 
\ifshort
\else
$\aobj_2.\alabelshort[{\tt add}]{a}\cdot \aobj_2.\alabelshort[{\tt add}]{b}$ 
\fi
for $\aobj_2$).

\ifshort
\else
The following theorem shows
\fi
that in the case of \crdtlinearizable{} objects that admit execution-order linearizations, there always exist per-object \crdtlinearization{s} that can be combined into global \crdtlinearization{s}, hence their composition is \crdtlinearizable{} and moreover, it also admits execution-order linearizations. 
\ifshort
\else
A first preliminary result states that the order in which concurrent operations are executed at the origin replica can be permuted arbitrarily while still leading to a valid trace. More precisely, for every linearization $\alinord$ of the visibility in the history of a trace, there exists another valid trace where operations are executed at the origin replica in the order defined by $\alinord$.

\begin{lemma}\label{lem:trace_closure}
Let $\atrace$ be a trace of an object $\aobj$ and $\hist{\atrace}=(\alabelset,\avisord)$. Then, for every sequence $(\alabelset,\alinord)$ which is consistent with $\avisord$ (i.e., $\avisord
    \cup \alinord$ is acyclic), there exists a trace $\atrace'$ of $\aobj$ such that $\src{}{\alabel_1}$ occurs before $\src{}{\alabel_2}$ in $\atrace'$ iff $\alabel_1$ occurs before $\alabel_2$ in $\alinord$. Moreover, $\atrace'$ has the same history as $\atrace$.
\end{lemma}
\vspace{-10pt}
\begin{proof}(Sketch)
We define a \emph{dependency} relation $\circledcirc$ between actions as follows: $\src{\arep_1}{\alabel_1}\circledcirc \dwn{\arep_2}{\alabel_2}$ iff $\arep_1=\arep_2$ or $\alabel_1=\alabel_2$, $\src{\arep_1}{\alabel_1}\circledcirc \src{\arep_2}{\alabel_2}$ iff $\arep_1=\arep_2$, $\dwn{\arep_1}{\alabel_1}\circledcirc \src{\arep_2}{\alabel_2}$ iff $\arep_1=\arep_2$, and $\dwn{\arep_1}{\alabel_1}\circledcirc \dwn{\arep_2}{\alabel_2}$ iff $\arep_1=\arep_2$. Given a trace $\atrace=\atrace_1\cdot \aact_1\cdot\aact_2\cdot\atrace_2$, we say that a trace $\atrace'=\atrace_1\cdot \aact_2\cdot\aact_1\cdot\atrace_2$ is derived from $\atrace$ by a \emph{$\circledcirc$-valid swap} iff $\aact_1$ and $\aact_2$ are \emph{not} related by $\circledcirc$. Using standard reasoning about traces of a concurrent system, it can be shown that if $\atrace$ is a trace of $\aobj$, then any trace $\atrace'$ derived through a sequence of $\circledcirc$-valid swaps is also a trace of $\aobj$. Moreover, $\atrace'$ has the same history as $\atrace$. Then, by the definition of the visibility relation $\avisord$ in the history of a trace and the causal delivery assumption, it can be shown that given a trace $\atrace$ of $\aobj$ containing two actions $\src{}{\alabel_1}$ and $\src{}{\alabel_2}$ such that $\alabel_1$ and $\alabel_2$ are concurrent and $\src{}{\alabel_1}$ occurs before $\src{}{\alabel_2}$, there exists another trace $\atrace'$ of $\aobj$ that is defined through a sequence of $\circledcirc$-valid swaps, where $\src{}{\alabel_2}$ occurs before $\src{}{\alabel_1}$. Thus, the order of concurrent generator actions can be permuted arbitrarily, concluding the lemma.
\end{proof}

Lemma~\ref{lem:trace_closure} implies that given an \crdtlinearizable{} object $\aobj$, if it admits execution-order linearizations, then every linearization of a history of $\aobj$ consistent with visibility is a valid \crdtlinearization{}.

\begin{lemma}\label{lem:t0_lin}
Let $\aobj$ be an object which is \crdtlinearizable{} w.r.t. a specification $\Spec$ and admits execution-order linearizations. Then, for every history $\ahist=(\alabelset,\avisord)$ of $\aobj$ and every sequence $(\alabelset,\alinord)$ which is consistent with $\avisord$, we have that $(\alabelset,\alinord)$ is an \crdtlinearization{} of $\ahist$ w.r.t. $\Spec$.
\end{lemma}
\vspace{-10pt}
\begin{proof}(Skech)
Let $(\alabelset,\alinord)$ be a sequence consistent with the visibility relation of $\ahist$. By Lemma~\ref{lem:trace_closure}, there exists a trace $\atrace$ of $\aobj$ such that  $\src{}{\alabel_1}$ occurs before $\src{}{\alabel_2}$ in $\atrace$ iff $\alabel_1$ occurs before $\alabel_2$ in $\alinord$. The definition of execution-order linearizations implies that $(\alabelset,\alinord)$ is an RA-linearization of $\hist{\atrace}=\ahist$.
\end{proof}

Lemma~\ref{lem:t0_lin} is the essential ingredient for proving that the composition of a set of \crdtlinearizable{} objects that admit execution-order linearizations is also \crdtlinearizable{}. Intuitively, given a history $\ahist$ with multiple such objects, any linearization of $\ahist$ is a valid \crdtlinearization{} since each of its projections on the set of operations of a single object is a linearization of the per-object visibility relation~\footnote{By definition, the visibility relation in the history of an object composition $\aobj_1\comp\aobj_2$ projected on operations of $\aobj_1$, resp., $\aobj_2$, is exactly the visibility relation stored in the global configurations of $\aobj_1$, resp., $\aobj_2$ (at the end of the execution).}  and thus, a valid \crdtlinearization{} of that object.
\fi

\begin{theorem}\label{th:comp_execution_order}
The composition of a set of \crdtlinearizable{} objects that admit execution-order linearizations is \crdtlinearizable{} and admits execution-order linearizations.
\end{theorem}
\vspace{-\topsep}

\subsection{Composition: Timestamp-Order Linearizability}

\begin{figure}[t]
  \centering
  \includegraphics[scale=.6]{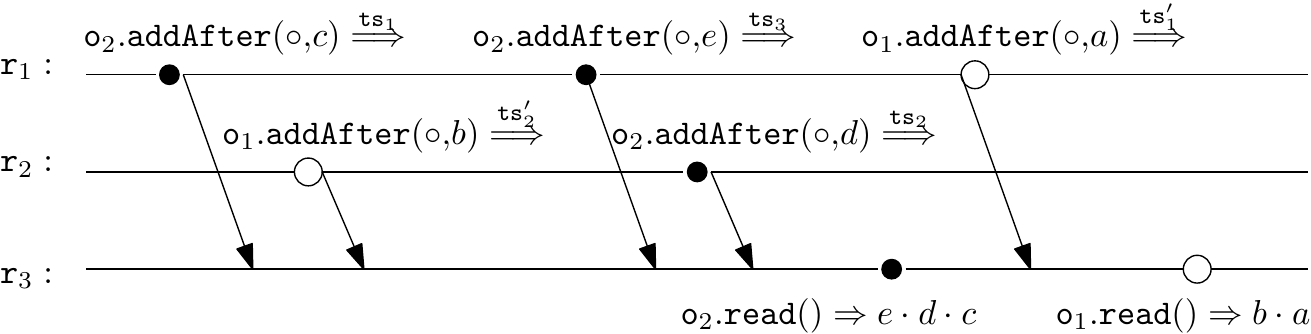}
\vspace{-2mm}
  \caption{A history in the composition $\otimes$ of two RGAs.}
  \label{fig:negative_ts_composition}
  \vspace{-2mm}
\end{figure}

Theorem~\ref{th:comp_execution_order} does not apply to objects that admit timestamp-order linearizations. The ``unrestricted'' object composition $\otimes$ allows different objects to generate timestamps independently, and in ``conflicting'' orders along some execution. For instance, \figureautorefname~\ref{fig:negative_ts_composition} shows a history with two RGA objects $\aobj_1$ and $\aobj_2$. We assume that $\ats_1 < \ats_2 < \ats_3$ and $\ats'_1 < \ats'_2$ (the order between other timestamps is not important). The operations of $\aobj_1$, resp., $\aobj_2$, can be linearized to\\[3pt]
\centerline{
\(
\begin{array}{l}
\bullet\ \ \aobj_1.\alabelshort[{\tt addAfter}]{\circ,a}\cdot \aobj_1.\alabelshort[{\tt addAfter}]{\circ,b}\cdot \aobj_1.\alabellongind[{\tt read}]{}{b\cdot a}{} \\
 \bullet\ \  \aobj_2.\alabelshort[{\tt addAfter}]{\circ,c}\ \cdot\ \aobj_2.\alabelshort[{\tt addAfter}]{\circ,d}\ \cdot\kern20pt\\
  \hspace{\fill}\aobj_2.\alabelshort[{\tt addAfter}]{\circ,e}\ \cdot\ \aobj_2.\alabellongind[{\tt read}]{}{e\cdot d\cdot c}{}
\end{array}
\)
}\\[3pt]
These are the only \crdtlinearization{s} possible. There is no ``global'' linearization consistent with these per-object linearizations: ordering $\alabelshort[{\tt addAfter}]{\circ,a}$ before $\alabelshort[{\tt addAfter}]{\circ,b}$ implies that $\alabelshort[{\tt addAfter}]{\circ,e}$ occurs before $\alabelshort[{\tt addAfter}]{\circ,d}$ which contradicts the second linearization above.
We solve this problem by constraining the composition operator $\comp$ such that intuitively, all objects share a common timestamp generator. 
This ensures that each new timestamp is bigger than the  timestamps used by operations delivered to a replica, independently of the object to which they pertain. For instance, the history of \figureautorefname~\ref{fig:negative_ts_composition} would not be admitted because $\ats_1'$ should be bigger than $\ats_3$ (since the operation that received $\ats_3$ from the timestamp generator originates from the same replica as the operation receiving $\ats_1'$ at a later time) and $\ats_2$ should be bigger than $\ats_2'$. These two constraints together with $\ats_1'<\ats_2'$ contradict $\ats_2 < \ats_3$.
While this requires a modification of the algorithms, where the timestamp generator is a parameter, this has no algorithmic or run-time cost, and in fact a similar idea have been suggested in the systems literature (e.g.~\cite{EnesPB17}).

\begin{figure}[t]
  \centering
  \footnotesize
\[
  \inferrule
  {\text{\sc Operation} \hspace{40pt}
    \alabel = \aobj_k.\alabellongind{\argv}{\retv}{(i,\ats)}\mbox{ with } k\in \{1,2\} \\ (\gstates_k, \avisord_k, \downstreams_k) \xrightarrow{\src{\arep}{\alabel}}_{k} (\gstates_k', \avisord_k', \downstreams_k') \hspace{30pt} \phantom{ }\\
  (\gstates'_{k'}, \avisord'_{k'}, \downstreams'_{k'}) = (\gstates_{k'}, \avisord_{k'}, \downstreams_{k'})\mbox{ for $k'\neq k$} \\
  \gstates_1(\arep) = (\alabelset_1, \astate_1) \\ \gstates_2(\arep) = (\alabelset_2, \astate_2) \\
  \ats\neq\bot\implies (\,\forall \alabel'\in\alabelset_1\cup\alabelset_2.\ \tsof(\alabel') < \ats\,) \\
  \forall \alabel'\in \labeldom{\avisord_1\cup\avisord_2}.\ \tsof(\alabel') \neq \ats
  }
  {((\gstates_1, \avisord_1, \downstreams_1),(\gstates_2, \avisord_2, \downstreams_2)) \xrightarrow{\src{\arep}{\alabel}} (\gstates_1', \avisord_1', \downstreams_1'),(\gstates_2', \avisord_2', \downstreams_2')}
\]
\vspace{-3mm}
\caption{The transition rule for generators in the object composition operator $\otimes_{\tsof}$.}
  \label{fig:comp-ts}
\vspace{-2mm}
\end{figure}

\ifshort
\else
Using again Lemma~\ref{lem:trace_closure}, which ensures that the generators of ``concurrent'' operations can be executed in any order, we get that for any \crdtlinearizable{} object $\aobj$ which admits timestamp-order linearizations, every linearization of a history of $\aobj$ which is also consistent with the order between timestamps is a valid \crdtlinearization{}.
For a history $\ahist=(\alabelset,\avisord)$, let $\atsord{\ahist}$ be an order between the labels in $\alabelset$ such that $\alabel_1\atsord{\ahist}\alabel_2$ iff $\tsof_\ahist(\alabel_1) < \tsof_\ahist(\alabel_2)$.
\vspace{-5pt}

\begin{lemma}
Let $\aobj$ be an object which is \crdtlinearizable{} w.r.t. a specification $\Spec$ and admits timestamp-order linearizations. Then, for every history $\ahist=(\alabelset,\avisord)$ of $\aobj$ and every sequence $(\alabelset,\alinord)$ which is consistent with $\avisord$ and $\atsord{\ahist}$, we have that $(\alabelset,\alinord)$ is an \crdtlinearization{} of $\ahist$ w.r.t. $\Spec$.
\end{lemma}
\vspace{-5pt}
\fi

We define a restriction $\otimes_{\tsof}$ of the object composition $\otimes$ such that the set of histories $h=(\alabelset,\avisord)$ in the composition $\aobj_1\otimes_{\tsof}\aobj_2$ satisfy the property that the order between timestamps (of all objects) is consistent with the visibility relation $\avisord$ (i.e., $\avisord\ \cup \atsord{\ahist}$ is acyclic). With respect to the ``unrestricted'' composition $\otimes$ defined in \sectionautorefname~\ref{ssec:comp_intro}, we only modify the transition rule corresponding to generators, as shown in \figureautorefname~\ref{fig:comp-ts}. This ensures that a new generated timestamp is bigger than all the timestamps ``visible'' to the replica executing that generator (irrespectively of the object). 
\ifshort
\else
Its extension to a set of objects is defined as usual. 
\fi
The composition operator  $\otimes_{\tsof}$ is called \emph{shared timestamp generator composition}. Practically, if we were to consider the standard timestamp mechanism used in CRDTs, i.e., each replica maintains a counter which is increased monotonically with every new operation (originating at the replica or delivered from another replica) and timestamps are defined as pairs of replica identifiers and counter values, then $\otimes_{\tsof}$ can be implemented using a ``shared'' counter which increases monotonically with every new operation, independently of the object on which it is applied.

The following theorem shows that the composition of \crdtlinearizable{} objects that admit execution-order or timestamp-order linearizations
is \crdtlinearizable{}, provided that all the objects share the same timestamp generator.

\begin{theorem}\label{th:comp_all}
The shared timestamp generator composition of a set of \crdtlinearizable{} objects that admit execution-order or timestamp-order linearizations is \crdtlinearizable{}.
\end{theorem}
\vspace{-5pt}

\ifshort
\else
The theorem above shows that any shared timestamp generator composition of the objects mentioned in next section is RA-linearizable.
\fi

\section{Mechanizing RA-Linearizability Proofs}\label{sec:mec}

To validate our approach, we considered a range of CRDTs listed
in~\figureautorefname~\ref{fig:crdt-implementaton of this paper, their
  correctness, and their interface} and mechanized their
RA-linearizability proofs using Boogie~\cite{BarnettCDJL05}, a
verification tool.
More precisely, we mechanized the proofs of conditions like $\mathsf{Commutativity}$ and
$\mathsf{Refinement}$ which imply RA-linearizability by the results in Section~\ref{sec:proofs}.
Beyond operation-based CRDTs (discussed in the paper), we have also
considered \emph{state-based} CRDTs, where an update occurs only at
the origin, and replicas exchange their \emph{states} instead of operations, and
states from other replicas are merged at the replica receiving them.
The merge function corresponds to the least upper bound operator in a
certain join semi-lattice defined over replica states.

For operation-based CRDTs, we have mechanized the proof of a strenghtening of 
$\mathsf{Commutativity}$ that avoids reasoning about traces and the proof of $\mathsf{Refinement}$ (or
$\mathsf{Refinement}_{\tsof{}}$). Concerning $\mathsf{Commutativity}$, our proofs encode 
two effectors as a single procedure which executes on two equal copies of the replica state. 
In some cases, the precondition of this procedure encodes conditions which are satisfied 
anytime the two effectors are concurrent, e.g., the effector of an {\tt add} and resp., {\tt remove} of OR-Set
are concurrent when the argument {\tt k} of {\tt add} is not in the
argument {\tt R} of {\tt remove}. At least for the CRDTs we consider,
such characterizations are obvious and apply generically to any
conflict-resolution policy based on unique identifiers. 
In some cases, the effectors commute even if they are not concurrent, so no additional precondition is needed. 
We prove that the resulting states are identical after performing the effectors in different order in each of the states.
$\mathsf{Refinement}$ (or
$\mathsf{Refinement}_{\tsof{}}$) is reduced to proving that the
refinement mapping is an inductive invariant for a lock-step execution
of the CRDT implementation and its specification. 

For state-based CRDTs, we have identified a set of conditions similar
to those of operation-based CRDTs that imply
RA-linearizability (see~\cite{arxiv}). In this case, we don't rely on the 
causal delivery assumption.
Extending their semantics with an auxiliary variable maintaining a
correspondence between replica states and sets of operations that
produced them, we extract the visibility relation between operations
as in the case of operation-based CRDTs.
This enables a similar reasoning about RA-linearizability.
In particular, $\mathsf{Commutativity}$ is
replaced by few conditions that now characterize the relationship
between applying operations at a given replica and the merge
function.

\begin{figure}[t]
{\footnotesize
  \begin{minipage}[t]{.6\linewidth}
  \begin{tabular}{|l|c|r|}
    \hline
    {\bf CRDT} & {\bf Imp.} &  {\bf Lin.}\\
    \hhline{|===|}
    Counter~\cite{ShapiroPBZ11}& OB & EO \\
    \hline
    PN-Counter~\cite{ShapiroPBZ11}& SB &  EO \\
    \hline
    LWW-Register~\cite{DBLP:journals/rfc/rfc677} & OB & TO\\
    \hline
    Multi-Value Reg.~\cite{DBLP:conf/sosp/DeCandiaHJKLPSVV07} & SB &  EO\\
    \hline
    LWW-Element Set~\cite{ShapiroPBZ11}& SB &  TO\\
    \hline
\end{tabular}
\end{minipage}
  \begin{minipage}[t]{.35\linewidth}
\begin{tabular}{|l|c|r|}
  \hline
  {\bf CRDT} & {\bf Imp.} &  {\bf Lin.}\\
  \hhline{|===|}
  2P-Set~\cite{ShapiroPBZ11}& SB &  EO\\
  \hline
  OR-Set~\cite{ShapiroPBZ11}& OB &  EO\\
  \hline
  RGA~\cite{RohJKL11}& OB &  TO \\ 
  \hline
  Wooki~\cite{DBLP:conf/wise/WeissUM07}& OB & EO \\
  \hline
\end{tabular}

\end{minipage}
}
\vspace{-2mm}
\caption{CRDTs proved RA-linearizable and the class of linearizations
  used.
  SB: State-Based, OB: Operation-Based, EO: Execution-Order, TO: Timestamp-Order.
}
\label{fig:crdt-implementaton of this paper, their correctness, and their interface}
\vspace{-4mm}
\end{figure}

\section{Related Work}
\label{sec:rel-work}

\noindent
{\bf Correctness Criteria.}
\citet{BurckhardtGYZ14} gives the first formal framework
where CRDTs and other weakly consistent replicated systems can be
specified.
Their CRDT specifications are defined in terms of sets of \emph{partial orders}
as opposed to our sequential specifications, which we think are
easier to reason about when verifying clients.
Beyond simpler specifications, RA-linearizability is related to their formalization of
causal consistency, called \emph{causal convergence} in~\cite{DBLP:conf/popl/BouajjaniEGH17}.
Overall RA-linearizability differs from causal convergence in three
points: (1) query-update rewritings, which enable sequential
specifications and avoid partial orders, (2) the linearization
projected on updates must be admitted by the specification
(intuitively, this ensures that the "final" convergence state is valid
w.r.t.
the specification), and (3) the linearization is required to be
consistent with the visibility order from the execution, and
not an arbitrary one as in causal convergence.
The latter difference makes causal convergence not compositional.

Regarding convergence, RA-linearizability implies that there is a unique total order of updates, and 
therefore if at some point all updates are visible to all replicas, all subsequent query 
operations at any replica will return the same value. This is observably equivalent to strong eventual consistency~\cite{ShapiroPBZ11,GomesKMB17,ZellerBP14}. 
RA-linearizability is also stronger than the session guarantees of~\citet{TerryDPSTW94},
but weaker than sequential consistency~\cite{DBLP:journals/tc/Lamport79} and linearizability~\cite{HerlihyW90}.
RA-linearizable objects that admit execution-order linearizations are close to being linearizable since
the operations are linearized as they were issued at the origin replica, relative to wall-clock time. 
This is similar to linearizability, where each operation appears to take effect instantaneously 
between the wall-clock time of its invocation its response. Unlike linearizability,
RA-linearizability allows queries to return a response 
consistent with only a subsequence of its linearized-before
operations.

\noindent
{\bf Sequential Specifications for CRDTs.}
\citet{PerrinMJ14} provides Update Consistency (UC), a criterion which to
the best of our knowledge is the first to consider sequential specifications and 
characterize linear histories of operations. 
However UC is not compositional due to
an existential quantification over visibility relations like in causal convergence.
Moreover, \citet{PerrinMJ14} doesn't investigate  UC proof methodologies.

\citet{JagadeesanR18} provide a correctness criterion called \emph{SEC}, 
which differs from \CRDTLinshort{} in several points:
\begin{inparaenum}
\item Firstly, \CRDTLinshort{} 
  has a global total order for updates, unlike SEC whose definition is
  quite complex.
\item
  \label{it-dependencies} Secondly, CRDT specifications in SEC are
  parameterized by a \emph{dependency} relation at the
  level of the type's API.
  Then, SEC assumes that all independent operations commute and
  disregards their order even when issued by the same client.
  It is unclear how such a specification could adequately capture
  systems enforcing session guarantees~\cite{TerryDPSTW94}.
 \ifshort
 \else
  \CRDTLinshort{} strives to preserve the guarantees of the
  underlying system network
  since 
  certain
  operations that appear independent at the API level could be
  made data or control dependent by the client.
  This happens in our query-update rewritings where queries
  provide arguments for the subsequent update.
\fi
\item While SEC is also compositional, since operations
  from different objects are assumed independent,
  a history of two  different SEC objects is trivially SEC since the order between
  operations of different objects is ignored.
  We find this notion of composition problematic since 
  the composition of specifications cannot capture
  causality between different objects, a common
  pattern when writing distributed applications (e.g. for referential
  integrity in a key-value store).
  In \CRDTLinshort{} the composition of a set of objects respects the
  client's causality as illustrated by the failure to combine
  some per-object linearizations in~\figureautorefname~\ref{fig:negative_composition}.~\footnote{There are however per-object linearizations for this history which can be
    merged into a global linearization
    (see~\sectionautorefname~\ref{sec:compositionality}).}
\end{inparaenum}

\noindent
{\bf Verification of CRDTs}
There are several works that approach the problem of verifying that a
CRDT implementation is correct w.r.t. a specification.
In~\cite{BurckhardtGYZ14, AttiyaBGMYZ16, Burckhardt14} along with the
formal specification, proofs of correctness of implementations are
given for several CRDTs. Our $\mathsf{Refinement}$ property is inspired by
the Replication Aware Simulations in~\cite{BurckhardtGYZ14}.
\ifshort
\else
Compared to these simulations, our proofs remain at a much more
abstract level -- and are therefore simpler -- since our
specifications are simpler.
We acknowledge however that their method is more general than ours
which only applies to CRDTs satisfying \CRDTLinshort{}.
\fi
\citet{ZellerBP14} and~\citet{GomesKMB17} provide frameworks for the
verification of CRDTs in Isabelle/HOL\@.
\ifshort
\else
Both works introduce a methodology to specify and prove CRDTs
correct.
\fi
Their
proofs are similar to the simulations of~\cite{BurckhardtGYZ14},
albeit in a different specification language also based on partial
orders.

\section{Conclusion}
\label{sec:conclusion}

We presented RA-linearizability, a correctness criterion inspired by
linearizability, intended to simplify the specification of CRDTs 
by resorting to sequential reasoning for the specifications.
We provide proof methodologies for \CRDTLinshort{} for some
well documented CRDTs, and we prove that under certain conditions
these proofs guarantee the compositionality of \CRDTLinshort{}.
In the extended version of this paper~\cite{arxiv} we show 
how our techniques extend to state-based CRDTs.

There are some limitations of \CRDTLinshort{}. Firstly, as we showed
before, some CRDTs might not be RA-linearizable under a certain
API, but a slight change in the API renders them
\crdtlinearizable{}. We would like to investigate what constitutes an
API that enables \CRDTLinshort{} specifications.
secondly, while we argue that \CRDTLinshort{} simplifies
specifications, we leave as future work to show whether it can be
effectively used to verify client applications of a CRDT.

\bibliography{biblio,dblp}

\newpage

\appendix

\section{Transition Rules of the CRDT Semantics}

\begin{figure}[t]
  \begin{subfigure}[!ht]{.3\linewidth}
    \includegraphics[scale=.7]{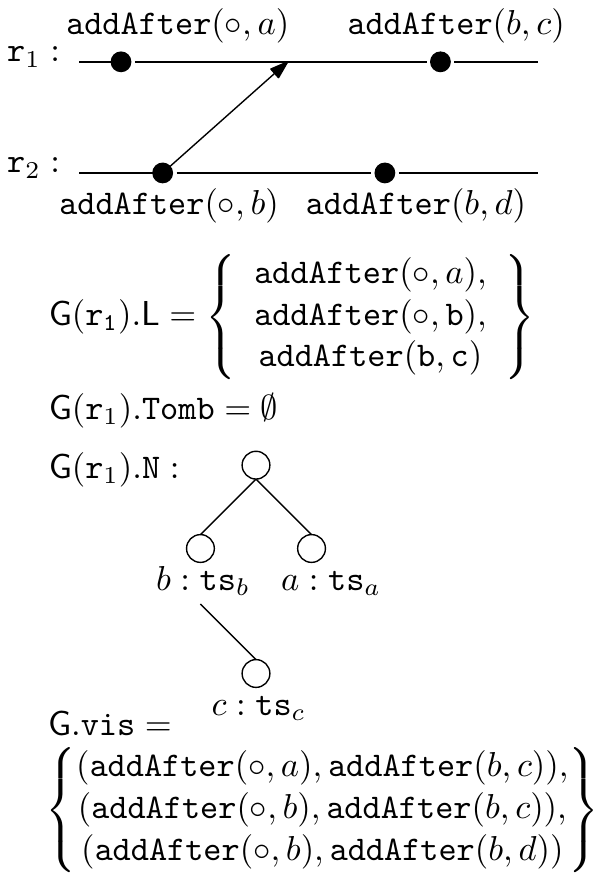}
      \vspace{1.3cm}
    \caption{}
    \label{fig:rga-sem-1}
  \end{subfigure}
  \begin{subfigure}[!ht]{.3\linewidth}
      \includegraphics[scale=.7]{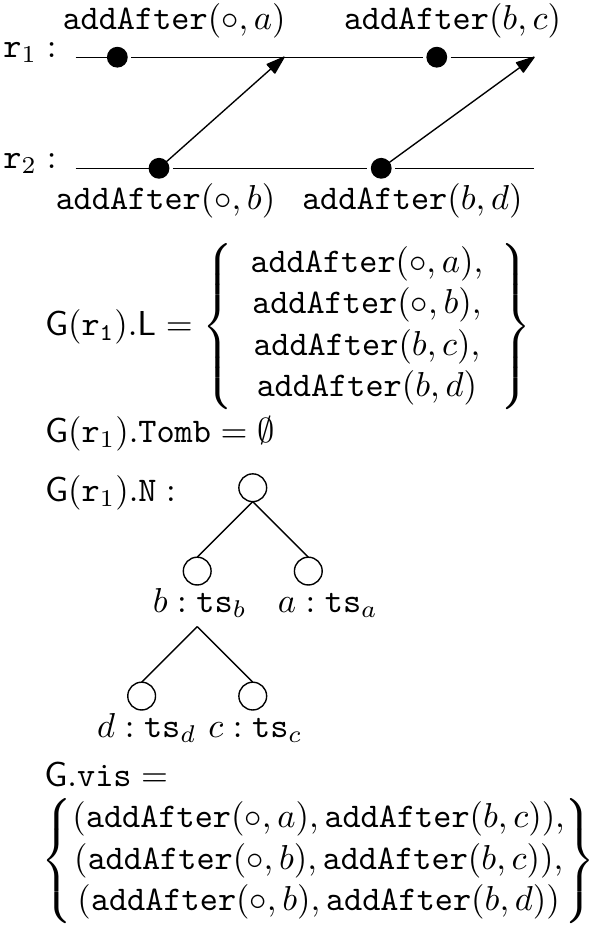}
      \vspace{1cm}
    \caption{}
    \label{fig:rga-sem-2}
  \end{subfigure}
  \begin{subfigure}[!ht]{.3\linewidth}
    \includegraphics[scale=.7]{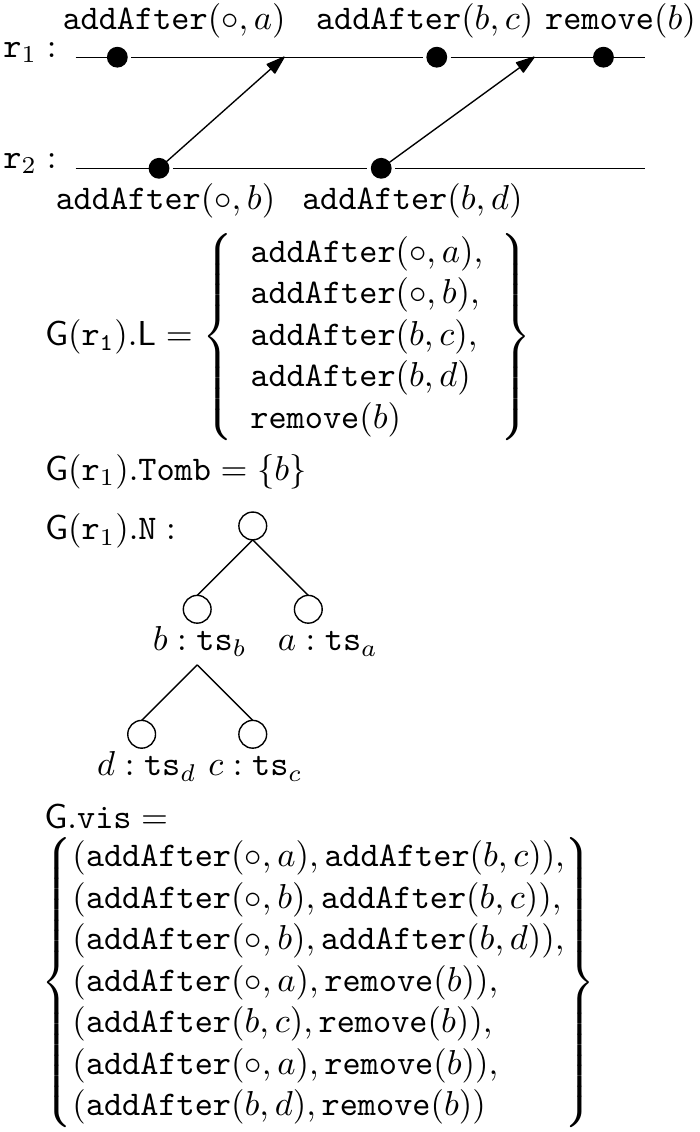}
    \caption{}
    \label{fig:rga-sem-3}
  \end{subfigure}
  \caption{Example of the semantics of RGA.}
  \label{fig:rga-sem}
\end{figure}

\autoref{fig:rga-sem} shows how some components of the semantics
progress according to the rules of~\autoref{fig:crdt-opsem} for the
RGA data type.
In particular we shown: the local labels of replica $\arep_1$
($\gstates(\arep_1).\alabelset$); its state, where we remove the
$\astate$ for succinctness (then $\gstates(\arep_1).\astate.\mathsf{N}$
becomes $\gstates(\arep_1).\mathsf{N}$); and the global visibility
relation $\gstates.\mathtt{vis}$.
The transition from~\autoref{fig:rga-sem-2} to~\autoref{fig:rga-sem-3}
shows an 
transition where the operation
$\alabelshort[{\tt remove}]{b}$ is executed by replica $\arep_1$.
Notice in particular how the global visibility relation is extended.

In~\autoref{fig:rga-sem}, the transition from~\autoref{fig:rga-sem-1}
to~\autoref{fig:rga-sem-2} corresponds to a 
transition which extends the visibility of the operation
$\alabelshort[{\tt addAfter}]{b,d}$ to $\arep_1$.
Notice that this is reflected in the $\gstates(\arep_1).\alabelset$
component.
The relation $\gstates.\mathtt{vis}$ does not change since this
relation only changes when a new operation is executed at the source
replica.

\section{Implementations of Operation-Based CRDT and Their Sequential Specifications}
\label{sec:implementation of operation-based CRDT and their sequential specifications}

\subsection{Operation-Based Counter and Its Sequential Specification}
\label{subsec:operation-based counter and its sequential specification}

\noindent {\bf Implementation}: The operation-based counter of \cite{ShapiroPBZ11} is shown in Listing~\ref{lst:operation-based counter}. It implements a counter interface with operations: ${\tt inc}()$, ${\tt dec}()$ and ${\tt read}$. A payload is a integer $ctr$. {\tt inc} increase the counter value of replica state by $1$, {\tt dec} decrease the counter value of replica state by $1$, and {\tt read} returns the counter value of replica state.

\begin{figure}[t]
\begin{lstlisting}[frame=top,caption={Pseudo-code of operation-based counter},
captionpos=b,label={lst:operation-based counter}]
  payload integer ctr
  initial ctr = 0

  inc() :
    generator :
    effector(inc) :
      ctr = ctr + 1

  dec() :
    generator :
    effector(dec) :
      ctr = ctr - 1

  read() :
    return ctr
\end{lstlisting}
\end{figure}

\noindent {\bf Sequential Specification $\specCounter$}: Each abstract state $\abstate$ is a integer. The sequential specification $\specCounter$ of counter is defined by:

\[
  \begin{array}[rcl]{rcl}
    \abstate & \specarrow{\alabelshort[{\tt inc}]{}} & \abstate+1\\
    \abstate & \specarrow{\alabelshort[{\tt dec}]{}} & \abstate-1\\
    \abstate & \specarrow{\alabellong[\mathsf{read}]{}{\abstate}} & \abstate
  \end{array}
\]

Method $\alabelshort[{\tt inc}]{}$ increase the counter value by $1$. Method $\alabelshort[{\tt dec}]{}$ decrease the counter value by $1$. Method $\alabellong[{\tt read}]{}{k}{}$ returns the counter value.

\subsection{Operation-Based Last-Writer-Win Register and its Sequential Specification}
\label{subsec:operation-based last-writer-win register and its sequential specification}

\noindent {\bf Implementation}: The operation-based last-writer-win Register (LWW-Register) of \cite{ShapiroPBZ11} is shown in Listing~\ref{lst:LWW-register}. It implements a register interface with operations: ${\tt write}(a)$ and ${\tt read}$. A payload is a tuple $(x,\ats)$ of a data value $x$ and its timestamp $\ats$. Here $x_0$ is an initial data and $\ats_0$ is an initial timestamp.

${\tt write}(a)$ generates a new timestamp $\ats'$ and modifies the replica state into $(a,\ats')$, and its effector uses the argument $(a,\ats')$. When applying an effector with arguments $(a,\ats')$ from a replica $(b,\ats_b)$, the resulting replica state is $(a,\ats')$ if $\ats_b<\ats'$, and is $(b,\ats_b)$ otherwise. {\tt read} returns the data value of replica state.

\begin{figure}[t]
\begin{lstlisting}[frame=top,caption={Pseudo-code of operation-based LWW-register},
captionpos=b,label={lst:LWW-register}]
  payload X x, timestamp @|$\ats$|@
  initial @|$x_0$|@, @|$\ats_0$|@

  write(a) :
    generator :
      let @|$\ats'$|@ = getTimestamp()
    effector((a,@|$\ats'$|@)) :
      if (@|$\ats<\ats'$|@)
        (x,@|$\ats$|@) = (a,@|$\ats'$|@)

  read() :
    return x
\end{lstlisting}
\end{figure}

\noindent {\bf Sequential Specification $\specReg$}: Each abstract state $\abstate$ is a data value. The sequential specification $\specReg$ of LWW-register is defined by:

\[
  \begin{array}{rcl}
    \abstate
    & \specarrow{\alabelshort[\mathtt{write}]{a}}
    & a\\
    \abstate
    & \specarrow{\alabellong[\mathtt{read}]{}{\abstate}{}}
    & \abstate
  \end{array}
\]

Method $\alabelshort[{\tt write}]{a}$ update the abstract state into data value $a$. Method $\alabellong[{\tt read}]{}{\abstate}{}$ returns the value of the register.

\subsection{Operation-Based Wooki and its Sequential Specification}
\label{subsec:operation-based wooki and its sequential specification}

\noindent {\bf Implementation}: The Wooki implementation of \cite{DBLP:conf/wise/WeissUM07} is given in Listing~\ref{lst:wooki algorithm}. 
Wooki is an optimized version of Woot \cite{DBLP:conf/cscw/OsterUMI06}. To make our introduction more clear, we borrow the notion of W-character and W-string from Woot.

Wooki implements a list interface with operations: ${\tt addBetween}(a,b,c)$, ${\tt remove}(a)$ and ${\tt read}$. A payload is a W-string (introduced below) $string_s$ which stores the information of list content as well as tombstone.

A W-character $w$ is a tuple $(id,v,degree,flag)$, and is used to stroe the information of an element of list. Here $id$ is the identifier of $w$; $v$ is the value of $w$; $degree$ is the degree of $w$; $flag \in \{ \mathit{true},\mathit{false} \}$ is the flag of $w$ and indicates whether $w$ is ``visible'' in list. A identifier $id$ of W-character is a unique timestamp. We use $degree(w)$ to denote the degree of $w$. A degree is a integer that is fixed when its W-character is generated; when inserting an W-character into $string_s$, the degree of W-characters of $string_s$ influenced the position where this W-character will be inserted into.

Let $\circ_{begin}$ and $\circ_{end}$ be two special values. Let $w_{begin} = (\_,\circ_{begin},0,\mathit{true})$ and $w_{end} = (\_,\circ_{end},\mathit{true})$ be two special W-characters. A W-string is an ordered sequence of W-characters $w_{begin} \cdot w_1 \cdot \ldots \cdot w_n \cdot w_{end}$. 
We never remove $\circ_{begin}$ or $\circ_{end}$, and never put value before $\circ_{begin}$ or after $\circ_{end}$. Since $w_{begin}$ and $w_{end}$ is fixed to be the head and tail of a W-string, in the latter part of this paper, when we considering the content of a W-string, we ignore $w_{begin}$ and $w_{end}$. We define the following functions for a W-string $str$:

\begin{itemize}
\setlength{\itemsep}{0.5pt}
\item[-] $\vert str \vert$ returns the length of $str$,

\item[-] $str[p]$ returns the W-character at position $p$ in $str$. Here we assume that the first element of $str$ is at position 0.

\item[-] ${\tt pos}(str,w)$ returns the position of W-character $w$ in $str$.

\item[-] ${\tt insert}(str,w,p)$ inserts W-character $w$ into $str$ at position $p$.

\item[-] ${\tt subseq}(str,w_1,w_2)$ returns the part of $str$ between the W-characters $w_1$ and $w_2$ (excluding $w_1$ and $w_2$).

\item[-] ${\tt contains}(str,a)$ returns true if there exists a W-character in $str$ with value $a$.

\item[-] ${\tt values}(str)$ returns the sequence of visible (with $\mathit{true}$ flag) values of $str$.

\item[-] ${\tt getWchar}(str,a)$ returns the W-character with value $a$ in $str$.

\item[-] ${\tt changeFlag}(str,pos,f)$ changes the flag of $str[pos]$ into boolean value $f$.
\end{itemize}

\begin{figure}[H]
\begin{lstlisting}[frame=top,caption={Pseudo-code of Wooki},
captionpos=b,label={lst:wooki algorithm}]
  payload W-string @|$string_s$|@
  initial @|$string_s$|@ = @|$\epsilon$|@

  addBetween(a,b,c) :
    generator :
      precondition : @|$c \neq \circ_{begin}$|@ @|$\wedge$|@ @|$a \neq \circ_{end}$|@ @|$\wedge$|@ @|$b \neq \circ_{begin}$|@ @|$\wedge$|@ @|$b \neq \circ_{end}$|@ @|$\wedge$|@ @|$contains(string_s,a)$|@ @|$\wedge$|@ @|$contains(string_s,c)$|@ @|$\wedge$|@ @|$pos(string_s,c) > pos(string_s,a)$|@ @|$\wedge$|@ @|$\neg contains(string_s,b)$|@
      let @|$\ats$|@ = getTimestamp()
      let @|$w_p$|@ = @|$getWchar(string_s,a)$|@
      let @|$w_n$|@ = @|$getWchar(string_s,c)$|@
    effector((w,@|$w_p$|@,@|$w_n$|@)) : with @|$w$|@ = @|$(\ats,b,max(degree(w_p),degree(w_n))+1,\mathit{true})$|@
      integrateIns(@|$w_p,w,w_n$|@)

  remove(a) :
    generator :
      precondition : @|$a \neq \circ_{begin} \wedge a \neq \circ_{end} \wedge contains(string_s,a)$|@
      let w = @|$getWchar(string_s,a)$|@
    effector(w) :
      let p = @|$pos(string_s,w)$|@
      @|$changeFlag(string_s,p,\mathit{false})$|@

  read() :
    let s = @|$values(string_s)$|@
    return s

  integrateIns(@|$w_p,w,w_n$|@)
    let @|$S$|@ = @|$string_s$|@
    let @|$S'$|@ = @|$subseq(S,w_p,w_n)$|@
    if @|$S' = \epsilon$|@
      then  @|$insert(S,w,pos(S,w_n))$|@
    else
      Let i = 0
      Let @|$d_{min}$|@ be the minimal degree of W-characters in @|$S'$|@
      Let F be projection of @|$S'$|@ into W-characters with degree @|$d_{min}$|@
      if (@|$w<_{id} F[0]$|@)
        integrateIns(@|$w_p,w,F[0]$|@)
      else
        while (@|$i < \vert F \vert -1 \wedge F[i] <_{id} w$|@) do
            i = i+1
        if (@|$i = \vert F \vert -1 \wedge F[i] <_{id} w$|@)
            integrateIns(@|$F[i],w,w_n$|@)
        else
            integrateIns(@|$F[i-1],w,F[i]$|@)
\end{lstlisting}
\end{figure}

A total order $<_{id}$ is given for identifiers of W-characters for conflict resolution, and let $<_{id}$ be the total order of timestamps (identifiers of W-charcaters). Given a W-string $str$ and two W-characters $w_1,w_2$ of $str$, we write $w_1 <_{str} w_2$ to indicate that $pos(str,w_1) < pos(str,w_2)$.

Note that, here we use {\tt addBetween} method, which ``adds a value between two values'', and is different from the {\tt addAfter} method of RGA. Also, note that a value can be inserted into the list only once.

${\tt addBetween}(a,b,c)$ intends to generate a new W-character $w$ with value $b$ and put it between W-characters with values $a$ and $c$,respectively. ${\tt addBetween}(a,b,c)$ work as follows: First, we check the precondition, such as three exists W-characters with values $a$ and $c$ in $string_s$, and the W-character with value $a$ is before the W-character with value $c$ in $string_s$. Then, let $w_p$ and $w_n$ be the W-characters with value $a$ and $c$ in $string_s$, respectively, and we generate a W-character $w = (\ats,b,max(degree(w_p),degree(w_n))+1,\mathit{true})$ with value $b$. Finally, we call method $integrateIns(w_p,w,w_n)$ to put $w$ into $string_s$ at some position between $w_p$ and $w_n$.

$integrateIns(w_p,w,w_n)$ is a recursive method and works as follows: If there are no W-character between $w_p$ and $w_n$, then $w$ is put after $w_p$. Else, Wooki selects a sequence $F$ of W-characters, such that each W-character of $F$ is between $w_p$ and $w_n$, and has minimal degree. 
Then, we choose the position of $w$ according to $<_{id}$ order, and recursive call $integrateIns(w_x,w,w_y)$ for some $w_x,w_y \in F \cup \{ w_p,w_n \}$. We can see that, the minimal degree of W-characters of $subseq(S,w_x,w_y)$ 
is larger than that of $subseq(S,w_p,w_n)$, and  $subseq(S,w_x,w_y)$ is a sub-sequence of $subseq(S,w_p,w_n)$.

$\alabelshort[{\tt remove}]{a}$ first finds the W-character with value $a$ in $string_s$ and then sets its flag into $\mathit{false}$. ${\tt read}$ uses $values(string_s)$ to return the list content of $string_s$.

\noindent {\bf Sequential Specification $\specWooki$}: Each abstract state $\abstate = (l,T)$ contains a sequence $l$ of elements of a given type and a set $T$ of elements appearing in the list. The element $l$ is the list of all input values, whether already removed or not; while $T$ stores the removed values and is used as \emph{tombstone}. The sequential specification $\specWooki$ of list with add-between interface is defined by:
\[
  \begin{array}{rcl}
    \big(\ (l_1 \cdot a \cdot l_2 \cdot l_3 \cdot c \cdot l_4,T\big)\ |\ b\text{ is fresh}, a \neq \circ_{end}, c \neq \circ_{begin} \big)
     & \specarrow{\alabelshort[\mathtt{addBetween}]{a,b,c}}
     & (l_1 \cdot a \cdot l_2 \cdot b \cdot l_3 \cdot c \cdot l_4,T)\\
     \big(\ (l,T)\ |\ a\ \in l, \ a \neq \circ_{begin}, \ a \neq \circ_{end} \big) 
     & \specarrow{\alabelshort[\mathtt{remove}]{a}}
     & (l,T \cup \{a\})\\
     (l,T)
     & \specarrow{\alabellong[\mathtt{read}]{}{(l/T)}{}}
     & (l,T)
   \end{array}
\]

The method $\alabelshort[\mathtt{addBetween}]{a,b,c}$ puts $b$ at some random position between $a$ and $c$ in $l$, assuming that each value is put into list at most once. Method $\alabelshort[\mathtt{remove}]{a}$ adds $a$ into $T$, hence removing $a$ from the list for subsequent calls to the $\mathtt{read}$ method. Finally, $\alabellong[\mathtt{read}]{}{s}{}$ returns the list content excluding any element appearing in $T$. Assume that the initial value of list is $(\circ_{begin} \cdot \circ_{end},\emptyset)$, we never put value after $\circ_{end}$ or before $\circ_{begin}$, and $\circ_{begin}$ and $\circ_{end}$ are never removed. We will sometimes ignore the value $\circ_{begin}$ and $\circ_{end}$ from the result of $\ensuremath{\tt read}$.

\section{Discussion about List with Interface of Index ({\tt addAt})}
\label{sec:discussion about list with addAt interface}

In this section, we consider list specification with {\tt addAt} method, which puts a value at an index instead of putting a value after another value or between two values.

We first give two version of list specifications with {\tt addAt} method, while one of them does not use tombstone and one of them uses tombstone. We prove that RGA is not \crdtlinearizable{} w.r.t any of them.

Then, we give a third version of list specification with {\tt addAt} method, $\specAddatThree$, where we use a ``local version of index'', and each method is required to returns the ``local list content''. We prove that RGA is \crdtlinearizable{} w.r.t $\specAddatThree$.

\subsection{A first version of list with {\tt addAt} Interface}
\label{subsec:a first version of list with addAt interface}

The first version of list with {\tt addAt} interface is as follows: A list uses the following three methods:

\begin{itemize}
\setlength{\itemsep}{0.5pt}
\item[-] $\alabelshort[{\tt addAt}]{a,k}$: Inserts value $a$ into position $k$ of the list of replica state. For $k$ exceeding the list size of replica state, $k$ will be inserted at the end of the list.

\item[-] $\alabelshort[{\tt remove}]{a}$: Remove value $a$ from list of the replica state.

\item[-] $\alabellong[{\tt read}]{}{s}{}$: Returns the list content of the replica state.
\end{itemize}

Here we assume the first element of a sequence is at position $0$.

Our RGA algorithm of \sectionautorefname \ref{sec:overview} can be modified as follows for this interface: To do $\alabelshort[{\tt addAt}]{a,k}$, we work as follows:

\begin{itemize}
\setlength{\itemsep}{0.5pt}
\item[-] Let $s = traverse(N, Tomb)$ be the list content in replica state. If $s = \epsilon$, then let $b = \circ$; else, if $\vert s \vert \geq k$, then let $b=s[k-1]$; else, let $b=s[\vert s \vert -1]$. Here let $\vert s \vert$ be the length of $s$. 

\item[-] Work as $\alabelshort[{\tt addAfter}]{b,a}$. Especially, the effector of $\alabelshort[{\tt addAt}]{a,k}$ is the effector of ${\tt addAfter}(b,$ $a)$.
\end{itemize}

\subsection{Two Sequential Specification for list with {\tt addAt} Interface}
\label{subsec:two sequential specification for list with addAt interface}

Depending on whether use tombstone or not, there are two sequential specification for list with {\tt addAt} interface.

\noindent {\bf Sequential specification $\specAddatOne$}: The first sequential specification $\specAddatOne$ of list with {\tt addAt} interface does not use tombstone. Each abstract state $\abstate = l$ contains a sequence $l$ of elements of a given type. The element $l$ is the list of values that has been input and not removed yet. The sequential specification $\specAddatOne$ of list with {\tt addAt} interface is defined by:
\[
  \begin{array}{rcl}
    \big(\ l_1 \cdot l_2 \ |\ a\text{ is fresh}\ , \vert l_1 \vert = k \big)
     & \specarrow{\alabelshort[\mathtt{addAt}]{a,k}}
     & l_1 \cdot a \cdot l_2 \\
      \big(\ l \ |\ a\text{ is fresh}\ , \vert l \vert < k \big)
     & \specarrow{\alabelshort[\mathtt{addAt}]{a,k}}
     & l \cdot a \\
     l_1 \cdot a \cdot l_2
     & \specarrow{\alabelshort[\mathtt{remove}]{a}}
     & l_1 \cdot l_2 \\
     l
     & \specarrow{\alabellong[\mathtt{read}]{}{l}{}}
     & l
   \end{array}
\]

If $\vert l_1 \cdot l_2 \vert \geq k$ and $\vert l_1 \vert = k$, the method $\alabelshort[\mathtt{addAt}]{a,k}$ puts $a$ immediately after $l_1$, assuming that each value is put into list at most once; Else, if $\vert l \vert <k$, the method $\alabelshort[\mathtt{addAt}]{a,k}$ puts $a$ immediately after $l$, assuming that each value is put into list at most once. Method $\alabelshort[\mathtt{remove}]{a}$ remove $a$ from $l_1 \cdot a \cdot l_2$. $\alabellong[\mathtt{read}]{}{l}{}$ returns the list content.

\noindent {\bf Sequential specification $\specAddatTwo$}: The second sequential specification $\specAddatTwo$ of list with {\tt addAt} interface uses tombstone. Each abstract state $\abstate = (l,T)$ contains a sequence $l$ of elements of a given type and a set $T$ of elements appearing in the list. The element $l$ is the list of all input values, whether already removed or not; while $T$ stores the removed values and is used as \emph{tombstone}. We can safely assume that $l$ contains more or equal values than $T$. The sequential specification $\specAddatTwo$ of list with {\tt addAt} interface is defined by: 

\[
  \begin{array}{rcl}
    \big(\ (l_1 \cdot l_2,T\big)\ |\ a\text{ is fresh}\ , \vert l_1/T \vert = k \big)
     & \specarrow{\alabelshort[\mathtt{addAt}]{a,k}}
     & (l_1 \cdot a \cdot l_2,T)\ \\
     \big(\ (l,T\big)\ |\ a\text{ is fresh}\ , \vert l/T \vert < k \big)
     & \specarrow{\alabelshort[\mathtt{addAt}]{a,k}}
     & (l \cdot a,T)\ \\
     \big(\ (l,T)\ |\ a\ \text{occurs in}\ l\ \big)
     & \specarrow{\alabelshort[\mathtt{remove}]{a}}
     & (l,T \cup \{a\})\\
     (l,T)
     & \specarrow{\alabellong[\mathtt{read}]{}{(l/T)}{}}
     & (l,T)\
   \end{array}
\]

If $l_1 \cdot l_2 \geq k$ and $\vert l_1/T \vert = k$, the method $\alabelshort[\mathtt{addAt}]{a,k}$ puts $a$ immediately after $l_1$, assuming that each value is put into list at most once; Else, if $\vert l/T \vert <k$, the method $\alabelshort[\mathtt{addAt}]{a,k}$ puts $a$ immediately after $l$, assuming that each value is put into list at most once. Method $\alabelshort[\mathtt{remove}]{a}$ adds $a$ into $T$, hence removing $a$ from the list for subsequent calls to the $\mathtt{read}$ method. Finally, $\alabellong[\mathtt{read}]{}{s}{}$ returns the list content excluding any element appearing in $T$.

Note that the {\tt addAt} method of $\specAddatTwo$ is nondeterministic, since the values of $T$ in $l_1$ does not influence exeuction. For example, $(a \cdot b, \{ a \}) \specarrow{\alabelshort[\mathtt{addAt}]{c,0}} (c \cdot a \cdot b, \{ a \})$ and $(a \cdot b, \{ a \}) \specarrow{\alabelshort[\mathtt{addAt}]{c,0}} (a \cdot c \cdot b, \{ a \})$ are both valid transitions. However, in the proof of Lemma \ref{lemma:The history h of fig an example history shows that RGA is not RA-linarizable w.r.t specAddat is not RA-linearizable w.r.t SpecAddatOne or SpecAddatTwo} we prove that when we consider only executions where each value can be removed only once, the admitted sequences of $\specAddatTwo$ is a subset of the admitted sequences of $\specAddatOne$. Or we can say, the behavior of $\specAddatTwo$ is ``deterministic''to some extent.

\subsection{RGA not \crdtlinearizable{} w.r.t $\specAddatOne$ or $\specAddatTwo$}
\label{subsec:RGA not RA-linearizable w.r.t specAddatOne or SpecAddatTwo}

\autoref{fig:an example history shows that RGA is not RA-linearizable w.r.t specAddat} shows an example that is a history of RGA and is not \crdtlinearizable{} w.r.t $\specAddatOne$ or $\specAddatTwo$. Here we assume that $\ats_a < \ats_b < \ats_c < \ats_d < \ats_e$. We also draw the replica state of replica $r_2$ and $r_3$ after the execution of $h$ in \autoref{fig:an example history shows that RGA is not RA-linearizable w.r.t specAddat}.

\begin{figure}[!h]
  \centering
  \includegraphics[width=0.9 \textwidth]{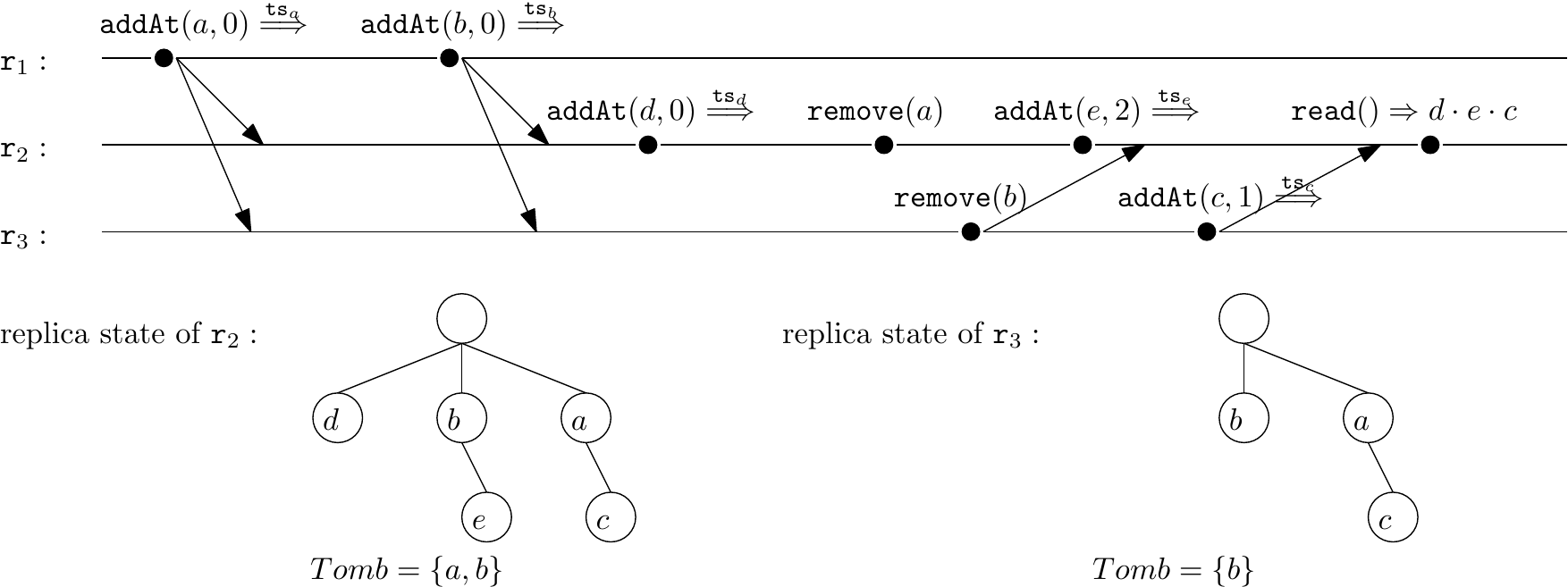}
\vspace{-10pt}
  \caption{An example history shows that RGA is not \crdtlinearizable{} w.r.t .}
  \label{fig:an example history shows that RGA is not RA-linearizable w.r.t specAddat}
\end{figure}

The following lemma states that RGA is not \crdtlinearizable{} w.r.t $\specAddatOne$ or $\specAddatTwo$.

\begin{lemma}
\label{lemma:The history h of fig an example history shows that RGA is not RA-linarizable w.r.t specAddat is not RA-linearizable w.r.t SpecAddatOne or SpecAddatTwo}
RGA is not \crdtlinearizable{} w.r.t $\specAddatOne$ or $\specAddatTwo$.
\end{lemma}

\begin {proof}
Let $h$ be the history of \autoref{fig:an example history shows that RGA is not RA-linearizable w.r.t specAddat}. 
It is obvious that $h$ is a history of RGA. We prove that RGA is not \crdtlinearizable{} w.r.t $\specAddatOne$ by proving that $h$ is not \crdtlinearizable{} w.r.t $\specAddatOne$, and we prove this by showing that all possible linarizations of $h$ can not validate the $\alabellong[\mathtt{read}]{}{d \cdot e \cdot c}{}$ operation. We list all possible linearization and their transitions in $\specAddatOne$ as follows:

\begin{itemize}
\setlength{\itemsep}{0.5pt}
\item[-] $\epsilon \specarrow{\alabelshort[\mathtt{addAt}]{a,0}}$ $a \specarrow{\alabelshort[\mathtt{addAt}]{b,0}}$ $b \cdot a \specarrow{\alabelshort[\mathtt{remove}]{b}}$ $a \specarrow{\alabelshort[\mathtt{addAt}]{c,1}}$ $a \cdot c \specarrow{\alabelshort[\mathtt{addAt}]{d,0}}$ $d \cdot a \cdot c \specarrow{\alabelshort[\mathtt{remove}]{a}}$ $d \cdot c \specarrow{\alabelshort[\mathtt{addAt}]{e,2}}$ $d \cdot c \cdot e \specarrow{\alabellong[\mathtt{read}]{}{d \cdot c \cdot e}{}} d \cdot c \cdot e$.

\item[-] $\epsilon \specarrow{\alabelshort[\mathtt{addAt}]{a,0}}$ $a \specarrow{\alabelshort[\mathtt{addAt}]{b,0}}$ $b \cdot a \specarrow{\alabelshort[\mathtt{remove}]{b}}$ $a \specarrow{\alabelshort[\mathtt{addAt}]{d,0}}$ $d \cdot a \specarrow{\alabelshort[\mathtt{addAt}]{c,1}}$ $d \cdot c \cdot a \specarrow{\alabelshort[\mathtt{remove}]{a}}$ $d \cdot c \specarrow{\alabelshort[\mathtt{addAt}]{e,2}}$ $d \cdot c \cdot e \specarrow{\alabellong[\mathtt{read}]{}{d \cdot c \cdot e}{}} d \cdot c \cdot e$.

\item[-] $\epsilon \specarrow{\alabelshort[\mathtt{addAt}]{a,0}}$ $a \specarrow{\alabelshort[\mathtt{addAt}]{b,0}}$ $b \cdot a \specarrow{\alabelshort[\mathtt{remove}]{b}}$ $a \specarrow{\alabelshort[\mathtt{addAt}]{d,0}}$ $d \cdot a \specarrow{\alabelshort[\mathtt{remove}]{a}}$ $d \specarrow{\alabelshort[\mathtt{addAt}]{c,1}}$ $d \cdot c \specarrow{\alabelshort[\mathtt{addAt}]{e,2}}$ $d \cdot c \cdot e \specarrow{\alabellong[\mathtt{read}]{}{d \cdot c \cdot e}{}} d \cdot c \cdot e$.

\item[-] $\epsilon \specarrow{\alabelshort[\mathtt{addAt}]{a,0}}$ $a \specarrow{\alabelshort[\mathtt{addAt}]{b,0}}$ $b \cdot a \specarrow{\alabelshort[\mathtt{remove}]{b}}$ $a \specarrow{\alabelshort[\mathtt{addAt}]{d,0}}$ $d \cdot a \specarrow{\alabelshort[\mathtt{remove}]{a}}$ $d \specarrow{\alabelshort[\mathtt{addAt}]{e,2}}$ $d \cdot e \specarrow{\alabelshort[\mathtt{addAt}]{c,1}}$ $d \cdot c \cdot e \specarrow{\alabellong[\mathtt{read}]{}{d \cdot c \cdot e}{}} d \cdot c \cdot e$.

\item[-] $\epsilon \specarrow{\alabelshort[\mathtt{addAt}]{a,0}}$ $a \specarrow{\alabelshort[\mathtt{addAt}]{b,0}}$ $b \cdot a \specarrow{\alabelshort[\mathtt{addAt}]{d,0}}$ $d \cdot b \cdot a \specarrow{\alabelshort[\mathtt{remove}]{b}}$ $d \cdot a \specarrow{\alabelshort[\mathtt{addAt}]{c,1}}$ $d \cdot c \cdot a \specarrow{\alabelshort[\mathtt{remove}]{a}}$ $d \cdot c \specarrow{\alabelshort[\mathtt{addAt}]{e,2}}$ $d \cdot c \cdot e \specarrow{\alabellong[\mathtt{read}]{}{d \cdot c \cdot e}{}} d \cdot c \cdot e$.

\item[-] $\epsilon \specarrow{\alabelshort[\mathtt{addAt}]{a,0}}$ $a \specarrow{\alabelshort[\mathtt{addAt}]{b,0}}$ $b \cdot a \specarrow{\alabelshort[\mathtt{addAt}]{d,0}}$ $d \cdot b \cdot a \specarrow{\alabelshort[\mathtt{remove}]{b}}$ $d \cdot a \specarrow{\alabelshort[\mathtt{remove}]{a}}$ $d \specarrow{\alabelshort[\mathtt{addAt}]{c,1}}$ $d \cdot c \specarrow{\alabelshort[\mathtt{addAt}]{e,2}}$ $d \cdot c \cdot e \specarrow{\alabellong[\mathtt{read}]{}{d \cdot c \cdot e}{}} d \cdot c \cdot e$.

\item[-] $\epsilon \specarrow{\alabelshort[\mathtt{addAt}]{a,0}}$ $a \specarrow{\alabelshort[\mathtt{addAt}]{b,0}}$ $b \cdot a \specarrow{\alabelshort[\mathtt{addAt}]{d,0}}$ $d \cdot b \cdot a \specarrow{\alabelshort[\mathtt{remove}]{b}}$ $d \cdot a \specarrow{\alabelshort[\mathtt{remove}]{a}}$ $d \specarrow{\alabelshort[\mathtt{addAt}]{e,2}}$ $d \cdot e \specarrow{\alabelshort[\mathtt{addAt}]{c,1}}$ $d \cdot c \cdot e \specarrow{\alabellong[\mathtt{read}]{}{d \cdot c \cdot e}{}} d \cdot c \cdot e$.

\item[-] $\epsilon \specarrow{\alabelshort[\mathtt{addAt}]{a,0}}$ $a \specarrow{\alabelshort[\mathtt{addAt}]{b,0}}$ $b \cdot a \specarrow{\alabelshort[\mathtt{addAt}]{d,0}}$ $d \cdot b \cdot a \specarrow{\alabelshort[\mathtt{remove}]{a}}$ $d \cdot b \specarrow{\alabelshort[\mathtt{remove}]{b}}$ $d \specarrow{\alabelshort[\mathtt{addAt}]{c,1}}$ $d \cdot c \specarrow{\alabelshort[\mathtt{addAt}]{e,2}}$ $d \cdot c \cdot e \specarrow{\alabellong[\mathtt{read}]{}{d \cdot c \cdot e}{}} d \cdot c \cdot e$.

\item[-] $\epsilon \specarrow{\alabelshort[\mathtt{addAt}]{a,0}}$ $a \specarrow{\alabelshort[\mathtt{addAt}]{b,0}}$ $b \cdot a \specarrow{\alabelshort[\mathtt{addAt}]{d,0}}$ $d \cdot b \cdot a \specarrow{\alabelshort[\mathtt{remove}]{a}}$ $d \cdot b \specarrow{\alabelshort[\mathtt{remove}]{b}}$ $d \specarrow{\alabelshort[\mathtt{addAt}]{e,2}}$ $d \cdot e \specarrow{\alabelshort[\mathtt{addAt}]{c,1}}$ $d \cdot c \cdot e \specarrow{\alabellong[\mathtt{read}]{}{d \cdot c \cdot e}{}} d \cdot c \cdot e$.

\item[-] $\epsilon \specarrow{\alabelshort[\mathtt{addAt}]{a,0}}$ $a \specarrow{\alabelshort[\mathtt{addAt}]{b,0}}$ $b \cdot a \specarrow{\alabelshort[\mathtt{addAt}]{d,0}}$ $d \cdot b \cdot a \specarrow{\alabelshort[\mathtt{remove}]{a}}$ $d \cdot b \specarrow{\alabelshort[\mathtt{addAt}]{e,2}}$ $d \cdot b \cdot e \specarrow{\alabelshort[\mathtt{remove}]{b}}$ $d \cdot e \specarrow{\alabelshort[\mathtt{addAt}]{c,1}}$ $d \cdot c \cdot e \specarrow{\alabellong[\mathtt{read}]{}{d \cdot c \cdot e}{}} d \cdot c \cdot e$.
\end{itemize}

Since the last {\tt read} operation returns $d \cdot c \cdot e$ instead of $d \cdot e \cdot c$, these linearization are not correct. Therefore, RGA is not \crdtlinearizable{} w.r.t $\specAddatOne$.

We prove that RGA is not \crdtlinearizable{} w.r.t $\specAddatTwo$ by proving that, when we consider only executions where each value can be removed only once, the admitted sequences of $\specAddatTwo$ is a subset of the admitted sequences of $\specAddatOne$. We need to prove that there exits a relation $R$ between abstracts of $\specAddatTwo$ and $\specAddatOne$, such that when there are transitions $\abstate_0 \specarrow{\alabel_0} \abstate_1 \ldots \specarrow{\alabel_n} \abstate_n$ of $\specAddatTwo$, and in $\alabel_0 \cdot \ldots \cdot \alabel_n$ each value is removed at most once, we have that, (1) $(\abstate_0,\abstate'_0) \in R$, where $\abstate_0$ and $\abstate'_0$ is the initial abstract state of $\specAddatTwo$ and $\specAddatOne$, respectively, and (2) if $\abstate_i \specarrow{\alabel_i} \abstate_{i+1}$ is a transition of $\specAddatTwo$ and $(\abstate_i,\abstate'_i) \in R$, then, $\abstate'_i \specarrow{\alabel_i} \abstate'_{i+1}$ is in $\specAddatOne$ and $(\abstate_{i+1},\abstate'_{i+1}) \in R$ for some abstract state $\abstate'_{i+1}$ of $\specAddatOne$.

Let relation $R$ be defined as follows: $((l,T),l') \in R$, if $l'=l/T$. Then, let us consider all possible transitions:

\begin{itemize}
\setlength{\itemsep}{0.5pt}
\item[-] If $(l_1 \cdot l_2, T) \specarrow{\alabelshort[\mathtt{addAt}]{a,k}} (l_1 \cdot a \cdot l_2,T)$ is in $\specAddatTwo$ and $\vert l_1/T \vert = k$: Then, we can see that $l' = l/T = (l_1/T) \cdot (l_2/T)$, and then, $(l_1/T) \cdot (l_2/T) \specarrow{\alabelshort[\mathtt{addAt}]{a,k}} (l_1/T) \cdot a \cdot (l_2/T)$ is in $\specAddatOne$. Since $a$ is fresh in $l_1 \cdot l_2$ and $l_1 \cdot l_2$ contains more or equal values than $T$, we can see that $a \notin T$ and $(l_1/T) \cdot a \cdot (l_2/T) = (l_1 \cdot a \cdot l_2)/T$. Or we can say, $((l_1 \cdot a \cdot l_2,T),(l_1/T) \cdot a \cdot (l_2/T)) \in R$.

\item[-] If $(l,T) \specarrow{\alabelshort[\mathtt{addAt}]{a,k}} (l \cdot a, T)$ is in $\specAddatTwo$ and $\vert l/T \vert <k$: Then, we can see that $l/T \specarrow{\alabelshort[\mathtt{addAt}]{a,k}} (l/T) \cdot a$ is in $\specAddatOne$ and $a$ is fresh in $l$. Since $l$ contains more or equal values than $T$, we can see that $a \notin T$. Then, $(l/T) \cdot a = (l \cdot a)/T$. Or we can say, $((l \cdot a,T),(l/T) \cdot a) \in R$.

\item[-] If $(l_1 \cdot a \cdot l_2,T) \specarrow{\alabelshort[\mathtt{remove}]{a}} (l_1 \cdot a \cdot l_2,T \cup \{a\})$ is in $\specAddatTwo$: Since each value is removed only once, we can see that $a \notin T$. Then, $( (l_1 \cdot a \cdot l_2,T), (l_1/T) \cdot a \cdot (l_2/T) ) \in R$.  $(l_1/T) \cdot a \cdot (l_2/T) \specarrow{\alabelshort[\mathtt{remove}]{a}} (l_1/T) \cdot (l_2/T)$ is in $\specAddatOne$.  Then, $(l_1/T) \cdot a \cdot (l_2/T) = (l_1 \cdot a \cdot l_2)/T$. Or we can say, $( (l_1 \cdot a \cdot l_2,T), (l_1/T) \cdot a \cdot (l_2/T) ) \in R$.

\item[-] If $(l, T) \specarrow{\alabellong[\mathtt{read}]{}{(l/T)}{}} (l,T)$ is in $\specAddatTwo$: Obviously, $l/T \specarrow{\alabellong[\mathtt{read}]{}{(l/T)}{}} l/T$ is in $\specAddatOne$ and $((l,T),(l/T)) \in R$.
\end{itemize}

Therefore, $R$ holds as required. Then, we can see that for each transitions $\abstate_0 \specarrow{\alabel_0} \abstate_1 \ldots \specarrow{\alabel_n} \abstate_n$ of $\specAddatTwo$ where each value is removed at most once in $\alabel_0 \cdot \ldots \cdot \alabel_n$, there exists transitions $\abstate'_0 \specarrow{\alabel_0} \abstate'_1 \ldots \specarrow{\alabel_n} \abstate'_n$ of $\specAddatOne$. Or we can say, for each sequence $\alabel_0 \cdot \ldots \cdot \alabel_n$ where each value is removed at most once, $\alabel_0 \cdot \ldots \cdot \alabel_n \in \specAddatTwo$ implies that $\alabel_0 \cdot \ldots \cdot \alabel_n \in \specAddatOne$. Since in the history $h$ of \autoref{fig:an example history shows that RGA is not RA-linearizable w.r.t specAddat}, each value is removed at most once, we can see that the set of possible linearization of $h$ of $\specAddatTwo$ is a subset of possible linearization of $h$ of $\specAddatOne$. Since $h$ is not \crdtlinearizable{} w.r.t $\specAddatOne$, we can see that $h$ is not \crdtlinearizable{} w.r.t $\specAddatTwo$. This completes the proof of this lemma. $\qed$

\end {proof}

\subsection{A second version of list with {\tt addAt} Interface}
\label{subsec:a second version of list with addAt interface}

As the second version of list with {\tt addAt} interface, let us introduce the interface of ~\cite{AttiyaBGMYZ16} as below:

\begin{itemize}
\setlength{\itemsep}{0.5pt}
\item[-] $\alabellong[{\tt addAt}]{a,k}{s}{}$: Inserts value $a$ into position $k$ of list of the replica state, and then returns the updated list content of the replica state. For $k$ exceeding the list size of replica state, $a$ will be inserted at the end of the list.

\item[-] $\alabellong[{\tt remove}]{a}{s}{}$: Remove value $a$ from list of the replica state, and returns the updated list content of the replica state.

\item[-] $\alabellong[{\tt read}]{}{s}{}$: Returns the list content of the replica state.
\end{itemize}

Similarly, we assume the first element of a sequence is at position $0$.

Our RGA algorithm of \sectionautorefname \ref{sec:overview} can be modified as follows for this interface:

\begin{itemize}
\setlength{\itemsep}{0.5pt}
\item[-] For {\tt addAt}: Given argument $a$ and $k$ of {\tt addAt}. Let $s = traverse(N, Tomb)$ be the list content in replica state. If $s = \epsilon$, then let $b = \circ$; else, if $\vert s \vert \geq k$, then let $b=s[k-1]$; else, let $b=s[\vert s \vert -1]$.

    Then, we work as $\alabelshort[{\tt addAfter}]{b,a}$. Especially, the effector of $\alabelshort[{\tt addAt}]{a,k}$ is the effector of $\alabelshort[{\tt addAfter}]{b,a}$. After the effector of $\alabelshort[{\tt addAfter}]{b,a}$ is applied in the replica where this operation originates, let $s'$ be the list content of replica state, and we return $s'$.

\item[-] For {\tt remove}: Given argument $a$, we work as $\alabelshort[{\tt remove}]{a}$. After the effector of $\alabelshort[{\tt remove}]{a}$ is applied in the replica where this operation originates, let $s'$ be the list content of replica state, and we return $s'$. 
\end{itemize}

\subsection{Another Sequential Specification for list with {\tt addAt} Interface}
\label{subsec:another sequential specification for list with addAt interface}

Let us introduce another sequential specification $\specAddatThree$ of list with {\tt addAt} interface that also use tombstone. Each abstract state $\abstate = (l,T)$ contains a sequence $l$ of elements of a given type and a set $T$ of elements appearing in the list. The element $l$ is the list of all input values, whether already removed or not; while $T$ stores the removed values and is used as \emph{tombstone}. The sequential specification $\specAddatThree$ of list with {\tt addAt} interface is defined by: 

\[
  \begin{array}{rcl}
    \big(\ (l_1 \cdot b \cdot l_2,T\big)\ |\ a\text{ is fresh}\ , \vert s_1 \cdot b\vert = k \big)
     & \specarrow{\alabellong[{\tt addAt}]{a,k}{s_1 \cdot b \cdot a \cdot s_2}{}}
     & (l_1 \cdot b \cdot a \cdot l_2,T)\
       \left[\begin{array}{c}
           \text{with $s_1 \cdot b \cdot s_2$ is a sub-}\\
           \text{sequence of $l_1 \cdot b \cdot l_2$}
       \end{array}\right]\\
     \big(\ (l_1 \cdot b \cdot l_2,T\big)\ |\ a\text{ is fresh}\ , \vert s \cdot b \vert <k \big)
     & \specarrow{\alabellong[{\tt addAt}]{a,k}{s \cdot b \cdot a}{}}
     & (l_1 \cdot b \cdot a \cdot l_2,T)\
       \left[\begin{array}{c}
           \text{with $s \cdot b$ is a sub-}\\
           \text{sequence of $l_1 \cdot b \cdot l_2$}
       \end{array}\right]\\
     \big(\ (l,T\big)\ |\ a\ \text{occurs in}\ l\ \big)
     & \specarrow{\alabellong[{\tt remove}]{a}{s}{}}
     & (l,T \cup \{ a \})\
       \left[\begin{array}{c}
           \text{with $s$ is a subsequence}\\
           \text{of $l$ and }a \notin s
       \end{array}\right]\\
     (l,T)
     & \specarrow{\alabellong[\mathtt{read}]{}{(l/T)}{}}
     & (l,T)\
   \end{array}
\]

When we do $\alabelshort[\mathtt{addAt}]{a,k}$ for an abstract state $(l,T)$, we do not try to check if there is a prefix of $l$ with length $k$ nor check if $\vert l/T \vert <k$. Instead, we check whether there exists a subsequence $l'_1 \cdot b \cdot l'_2$ of $l$, such that $\vert l'_1 \cdot b \vert = k$, or check whether there exists a subsequence $l' \cdot b$ of $l$, such that $\vert l' \cdot b \vert < k$. Note that here we do not check the tombstone $T$. Intuitively, here $l'_1 \cdot b \cdot l'_2$ or $l' \cdot b$ is the list content of a replica, which is a local view. A intuitively explanation of RGA not being \crdtlinearizable{} w.r.t. $\specAddatOne$ or $\specAddatTwo$ is that index is chosen locally to each replica and may be hard to have a global explanation. The fact that $\specAddatThree$ also works in a ``local'' way may be a intuitive reason of RGA being \crdtlinearizable{} w.r.t. $\specAddatThree$ (proved in the next subsection).

If $s_1 \cdot b \cdot s_2$ is a subsequence of $l_1 \cdot b \cdot l_2$ and $\vert s_1 \cdot b \vert = k$, the method $\alabelshort[\mathtt{addAt}]{a,k}$ puts $a$ immediately after $b$, and returns $s_1 \cdot b \cdot a \cdot s_2$, assuming that each value is put into list at most once.

If $s \cdot b$ is a subsequence of $l_1 \cdot b \cdot l_2$ and $\vert s \cdot b \vert < k$, the method $\alabelshort[\mathtt{addAt}]{a,k}$ puts $a$ immediately after $b$, and returns $s \cdot b \cdot a$, assuming that each value is put into list at most once.

Method $\alabelshort[\mathtt{remove}]{a}$ adds $a$ into $T$, hence removing $a$ from the list for subsequent calls to the $\mathtt{read}$ method. Finally, $\alabellong[\mathtt{read}]{}{s}{}$ returns the list content excluding any element appearing in $T$.

\subsection{RGA is \crdtlinearizable{} w.r.t $\specAddatThree$}
\label{subsec:RGA is RA-linearizable w.r.t specAddatThree}

The following lemma states that RGA is \crdtlinearizable{} w.r.t. $\specAddatThree$.

\begin{lemma}
\label{lemma:RGA is RA-linearizable w.r.t SpecAddatThree}
RGA is \crdtlinearizable{} w.r.t $\specAddatThree$.
\end{lemma}

\begin {proof}

Let us first prove $\mathsf{Commutativity}$. By the causal delivery assumption and the preconditions of ${\tt addAt}$, it cannot happen that an $\alabellong[{\tt addAt}]{b,\_}{\_}{}$ operation, which puts ${\tt b}$ after ${\tt a}$, is concurrent with an operation that adds ${\tt a}$ to the list, i.e., $\alabellong[{\tt addAt}]{a,\_}{\_}{}$. Therefore, applying effectors of two concurrent {\tt addAt} operations commute. By the causal delivery assumption and the preconditions of ${\tt addAt}$ and ${\tt remove}$, it can not happen that an $\alabellong[{\tt addAt}]{b,\_}{\_}{}$ operation adding ${\tt b}$ is concurrent with an operation $\alabellong[{\tt remove}]{b}{\_}{}$ that removes ${\tt b}$. Threfore, applying effectors of a concurrent {\tt addAt} operation and a concurrent {\tt remove} operation commute. Since applying effectors of {\tt remove} do set union to tombstone set $T$, applying effectors of concurrent {\tt remove} operations commute.

Let us prove $\mathsf{Refinement}_{\tsof{}}$. Let the refinement mapping $\refmap$ be defined as: $\refmap(N,Tomb) = (l,T)$, where $l={\tt traverse(N, \emptyset)}$, and $T={\tt Tomb}$. Then,

\begin{itemize}
\setlength{\itemsep}{0.5pt}
\item[-] The fact that effectors of ${\tt read}$ queries are simulated by the corresponding operations of the specification is straightforward.

\item[-] Concerning effectors of $\alabel = \alabellongind[{\tt addAt}]{a,k}{s}{\ats_a}{}$. we show that they are simulated by the corresponding specification operation $\alabellongind[{\tt addAt}]{a,k}{s}{}{}$ only when the timestamp $\ats_{\tt a}$ is strictly greater than all the timestamps stored in the replica state where it applies. Thus, let $({\tt N},{\tt Tomb})$ be a replica state such that $\ats < \ats_{\tt a}$ for every $\ats$ with $(\_,\ats,\_)\in {\tt N}$.

    Assume that when $\alabel$ is generated from a replica state $(N_{\alabel},Tomb_{\alabel})$ and let $s_{\alabel}  = {\tt traverse}(N_{\alabel},$ $Tomb_{\alabel})$,

    \begin{itemize}
    \setlength{\itemsep}{0.5pt}
    \item[-] If $\vert s_{\alabel} \vert \geq k$: There exists $b \in s_{\alabel}$ such that $\vert s_1 \cdot b \vert = k$, where $s_{\alabel} = s_1 \cdot b \cdot s_2$. Then, $s = s_1 \cdot b \cdot a \cdot s_2$. By the causal delivery assumption, all values of $s_{\alabel}$ is in $N$ and is in $l$. Assume that $l = l_1 \cdot b \cdot l_2$. Then, $s_1 \cdot b \cdot s_2$ is a sub-sequence of $l_1 \cdot b \cdot l_2$.

        The result of applying the effector $\effector$ corresponding to $\alabellongind[{\tt addAt}]{a,k}{s}{\ats_a}{}$ is to add ${\tt a}$ as a child of ${\tt b}$. Then, applying ${\tt traverse}$ on the new state will result in a sequence where ${\tt a}$ is placed just after ${\tt b}$ because it has the biggest timestamp among the children of ${\tt b}$ (and all the nodes in the tree ${\tt N}$). This corresponds exactly to the sequence obtained by applying the operation $\alabellongind[{\tt addAt}]{a,k}{s}{}{}$ in the context of the specification.

    \item[-] If $\vert s_{\alabel} \vert < k$: Let $b$ be the last value of $s_{\alabel}$ and assume that $s_{\alabel} = s_1 \cdot b$. Then, $s = s_1 \cdot b \cdot a$. By the causal delivery assumption, all values of $s_{\alabel}$ is in $N$ and is in $l$. Assume that $l = l_1 \cdot b \cdot l_2$. Then, $s_1 \cdot b$ is a sub-sequence of $l_1 \cdot b \cdot l_2$.

        The result of applying the effector $\effector$ corresponding to $\alabellongind[{\tt addAt}]{a,k}{s}{\ats_a}{}$ is to add ${\tt a}$ as a child of ${\tt b}$. Then, applying ${\tt traverse}$ on the new state will result in a sequence where ${\tt a}$ is placed just after ${\tt b}$ because it has the biggest timestamp among the children of ${\tt b}$ (and all the nodes in the tree ${\tt N}$). This corresponds exactly to the sequence obtained by applying the operation $\alabellongind[{\tt addAt}]{a,k}{s}{}{}$ in the context of the specification.

    \item[-] If $s_{\alabel} = \epsilon$: Let $b = \circ$. Then, $s = a$.

        The result of applying the effector $\effector$ corresponding to $\alabellongind[{\tt addAt}]{a,k}{s}{\ats_a}{}$ is to add ${\tt a}$ as a child of $b = \circ$. Then, applying ${\tt traverse}$ on the new state will result in a sequence where ${\tt a}$ is the first value because it has the biggest timestamp among all the nodes in the tree ${\tt N}$. This corresponds exactly to the sequence obtained by applying the operation $\alabellongind[{\tt addAt}]{a,k}{s}{}{}$ in the context of the specification.
    \end{itemize}

\item[-] Concerning effectors of $\alabel = \alabellongind[{\tt remove}]{a}{s}{}{}$: Assume that when $\alabel$ is generated from a replica state $(N_{\alabel},Tomb_{\alabel})$ and let $s_{\alabel}  = {\tt traverse}(N_{\alabel},Tomb_{\alabel})$. Since $a \in s_{\alabel}$, assume that $s_{\alabel} = s_{\alabel}^1 \cdot a \cdot s_{\alabel}^2$. We can see that $s = s_{\alabel}^1 \cdot s_{\alabel}^2$.

    By the causal delivery assumption we can see that all values of ${\tt N_{\alabel}}$ are in ${\tt N}$. Therefore, $a \in l$ and $s_{\alabel}^1 \cdot s_{\alabel}^2$ is a sub-sequence of $l$. The result of applying the effector $\effector$ corresponding to $\alabellongind[{\tt remove}]{a}{s}{}{}$ is to add ${\tt a}$ into $Tomb$. Then, applying ${\tt traverse}$ on the new state will result in a sequence without ${\tt a}$. This corresponds exactly to the sequence obtained by applying the operation $\alabellongind[{\tt remove}]{a}{s}{}{}$ in the context of the specification.
\end{itemize}

This completes the proof of this lemma. $\qed$
\end {proof}

\section{\crdtlin{} Proof for State-Based CRDT Implementations}
\label{sec:RA-linearizability proof for state-based CRDT implementations}

In this section, we briefly introduce the notion of state-based CRDT of \cite{ShapiroPBZ11}, and propose its semantics. Then, we propose our prove methodology. Similarly as ~\autoref{sec:proofs}, our prove methodology intends to prove a similar lemma of Lemma \ref{lem:replica_states} and $\mathsf{Refinement}$. In state-based CRDT, a message contains a updated replica state instead of information of only one operation. Therefore, our proof is more complex.

We associate each operation with a ``local'' effector, which is only used for proof. Our prove methodology is used in three forms depending on whether the arguments of ``local'' effector of each update operation is unique.

\subsection{Brief Introduction of State-Based CRDT}
\label{subsec:brief introduction of state-baed CRDT}

Let us first introduce the notion of join semilattice. Given a partial order $<$ and two values $x$ and $y$, we say that $z$ is a least upper bound of $x$ and $y$, if $x<z$, $y<z$, and there does not exists $z'$, such that $(z'<z) \wedge (x<z') \wedge (y<z')$. Let us use $x \sqcup y$ to denote the least upper bound of $x$ and $y$. A join semilattice is a partial order equipped with a least upper bound. According to the definition, it is obvious that the following properties hold:

\begin{itemize}
\setlength{\itemsep}{0.5pt}
\item[-] $x \sqcup y = y \sqcup x$,

\item[-] $x \sqcup x = x$,

\item[-] $(x \sqcup y) \sqcup z = x \sqcup (y \sqcup z)$.
\end{itemize}

Let us introduce the state-based \crdtimp{} of \cite{ShapiroPBZ11}, which is shown in Listing~\ref{lst:outline of state-based CRDT}.

Each replica contains a copy of data and gives its initial value. Both query methods and update methods can have arguments and return values, and can check pre-condition before a method is executed. Note that both query methods and update methods executes locally without synchronization with other replica. Instead, nondeterministically, state-based CRDT transmit modified payload as message between replicas. When a replica receives a message of modified payload, it calls method {\tt merge}, which takes the current payload and the payload in the message, and returns a new payload. The intuition of state-based CRDT is that, the domain of replica state is a join semilattice; the order of the join semilattice is given in {\tt compare} method; $\alabelshort[\mathtt{merge}]{x,y}$ returns the least upper bound of the join semilattice.

\begin{minipage}[t]{1.0\linewidth}
\begin{lstlisting}[frame=top,caption={Outline of state-based CRDT},
captionpos=b,label={lst:outline of state-based CRDT}]
  payload: type of payload
  initial: initial value

  //Query method
  query(arguments): returns
    pre: Precondition
    evaluate locally

  //Update method
  query(arguments): returns
    pre: Precondition
    evaluate locally

  compare(value1, value2): boolean b
    is value1 < value2 in the join semilattice?

  merge(value1, value2)
    returns the least upper bound of value1 and value2
\end{lstlisting}
\end{minipage}

\subsection{Semantics of State-Based CRDT Implementations}
\label{subsec:semantics of state-based CRDT implementations}

Given a state-based CRDT object $\aobj$, its semantics is defined as a labeled transition system (LTS) $\llbracket \aobj \rrbracket_s = (\globalstates,\acts,\aglobalstate_0,\rightarrow)$, where $\globalstates$ is a set of global configurations, $\acts$ is the set of transition labels, $\aglobalstate_0$ is the initial configuration, and $\rightarrow\subseteq \globalstates \times \acts\times \globalstates$ is the transition relation.

A global configuration $(\gstates, \avisord, \msgs)$ is a ``snapshot'' of the system that records all the operations that have been executed. 
$\gstates \in [\reps \rightarrow \localstates]$ stores the local configuration of each replica, and $\avisord\subseteq \powerset{\labels \times \labels}$ is the visibility relation. A local configuration $(\alabelset, \astate)$ contains the state $\astate$ of a replica and the set $\alabelset$ of labels of operations that are ``visible'' to current replica. $\alabel \in \alabelset$, if either $\alabel$ originates in this replica, or there exists a message $(Ls',\astate')$ that is ``aware of $\alabel$'', or we can say, $\alabel \in Ls'$, and such message has been applied in current replica state.
$\msgs$ stores the messages in current distributed system. Each message of $\msgs$ is a local configuration $(Ls,\astate)$. 
The reason for storing $Ls$ in a message is to construct the visibility relation. 

The transition labels are operations. For some fixed initial replica state $\astate_0$, the initial global configuration is defined by $\aglobalstate_0 = (\gstates_0, \emptyset, \emptyset) \in \globalstates$, where $\gstates_0$ maps each replica $\arep$ into $(\emptyset, \astate_0)$.

The transition relation between global configurations is defined as follows:

The first transition rule OPERATION describes a replica $\arep$ in state $\astate$ executing an invocation of method $\amethod$ with argument $\argv$. We use a function $\atsource$ to represent the behavior of the method. 
$\atsource(\astate,\amethod,\argv)$ stands for calling method $\amethod$ with argument $\argv$ on the replica state $\sigma$, and then results in a return value $\retv$, a new replica state $\astate'$, and possibly, a timestamp $\ats$. We assume that timestamps are consistent with the visibility relation
$\avisord$. After this transition, the local configuration $(\alabelset,\sigma)$ of $\arep$ is changed by modifying replica state into $\astate'$ and adding $\alabel$ into the set $\alabelset$ of labels; the visibility relation $\avisord$ is changed to record the fact that all operations in $\alabelset$ is now visible to $\alabel$. 

\[
  \inferrule[\text{\sc Operation}]
  {\gstates(\arep) = (\alabelset, \astate) \\ \atsource(\astate,\amethod,\argv) = (\retv,\astate',\ats) \\ \alabel = \alabelobjind{\argv}{\retv}{(i,\ats)} \\ \mathit{unique}(i) \\
  \ats\neq\bot\implies (\,\forall \alabel'\in\alabelset.\ \tsof(\alabel') < \ats\,) \\
  \forall \alabel'\in \labeldom{\avisord}.\ \tsof(\alabel') \neq \ats}
  {(\gstates, \avisord, \msgs) \xrightarrow{\alabel} (\gstates[\arep \leftarrow (\alabelset \cup \{\alabel\}, \astate')], \avisord \cup (\alabelset \times \{\alabel\}), \msgs)}
\]

The second rule GENERATE describes a replica $\arep$ sending a message. The content of this message is the local configuration $(\alabelset,\astate)$ of current replica. 
After this transition, only the set $\msgs$ of messages is changed by adding this message.

\[
  \inferrule[\text{\sc Generate}]
  {\gstates(\arep) = (\alabelset, \astate)}
  {(\gstates, \avisord, \msgs) \xrightarrow{} (\gstates, \avisord, \msgs \cup \{ (\alabelset,\astate) \} )}
\]

The third rule APPLY describes a replica $\arep$ with local configuration $(\alabelset, \astate_1)$ applying a message $(Ls,\astate_2)$. This transition results in modifying the state of $\arep$ to $\alabelshort[{\tt merge}]{\astate_1,\astate_2}$ and adding $Ls$ to the set of operations of local configuration of replica $\arep$.

\[
  \inferrule[\text{\sc Apply}]
  {\gstates(\arep) = (\alabelset, \astate_1) \\ (Ls,\astate_2) \in \msgs }
  {(\gstates, \avisord, \msgs) \xrightarrow{} (\gstates[\arep \leftarrow (\alabelset \cup Ls, \alabelshort[{\tt merge}]{\astate_1,\astate_2})], \avisord, \msgs)}
\]

It is clear from the definition of the transition relation that, $\avisord$ is a strict partial order.

A message can be applied multiple times on a same replica, or can be lost, or can be applied in any order.

\subsection{Proof Methodology for the Uniquely-Identified-Effectors}
\label{subsec:prove methodology for the uniquely-indenfied-effectors}

\noindent {\bf ``Local'' Effectors and {\tt apply} method}: In state-based CRDT, messages sent between replicas contain replica states instead of operations. To be able to prove a similar result as in Lemma \ref{lem:replica_states}, we extend each method as follows: For each method, the reading part is called a generator, and the updating part is called a ``local'' effector. Here we use ``local'' to emphasize that such effectors are only applied at the original replica, they are a proof artifact with no implications on the state-based CRDT semantics.

As in the operation based CRDTs, the ``local'' effector uses as arguments values produced by the generator. We give examples in Section~\ref{sec:implementation of state-based CRDT and their sequential specifications}. Let ${\tt arg}(\alabel)$ be the arguments that are generated by the generator of $\alabel$. We define a ``universal local effector function'' ${\tt apply}(\astate, {\tt arg}(\alabel) )$, which is a replica state transformer. ${\tt apply}(\astate, {\tt arg}(\alabel) )$ works as applying the ``local'' effector of $\alabel$ on $\astate$. The arguments of ``local'' effectors are important in our proof, and this is the reason why we write them explicitly in defining {\tt apply}.

Since the ``local'' effector contains information of only one operation, applying the ``local'' effector is normally ``simpler'' than {\tt merge}.

\noindent {\bf Three Cases}: We consider three different cases of state-based CRDT as follows:

\begin{itemize}
\setlength{\itemsep}{0.5pt}
\item[-] {\bf Uniquely-Identified-Effectors}: The argument of ``local'' effector of each update operation is unique. Moreover, there is a partial order among arguments of ``local'' effectors, and such partial order is consistent with the visibility relation. 

\item[-] {\bf Cumulative Effectors}: For each update operations $\alabel,\alabel'$, ${\tt arg}(\alabel) = {\tt arg}(\alabel')$, if and only if $\alabel$ and $\alabel'$ use a same method, same input values, same return values, and originate in a same replica. 

\item[-] {\bf Idempotent Effectors}: For each update operations $\alabel,\alabel'$, ${\tt arg}(\alabel) = {\tt arg}(\alabel')$, if and only if $\alabel$ and $\alabel'$ use a same method, same input values and same return values. 
\end{itemize} 

Uniquely-identified-effectors contains state-based multi-value register and state-based LWW-element-set. Cumulative effectors contains state-based PN-counter. Idempotent effectors contains state-based 2P-set. 

In this subsection, we propose our proof methodology of \crdtlin{} for the uniquely-identified-effectors. 

\noindent {\bf A Lemma Similar to Lemma \ref{lem:replica_states}}: We need to prove the following lemma, which states that $\astate$ can be obtained from the initial replica state by applying ``local'' effectors of visible operations in any order consistent with the visibility relation. Especially, this holds when this order is the projection of linearization into visible operations, since the linearization is consistent with the visibility relation. Here $Prop_1$, $Prop_2$, $Prop_3$, $Prop_4$ and $Prop_5$ are properties that are introduced in the latter part of this subsection. We will prove that they make Lemma \ref{lemma:replica state and message can be obtained by doing merge to local effectors for the first case} holds.

\begin{lemma}
\label{lemma:replica state and message can be obtained by doing merge to local effectors for the first case}
Let $\aobj$ be a state-based CRDT that is uniquely-identified-effectors, and satisfies properties $Prop_1$, $Prop_2$, $Prop_3$, $Prop_4$ and $Prop_5$. Given a history of $\aobj$ with a linearization of the operations in history (possibly, rewritten using a query-update rewriting $\gamma$), consistent with the visibility relation. Then, for each local configuration $(\alabelset,\astate)$ or message with content $(\alabelset,\astate)$,

\begin{align*}
\astate = {\tt apply}( {\tt apply}( \ldots {\tt apply} (\astate_0,{\tt arg}(\alabel_1)),\ldots,{\tt arg}(\alabel_k) )
\end{align*}
where $\astate_0$ is the initial state, $\alabelset = \{ \alabel_1,\ldots,\alabel_k \}$, and $\alabel_1 \cdot \ldots \cdot \alabel_k$ is consistent with the visibility relation. Especially, this holds when $\alabel_1 \cdot \ldots \cdot \alabel_k = \alinord\downarrow_{\alabelset}$.
\end{lemma}

\noindent {\bf Definition of Predicate $P1$ and Properties $Prop_1,\ldots,Prop_5$}: Let us formally introduce the necessary predicate and properties for proving Lemma \ref{lemma:replica state and message can be obtained by doing merge to local effectors for the first case}.

We define a predicate $P1$ that has the following property: $P1(\astate, {\tt arg}(\alabel))$ holds, if and only if, if $\astate = {\tt apply}( {\tt apply}( \ldots {\tt apply} (\astate_0,{\tt arg}(\alabel_1)),\ldots,{\tt arg}(\alabel_n) )$ for some operations $\alabel_1,\ldots,\alabel_n$, then for each $1 \leq i \leq n$, ${\tt arg}(\alabel)$ is not less than ${\tt arg}(\alabel_i)$.

Let us propose the following properties $Prop_1,\ldots,Prop_5$:

\begin{itemize}
\setlength{\itemsep}{0.5pt}
\item[-] $Prop_1$: If operations $\alabel$ and $\alabel'$ are concurrent, then for each replica state $\astate$, ${\tt apply}( {\tt apply}( \astate, {\tt arg}(\alabel) ),$ ${\tt arg}(\alabel') ) = {\tt apply}( {\tt apply}( \astate, {\tt arg}(\alabel') ), {\tt arg}(\alabel) )$.

\item[-] $Prop_2$: If $P1(\astate, {\tt arg}(\alabel))$ and $P1(\astate', {\tt arg}(\alabel))$ holds, then ${\tt merge}( \astate, {\tt apply}( \astate', {\tt arg}(\alabel) ) ) = {\tt apply}($ ${\tt merge}(\astate,\astate'), {\tt arg}(\alabel) )$.

\item[-] $Prop_3$: If $P1(\astate, {\tt arg}(\alabel))$ and $P1(\astate', {\tt arg}(\alabel))$ holds, then ${\tt merge}( {\tt apply}( \astate, {\tt arg}(\alabel) ), {\tt apply}( \astate', {\tt arg}(\alabel$ $) ) ) = {\tt apply}( {\tt merge}(\astate,$ $\astate'), {\tt arg}(\alabel) )$.

\item[-] $Prop_4$: ${\tt merge}(\astate_0,\astate_0) = \astate_0$. For each replica state $\astate,\astate'$, ${\tt merge}(\astate,\astate') = {\tt merge}(\astate',\astate)$.

\item[-] $Prop_5$: During the execution, if the current replica state of a replica is $\astate$, this replica does operation $\alabel$ and changes the replica state into $\astate'$. Then, ${\tt apply}(\astate, {\tt arg}(\alabel) ) = \astate'$.
\end{itemize}

\noindent {\bf Proof of Lemma \ref{lemma:replica state and message can be obtained by doing merge to local effectors for the first case}}:

\begin {proof}

We prove by induction on executions. Obvious they hold in $\aglobalstate_0$. Assume they hold along the execution $\aglobalstate_0 \xrightarrow{}^* \aglobalstate$ and there is a new transition $\aglobalstate \xrightarrow{} \aglobalstate'$. We need to prove that they still hold in $\aglobalstate'$.

\begin{itemize}
\setlength{\itemsep}{0.5pt}
\item[-] For case when replica $\arep$ do operation $\alabel$: Let $(\alabelset,\astate)$ and $(\alabelset',\astate')$ be the local configuration of replica $\arep$ of $\aglobalstate$ and $\aglobalstate'$, respectively. By the semantics we have $\alabelset' = \alabelset \cup \{ \alabel \}$. We need to prove that this lemma holds for local configuration $(\alabelset',\astate')$.

    By the induction assumption, we know that $\astate = {\tt apply}( {\tt apply}(\ldots {\tt apply} (\astate_0,{\tt arg}(\alabel_1)),\ldots,{\tt arg}($ $\alabel_n) )$, where $\alabelset = \{ \alabel_1,\ldots,\alabel_n \}$ and $\alabel_1 \cdot \ldots \cdot \alabel_n$ is consistent with the visibility relation. By $Prop_5$ we can see that $\astate' = {\tt apply}( \astate, {\tt arg}(\alabel) )$. By $Prop_1$, we can see that for each permutation $\alabel'_1 \cdot \ldots \cdot \alabel'_{n+1}$ of $\alabel_1 \cdot \ldots \cdot \alabel_n \cdot \alabel$ that is consistent with visibility relation, we have that $\astate = {\tt apply}( {\tt apply}(\ldots {\tt apply} (\astate_0,{\tt arg}(\alabel'_1)),\ldots,{\tt arg}(\alabel'_{n+1}) )$. Therefore, this lemma still holds for local configuration $(\alabelset',\astate')$.

\item[-] For case when replica $\arep$ send a message: Let $(\alabelset,\astate)$ and $(\alabelset',\astate')$ be the local configuration of replica $\arep$ of $\aglobalstate$ and $\aglobalstate'$, respectively. By the semantics we have $(\alabelset',\astate') = (\alabelset,\astate)$, and the content of the message is $(\alabelset,\astate)$. We need to prove that this lemma holds for the message with content $(\alabelset,\astate)$. This is obvious since this lemma holds for the local configuration $(\alabelset,\astate)$ by the induction assumption.

\item[-] For case when replica $\arep$ apply a message: Let $(\alabelset,\astate)$ and $(\alabelset',\astate')$ be the local configuration of replica $\arep$ of $\aglobalstate$ and $\aglobalstate'$, respectively. Let $(\alabelset'',\astate'')$ be the content of the message. We need to prove that this lemma still holds for local configuration $(\alabelset',\astate')$. By the semantics we have that $\alabelset' = \alabelset \cup \alabelset''$ and $\astate' = {\tt merge}(\astate, \astate'')$.

    By the induction assumption, we know that $\astate = {\tt apply}( {\tt apply}(\ldots {\tt apply} (\astate_0,{\tt arg}(\alabel_1)),\ldots,{\tt arg}($ $\alabel_u) )$, $\astate'' = {\tt apply}( {\tt apply}(\ldots {\tt apply} (\astate_0,{\tt arg}(\alabel''_1)),\ldots,{\tt arg}($ $\alabel''_v) )$, where $\alabelset = \{ \alabel_1,\ldots,\alabel_u \}$, $\alabelset'' = \{ \alabel''_1,\ldots,\alabel''_v \}$, and $\alabel_1 \cdot \ldots \cdot \alabel_u$ and $\alabel''_1 \cdot \ldots \cdot \alabel''_v$ are consistent with the visibility relation.

    Then, our proof proceed as follows:

    \begin{itemize}
    \setlength{\itemsep}{0.5pt}
    \item[-] If $\alabelset \neq \alabelset''$: There exists an operation $\alabel_a$ that is a maximal operation of $\alabelset \cup \alabelset''$ w.r.t the partial order of arguments of ``local'' effectors, and $(\alabel_a \in \alabelset'' \setminus \alabelset) \vee (\alabel_a \in \alabelset \setminus \alabelset'')$.

        \begin{itemize}
        \setlength{\itemsep}{0.5pt}
        \item[-] If $\alabel_a \in \alabelset'' \setminus \alabelset$: Assume that $\alabel_a = \alabel''_{idx}$. Since $\alabel''_1 \cdot \ldots \cdot \alabel''_v$ is consistent with the visibility relation, we can see that $\forall j>idx$, $\alabel''_{idx}$ and $\alabel''_j$ are concurrent. We can obtain $\alabel''_1 \cdot \ldots \cdot \alabel''_{idx-1} \cdot \alabel''_{idx+1} \cdot \ldots \cdot \alabel''_v \cdot \alabel''_{idx}$ from $\alabel''_1 \cdot \ldots \cdot \alabel''_v$ by several times of swapping pairs of operations that are concurrent and adjacent. Let $\astate''_1 = {\tt apply}( {\tt apply}(\ldots {\tt apply} (\astate_0,{\tt arg}(\alabel''_1)),\ldots, {\tt arg}(\alabel''_{idx-1})) \cdot {\tt arg}(\alabel''_{idx+1})) \cdot \ldots \cdot {\tt arg}(\alabel''_v))$. By $Prop_1$, we can see that $\astate'' = {\tt apply}( \astate''_1, {\tt arg}(\alabel''_{idx}) )$.

            Since $\alabel_a = \alabel''_{idx}$ is a maximal operation of $\alabelset \cup \alabelset''$ w.r.t the partial order of arguments of ``local'' effectors, we can see that $P1(\astate, {\tt arg}(\alabel''_{idx}))$ and $P1(\astate''_1, {\tt arg}(\alabel''_{idx}))$ hold. By $Prop_2$, we can see that $\astate' = {\tt merge}(\astate, \astate'') = {\tt merge}(\astate, {\tt apply}( \astate''_1, {\tt arg}(\alabel''_{idx}) ) ) = {\tt apply}($ ${\tt merge}(\astate,\astate''_1), {\tt arg}(\alabel''_{idx}) )$.

        \item[-] If $\alabel_a \in \alabelset \setminus \alabelset''$: By $Prop4$, we can see that $\astate' = {\tt merge}(\astate, \astate'') = {\tt merge}(\astate'',\astate)$. Then, we can similarly work as above case.
        \end{itemize}

    \item[-] If $\alabelset = \alabelset''$: There exists an operation $\alabel_a$ that is a maximal operation of $\alabelset \cup \alabelset''$ w.r.t the partial order of ``local'' effectors.

        Assume that $\alabel_a = \alabel_{idx1}$ and $\alabel_a = \alabel''_{idx2}$. Similarly, we can prove that $\astate = {\tt apply}( \astate_1, {\tt arg}(\alabel_a) )$ and $\astate'' = {\tt apply}( \astate''_1, {\tt arg}(\alabel_a) )$, where $\astate_1 = {\tt apply}( {\tt apply}(\ldots {\tt apply} (\astate_0,{\tt arg}(\alabel_1)),\ldots, {\tt arg}($ $\alabel_{idx1-1})) \cdot {\tt arg}(\alabel_{idx1+1})) \cdot \ldots \cdot {\tt arg}(\alabel_u))$ and $\astate''_1 = {\tt apply}( {\tt apply}(\ldots {\tt apply} (\astate_0,{\tt arg}(\alabel''_1)),\ldots, {\tt arg}($ $\alabel''_{idx2-1})) \cdot {\tt arg}(\alabel''_{idx2+1})) \cdot \ldots \cdot {\tt arg}(\alabel''_v))$.

        Since $\alabel_a$ is a maximal operation of $\alabelset \cup \alabelset''$w.r.t the partial order of arguments of ``local'' effectors, we can see that $P1(\astate_1, {\tt arg}(\alabel_a))$ and $P1(\astate''_1, {\tt arg}(\alabel_a))$ hold. By $Prop_3$, we can see that $\astate' = {\tt merge}(\astate, \astate'') = {\tt merge}( {\tt apply}(\astate_1, {\tt arg}(\alabel_a) ), {\tt apply}( \astate''_1, {\tt arg}($ $\alabel_a) ) ) = {\tt apply}( {\tt merge}(\astate_1,\astate''_1),$ ${\tt arg}(\alabel_a) )$.
    \end{itemize}

    By doing above process for several times, we finally can prove that, $\astate' = {\tt apply}( {\tt apply}(\ldots {\tt apply} ($ ${\tt merge}(\astate_0,\astate_0), {\tt arg}(\alabel'_1)),\ldots,{\tt arg}(\alabel'_n))$, where $\alabelset' = \{ \alabel'_1,\ldots,\alabel'_n \}$, and for each $i$, ${\arg}(\alabel'_i)$ is maximal among $\{ {\tt arg}(\alabel'_1),\ldots,{\tt arg}(\alabel'_i) \}$ w.r.t the partial order of arguments of ``local'' effectors. Let us prove that $\alabel'_1 \cdot \ldots \cdot \alabel'_n$ is consistent with the visibility relation by contradiction. Assume that there exists indexes $i1<i2$, such that $(\alabel'_{i2},\alabel'_{i1}) \in \avisord$. Since the partial order of arguments of ``local'' effectors is consistent with the visibility relation, we can see that ${\tt arg}(\alabel'_{i2}) < {\tt arg}(\alabel'_{i1})$. However, we already knows that ${\tt arg}(\alabel'_{i2})$ is maximal among $\{ {\tt arg}(\alabel'_1),\ldots,{\tt arg}(\alabel'_{i1}),\ldots,{\tt arg}(\alabel'_{i2}) \}$ w.r.t the partial order of arguments of ``local'' effectors, which is a contradiction. Therefore, $\alabel'_1 \cdot \ldots \cdot \alabel'_n$ is consistent with the visibility relation.

    By $Prop_4$, we can see that $\astate' = {\tt apply}( {\tt apply}(\ldots {\tt apply} ( \astate_0, {\tt arg}(\alabel'_1)),\ldots,$ ${\tt arg}(\alabel'_n))$. By $Prop_1$, we can see that, for each permutation $\alabel_1^1 \cdot \ldots \cdot \alabel_n^1$ of $\alabel'_1 \cdot \ldots \cdot \alabel'_n$ that is consistent with visibility relation, we have that $\astate = {\tt apply}( {\tt apply}(\ldots {\tt apply} (\astate_0,{\tt arg}(\alabel_1^1)),\ldots,{\tt arg}(\alabel_n^1) )$. Therefore, this lemma still holds for local configuration $(\alabelset',\astate')$.
\end{itemize}

This completes the proof of this lemma. $\qed$
\end {proof}

\noindent {\bf Proving Refinement}: Since our method has generator and ``local'' effector, similarly as \autoref{sec:distributed-lin}, 
when the state-based CRDT object admits timestamp-order linearizations, we prove the existence of a refinement mapping by proving $\mathsf{Refinement}_{\tsof{}}$. Then, together with Lemma \ref{lemma:replica state and message can be obtained by doing merge to local effectors for the first case}, we can prove \crdtlin{}.

For the state-based multi-value register (introduced in \autoref{subsec:state-based multi-value register, its sequential specification and its ``local'' effectors}), each {\tt write} method generates an unique version vector. Such version vector is contained in the arguments of ``local'' effectors of {\tt write}, and the partial order of version vectors are consistent with the visibility relation. The execution-order linearization is consistent with the partial order of version vectors. The state-based multi-value register admits execution-order linearizations. Then, we define $\mathsf{Refinement}_v$ as the existence of a mapping $\refmap$ such that:

\begin{itemize}
\item[Simulating effectors:] For every effector $\delta$ corresponding to an operation $\alabel$,
\begin{align*}
\forall \sigma\in\Sigma.\ P1(\astate, \alabel)  \land \sigma'=\delta(\sigma)\implies \refmap(\sigma)\specarrow{\secondrep(\gamma(\alabel))}\refmap(\sigma')
\end{align*}

\item[Simulating generators:] For every query $\amethod$, and every $\sigma \in \Sigma$,
  \[
    \atsource(\sigma,\amethod,\argv)= (\retv,\_,\_) \implies \refmap(\sigma)\specarrow{\alabel}\refmap(\sigma)
  \]
\noindent where $\alabel=\alabellong{\argv}{\retv}$. Also, for every query-update $\amethod$, and $\sigma\in\Sigma$,
\[
  \atsource(\sigma,\amethod,\argv)= (\retv,\_,\_) \implies \refmap(\sigma)\specarrow{\firstrep(\gamma(\alabel))}\refmap(\sigma).
\]
\end{itemize}

We need to prove $\mathsf{Refinement}_v$. Then, together with Lemma \ref{lemma:replica state and message can be obtained by doing merge to local effectors for the first case}, we can prove \crdtlin{}.

\subsection{Proof Methodology for the Cumulative Effectors}
\label{subsec:prove methodology for the cumulative effectors}

In this subsection, we propose our prove methodology of \crdtlin{} for the cumulative effectors. 

\noindent {\bf A Lemma Similar to Lemma \ref{lem:replica_states}}: We need to prove the following lemma, which states that $\astate$ can be obtained from the initial replica state by applying ``local'' effectors of visible operations in any order. Especially, this holds when this order is the projection of linearization into visible operations. Here $Prop'_1$, $Prop'_2$ and $Prop'_3$ are properties that are introduced in the latter part of this subsection. They make Lemma \ref{lemma:replica state and message can be obtained by doing merge to local effectors for the second case} holds.

\begin{lemma}
\label{lemma:replica state and message can be obtained by doing merge to local effectors for the second case} 
Let $\aobj$ be a state-based CRDT that is cumulative effectors, and satisfies properties $Prop'_1$, $Prop'_2$, $Prop'_3$, $Prop_4$ and $Prop_5$. Given a history of $\aobj$ with a linearization of the operations in history(possibly, rewritten using a query-update rewriting $\gamma$), consistent with the visibility relation. Then, for each local configuration $(\alabelset,\astate)$ or message with content $(\alabelset,\astate)$,

\begin{align*}
\astate = {\tt apply}( {\tt apply}( \ldots {\tt apply} (\astate_0,{\tt arg}(\alabel_1)),\ldots,{\tt arg}(\alabel_k) )
\end{align*}
where $\astate_0$ is the initial state and $\alabelset = \{ \alabel_1,\ldots,\alabel_k \}$. Especially, this holds when $\alabel_1 \cdot \ldots \cdot \alabel_k = \alinord\downarrow_{\alabelset}$.
\end{lemma}

\noindent {\bf Definition of Predicate $P2$ and Properties $Prop'_1,Prop'_2,Prop'_3$}: Let us formally introduce the necessary predicate and properties for proving Lemma \ref{lemma:replica state and message can be obtained by doing merge to local effectors for the second case}.

We define a predicate $P_2$ that that has the following property: $P2(\astate, {\tt arg}(\alabel))$ holds, if and only if, if $\astate = {\tt apply}( {\tt apply}( \ldots {\tt apply} (\astate_0,{\tt arg}(\alabel_1)),\ldots,{\tt arg}(\alabel_n) )$ for some operations $\alabel_1,\ldots,\alabel_n$, then for each $1 \leq i \leq n$, ${\tt arg}(\alabel) \neq {\tt arg}(\alabel_i)$.

Let us propose the following properties $Prop'_1$, $Prop'_2$ and $Prop'_3$:

\begin{itemize}
\setlength{\itemsep}{0.5pt}
\item[-] $Prop'_1$: For each replica state $\astate$, ${\tt apply}( {\tt apply}( \astate, {\tt arg}(\alabel) ), {\tt arg}(\alabel') ) = {\tt apply}( {\tt apply}( \astate, {\tt arg}(\alabel') ),$ ${\tt arg}(\alabel) )$.

\item[-] $Prop'_2$: If $P2(\astate, {\tt arg}(\alabel))$ and $P2(\astate', {\tt arg}(\alabel))$ hold, then ${\tt merge}( \astate, {\tt apply}( \astate', {\tt arg}(\alabel) ) ) = {\tt apply}($ ${\tt merge}(\astate,\astate'), {\tt arg}(\alabel) )$.

\item[-] $Prop'_3$: For each replica state $\astate,\astate'$ and argument $x$ of a ``local'' effector , ${\tt merge}( {\tt apply}( \astate, x ),$ ${\tt apply}( \astate', x ) = {\tt apply}({\tt merge}(\astate, \astate'), x )$. 
\end{itemize}

\noindent {\bf Proof of Lemma \ref{lemma:replica state and message can be obtained by doing merge to local effectors for the second case}}:

\begin {proof}

We prove by induction on executions. Obvious they hold in $\aglobalstate_0$. Assume they hold along the execution $\aglobalstate_0 \xrightarrow{}^* \aglobalstate$ and there is a new transition $\aglobalstate \xrightarrow{} \aglobalstate'$. We need to prove that they still hold in $\aglobalstate'$.

\begin{itemize}
\setlength{\itemsep}{0.5pt}
\item[-] For case when replica $\arep$ do operation $\alabel$: Let $(\alabelset,\astate)$ and $(\alabelset',\astate')$ be the local configuration of replica $\arep$ of $\aglobalstate$ and $\aglobalstate'$, respectively. By the semantics we have $\alabelset' = \alabelset \cup \{ \alabel \}$. We need to prove that this lemma holds for local configuration $(\alabelset',\astate')$.

    By the induction assumption, we know that $\astate = {\tt apply}( {\tt apply}(\ldots {\tt apply} (\astate_0,{\tt arg}(\alabel_1)),\ldots,{\tt arg}($ $\alabel_n) )$, where $\alabelset = \{ \alabel_1,\ldots,\alabel_n \}$. By $Prop_5$ we can see that $\astate' = {\tt apply}( \astate, {\tt arg}(\alabel) )$. By $Prop'_1$, we can see that for each permutation $\alabel'_1 \cdot \ldots \cdot \alabel'_{n+1}$ of $\alabel_1 \cdot \ldots \cdot \alabel_n \cdot \alabel$, we have that $\astate = {\tt apply}( {\tt apply}(\ldots {\tt apply} (\astate_0,{\tt arg}(\alabel'_1)),\ldots,{\tt arg}(\alabel'_{n+1}) )$. Therefore, this lemma still holds for local configuration $(\alabelset',\astate')$.

\item[-] For case when replica $\arep$ send a message: Let $(\alabelset,\astate)$ and $(\alabelset',\astate')$ be the local configuration of replica $\arep$ of $\aglobalstate$ and $\aglobalstate'$, respectively. By the semantics we have $(\alabelset',\astate') = (\alabelset,\astate)$, and the content of the message is $(\alabelset,\astate)$. We need to prove that this lemma holds for the message with content $(\alabelset,\astate)$. This is obvious since this lemma holds for the local configuration $(\alabelset,\astate)$ by the induction assumption.

\item[-] For case when replica $\arep$ apply a message: Let $(\alabelset,\astate)$ and $(\alabelset',\astate')$ be the local configuration of replica $\arep$ of $\aglobalstate$ and $\aglobalstate'$, respectively. Let $(\alabelset'',\astate'')$ be the content of the message. We need to prove that this lemma still holds for local configuration $(\alabelset',\astate')$. By the semantics we have $\alabelset' = \alabelset \cup \alabelset''$ and $\astate' = {\tt merge}(\astate, \astate'')$.

    By the induction assumption, we know that $\astate = {\tt apply}( {\tt apply}(\ldots {\tt apply} (\astate_0,{\tt arg}(\alabel_1)),\ldots,{\tt arg}($ $\alabel_u) )$, $\astate'' = {\tt apply}( {\tt apply}(\ldots {\tt apply} (\astate_0,{\tt arg}(\alabel''_1)),\ldots,{\tt arg}($ $\alabel''_v) )$, where $\alabelset = \{ \alabel_1,\ldots,\alabel_u \}$ and $\alabelset'' = \{ \alabel''_1,\ldots,\alabel''_v \}$.

    Let us ``move the local effectors of duplicate operations out''. If $\alabelset \cap \alabelset'' \neq \emptyset$: Let $\alabel_a \in \alabelset \cap \alabelset''$, and assume that $\alabel_a = \alabel_{idx1}$ and $\alabel_a = \alabel''_{idx2}$. Let $\astate_1 = {\tt apply}( {\tt apply}(\ldots {\tt apply} (\astate_0,{\tt arg}(\alabel_1)),\ldots, {\tt arg}($ $\alabel_{idx1-1})) \cdot {\tt arg}(\alabel_{idx1+1})) \cdot \ldots \cdot {\tt arg}(\alabel_u))$ and $\astate''_1 = {\tt apply}( {\tt apply}(\ldots {\tt apply} (\astate_0,{\tt arg}(\alabel''_1)),\ldots, {\tt arg}($ $\alabel''_{idx2-1})) \cdot {\tt arg}(\alabel''_{idx2+1})) \cdot \ldots \cdot {\tt arg}(\alabel''_v))$. By $Prop'_1$, we can see that $\astate = {\tt apply}( \astate_1, {\tt arg}(\alabel_{idx1}) )$ and $\astate'' = {\tt apply}( \astate''_1, {\tt arg}(\alabel''_{idx2}) )$. By $Prop'_3$, we can see that $\astate' = {\tt merge}(\astate, \astate'') = {\tt merge}( {\tt apply}(\astate_1$, ${\tt arg}(\alabel_a$ $) ), {\tt apply}( \astate''_1, {\tt arg}(\alabel_a) ) ) = {\tt apply}( {\tt merge}(\astate_1,\astate''_1), {\tt arg}(\alabel_a) )$. Assume that $\alabelset \cap \alabelset'' = \{ x_1,\ldots,x_l \}$, $\alabelset \setminus \alabelset'' = \{ y_1,\ldots,y_m \}$ and $\alabelset'' \setminus \alabelset = \{ z_1,\ldots,z_n \}$. Then, by doing above process for several times, we can obtain that $\astate' = {\tt apply}( {\tt apply}(\ldots {\tt apply} ( {\tt merge}( \astate_2, \astate''_2 ), {\tt arg}(x_1)),\ldots,{\tt arg}($ $x_l) )$, where $\astate_2 = {\tt apply}( {\tt apply}(\ldots {\tt apply} (\astate_0,{\tt arg}(y_1)),\ldots,{\tt arg}(y_m) )$ and $\astate''_2 = {\tt apply}( {\tt apply}($ $\ldots {\tt apply} (\astate_0,$ ${\tt arg}(z_1)),\ldots,{\tt arg}(z_n) )$. 
    
    Then, since $( \alabelset \setminus \alabelset'' ) \cap ( \alabelset'' \setminus \alabelset ) = \emptyset$. and the visibility relation is transitive, we can see that for each method $m$, input value $a$, return value $b$ and replica $\arep'$, it is impossible that both $( \alabelset \setminus \alabelset'' )$ and $( \alabelset'' \setminus \alabelset )$ contains operations that use method $m$, input value $a$, return value $b$ and originates in replica $\arep'$. Therefore, the arguments of``local'' effector of operations of $\alabelset \setminus \alabelset''$ is disjoint to the arguments of ``local'' effector of operations of $\alabelset'' \setminus \alabelset$. 

    Then, let us ``move the local effectors of other operations out''. By $Prop'_1$, we can see that $\astate_2 = {\tt apply}( \astate_5, {\tt arg}(y_1) )$, where $\astate_5 = {\tt apply}( {\tt apply}(\ldots$ ${\tt apply} (\astate_0,{\tt arg}(y_2)),\ldots,{\tt arg}(y_m) )$. We already know that the argument of ``local'' effector of $y_1$ is different than arguments of ``local'' effectors in $\{ y_2,\ldots,y_m\} \cup \{z_1,\ldots,z_n\}$. Therefore, $P2(\astate_5, {\tt arg}(y_1))$ and $P2(\astate''_2, {\tt arg}(y_1))$ hold. Therefore, by $Prop'_2$ and $Prop_4$, we can see that ${\tt merge}(\astate_2, \astate''_2) = {\tt merge}(\astate''_2, \astate_2) = {\tt merge}(\astate''_2, {\tt apply}( \astate_5, {\tt arg}(y_1) ) ) = {\tt apply}( {\tt merge}(\astate''_2,\astate_5), {\tt arg}(y_1) )$. By doing this process for several times, we can prove that ${\tt merge}(\astate_2,\astate''_2) = {\tt merge}(\astate''_2,\astate_2) = {\tt apply}( {\tt apply}(\ldots{\tt apply} ($ ${\tt merge}(\astate''_2,\astate_0) ,{\tt arg}(y_m)),\ldots,{\tt arg}(y_1) )$.

    Similarly, we can prove that ${\tt merge}(\astate''_2,\astate_0) = {\tt merge}(\astate_0,\astate''_2) = {\tt apply}( {\tt apply}(\ldots{\tt apply} ( {\tt merge}($ $\astate_0,\astate_0) ,{\tt arg}(z_n)),\ldots,{\tt arg}(z_1) )$. Therefore, by $Prop'_1$, we can see that, $\astate'$ can be obtained from ${\tt merge}(\astate_0,\astate_0)$ by applying the ``local'' effectors of operations of $\alabelset \cup \alabelset''$ in any order. By $Prop'_4$, we can see that, $\astate'$ can be obtained from $\astate_0$ by applying the ``local'' effectors of operations of $\alabelset \cup \alabelset''$ in any order. Therefore, this lemma still holds for local configuration $(\alabelset',\astate')$. 
\end{itemize}

This completes the proof of this lemma. $\qed$
\end {proof}

\noindent {\bf Proving Refinement}: Since our method has generator and ``local'' effector, similarly as \autoref{sec:distributed-lin}, when the state-based CRDT object admits execution-order linearizations, we prove the existence of a refinement mapping by proving $\mathsf{Refinement}$. Then, together with Lemma \ref{lemma:replica state and message can be obtained by doing merge to local effectors for the second case}, we can prove \crdtlin{}.

\subsection{Proof Methodology for Idempotent Effectors}
\label{subsec:prove methodology for idempotent effectors} 

In this subsection, we propose our proof methodology of \crdtlin{} for the idempotent effectors. 

\noindent {\bf A Lemma Similar to Lemma \ref{lem:replica_states}}: We need to prove the following lemma, which states that $\astate$ can be obtained from the initial replica state by applying ``local'' effectors of visible operations in any order. Especially, this holds when this order is the projection of linearization into visible operations. Here $Prop_6$ is a property that are introduced in the latter part of this subsection. It makes Lemma \ref{lemma:replica state and message can be obtained by doing merge to local effectors for the third case} holds.

\begin{lemma}
\label{lemma:replica state and message can be obtained by doing merge to local effectors for the third case}
Let $\aobj$ be a state-based CRDT that is idempotent effectors, and satisfies properties $Prop'_1$, $Prop'_2$, $Prop'_3$, $Prop_4$, $Prop_5$ and $Prop_6$. Given a history of $\aobj$ with a linearization of the operations in history(possibly, rewritten using a query-update rewriting $\gamma$), consistent with the visibility relation. Then, for each local configuration $(\alabelset,\astate)$ or message with content $(\alabelset,\astate)$, 

\begin{align*}
\astate = {\tt apply}( {\tt apply}( \ldots {\tt apply} (\astate_0,{\tt arg}(\alabel_1)),\ldots,{\tt arg}(\alabel_k) )
\end{align*}
where $\astate_0$ is the initial state and $\alabelset = \{ \alabel_1,\ldots,\alabel_k \}$. Especially, this holds when $\alabel_1 \cdot \ldots \cdot \alabel_k = \alinord\downarrow_{\alabelset}$.
\end{lemma}

\noindent {\bf Definition of Property $Prop_6$}: The definition of property $Prop_6$ is as follows: 

\begin{itemize}
\setlength{\itemsep}{0.5pt}
\item[-] $Prop_6$: For each replica state $\astate$ and operation $\alabel$, ${\tt apply}( {\tt apply}( \astate, {\tt arg}(\alabel) ), {\tt arg}(\alabel)) = {\tt apply}( \astate,$ ${\tt arg}(\alabel) )$. 
\end{itemize}

\noindent {\bf Proof of Lemma \ref{lemma:replica state and message can be obtained by doing merge to local effectors for the third case}}:

\begin {proof}

We prove by induction on executions. Obvious they hold in $\aglobalstate_0$. Assume they hold along the execution $\aglobalstate_0 \xrightarrow{}^* \aglobalstate$ and there is a new transition $\aglobalstate \xrightarrow{} \aglobalstate'$. We need to prove that they still hold in $\aglobalstate'$.

\begin{itemize}
\setlength{\itemsep}{0.5pt}
\item[-] For case when replica $\arep$ do operation $\alabel$: Same as that of Lemma \ref{lemma:replica state and message can be obtained by doing merge to local effectors for the second case}. 

\item[-] For case when replica $\arep$ send a message: Same as that of Lemma \ref{lemma:replica state and message can be obtained by doing merge to local effectors for the second case}. 

\item[-] For case when replica $\arep$ apply a message: Let $(\alabelset,\astate)$ and $(\alabelset',\astate')$ be the local configuration of replica $\arep$ of $\aglobalstate$ and $\aglobalstate'$, respectively. Let $(\alabelset'',\astate'')$ be the content of the message. We need to prove that this lemma still holds for local configuration $(\alabelset',\astate')$. By the semantics we have $\alabelset' = \alabelset \cup \alabelset''$ and $\astate' = {\tt merge}(\astate, \astate'')$.

    By the induction assumption, we know that $\astate = {\tt apply}( {\tt apply}(\ldots {\tt apply} (\astate_0,{\tt arg}(\alabel_1)),\ldots,{\tt arg}($ $\alabel_u) )$, $\astate'' = {\tt apply}( {\tt apply}(\ldots {\tt apply} (\astate_0,{\tt arg}(\alabel''_1)),\ldots,{\tt arg}($ $\alabel''_v) )$, where $\alabelset = \{ \alabel_1,\ldots,\alabel_u \}$ and $\alabelset'' = \{ \alabel''_1,\ldots,\alabel''_v \}$.

    Let us ``move the local effectors of duplicate operations out''. If $\alabelset \cap \alabelset'' \neq \emptyset$: Let $\alabel_a \in \alabelset \cap \alabelset''$, and assume that $\alabel_a = \alabel_{idx1}$ and $\alabel_a = \alabel''_{idx2}$. Let $\astate_1 = {\tt apply}( {\tt apply}(\ldots {\tt apply} (\astate_0,{\tt arg}(\alabel_1)),\ldots, {\tt arg}($ $\alabel_{idx1-1})) \cdot {\tt arg}(\alabel_{idx1+1})) \cdot \ldots \cdot {\tt arg}(\alabel_u))$ and $\astate''_1 = {\tt apply}( {\tt apply}(\ldots {\tt apply} (\astate_0,{\tt arg}(\alabel''_1)),\ldots, {\tt arg}($ $\alabel''_{idx2-1})) \cdot {\tt arg}(\alabel''_{idx2+1})) \cdot \ldots \cdot {\tt arg}(\alabel''_v))$. By $Prop'_1$, we can see that $\astate = {\tt apply}( \astate_1, {\tt arg}(\alabel_{idx1}) )$ and $\astate'' = {\tt apply}( \astate''_1, {\tt arg}(\alabel''_{idx2}) )$. By $Prop'_3$, we can see that $\astate' = {\tt merge}(\astate, \astate'') = {\tt merge}( {\tt apply}(\astate_1$, ${\tt arg}(\alabel_a$ $) ), {\tt apply}( \astate''_1, {\tt arg}(\alabel_a) ) ) = {\tt apply}( {\tt merge}(\astate_1,\astate''_1), {\tt arg}(\alabel_a) )$. Assume that $\alabelset \cap \alabelset'' = \{ x_1,\ldots,x_l \}$, $\alabelset \setminus \alabelset'' = \{ y_1,\ldots,y_m \}$ and $\alabelset'' \setminus \alabelset = \{ z_1,\ldots,z_n \}$. Then, by doing above process for several times, we can obtain that $\astate' = {\tt apply}( {\tt apply}(\ldots {\tt apply} ( {\tt merge}( \astate_2, \astate''_2 ), {\tt arg}(x_1)),\ldots,{\tt arg}($ $x_l) )$, where $\astate_2 = {\tt apply}( {\tt apply}(\ldots {\tt apply} (\astate_0,{\tt arg}(y_1)),\ldots,{\tt arg}(y_m) )$ and $\astate''_2 = {\tt apply}( {\tt apply}($ $\ldots {\tt apply} (\astate_0,$ ${\tt arg}(z_1)),\ldots,{\tt arg}(z_n) )$.

    Then, Let us ``move the local effectors of duplicate arguments out''. 
    
    If there exists operation $\alabel_b \in \alabelset \setminus \alabelset''$ and operation $\alabel_c \in \alabelset'' \setminus \alabelset$, such that they use a same method, same input value and same return value, and thus, have same argument of ``local'' effector: Assume that $\alabel_b = y_{idxb}$ and $\alabel_c = z_{idxc}$. Let $\astate_3 = {\tt apply}( {\tt apply}(\ldots {\tt apply} (\astate_0,{\tt arg}(y_1)),\ldots, {\tt arg}($ $y_{idxb-1})) \cdot {\tt arg}(y_{idxb+1})) \cdot \ldots \cdot {\tt arg}(y_m))$ and $\astate''_3 = {\tt apply}( {\tt apply}(\ldots {\tt apply} (\astate_0,{\tt arg}(z_1)),\ldots, {\tt arg}($ $z_{idxc-1})) \cdot {\tt arg}(z_{idxc+1})) \cdot \ldots \cdot {\tt arg}(z_n))$. By $Prop'_1$, we can see that $\astate_2 = {\tt apply}( \astate_3, {\tt arg}(y_{idxb}) )$ and $\astate''_2 = {\tt apply}( \astate''_3, {\tt arg}(z_{idxc}) )$. By $Prop'_3$ and $Prop_6$, we can see that ${\tt merge}(\astate_2, \astate''_2) = {\tt merge}( {\tt apply}(\astate_3, {\tt arg}(y_{idxb}) ), {\tt apply}( \astate''_3,$ ${\tt arg}(z_{idxc}) ) )$ $= {\tt apply}( {\tt merge}(\astate_3,\astate''_3), {\tt arg}(y_{idxb}) ) = {\tt apply}( {\tt apply}( {\tt merge}(\astate_3,\astate''_3), {\tt arg}(y_{idxb}) ),$ ${\tt arg}(z_{dixc}) )$.

    Let $S_1 = \{ \alabel \vert \alabel \in \alabelset \setminus \alabelset'', \alabelset'' \setminus \alabelset$ contains an operation with same method, same input value and same return value of $\alabel\} \cup \{ \alabel \vert \alabel \in \alabelset'' \setminus \alabelset, \alabelset \setminus \alabelset''$ contains an operation with same method, same input value and same return value of $\alabel\}$. 
    Assume that $S_1 \cup (\alabelset \cap \alabelset'')= \{ o_1,\ldots,o_r \}$, $(\alabelset \setminus \alabelset'') \setminus S_1 = \{ p_1,\ldots,p_s \}$ and $(\alabelset'' \setminus \alabelset) \setminus S_1 = \{ q_1,\ldots,q_t \}$. By doing this process for several times, we can obtain that $\astate' = {\tt apply}( {\tt apply}(\ldots {\tt apply} ( {\tt merge}( \astate_4, \astate''_4 ), {\tt arg}(o_1)),\ldots,$ ${\tt arg}(o_r) )$, where $\astate_4 = {\tt apply}( {\tt apply}(\ldots {\tt apply} (\astate_0,{\tt arg}(p_1)),\ldots,{\tt arg}(p_s) )$ and $\astate''_4 = {\tt apply}($ ${\tt apply}(\ldots {\tt apply} (\astate_0,{\tt arg}(q_1)),\ldots,{\tt arg}(q_t) )$. Moreover, we can see that, the arguments of ``local'' effector of operations of $(\alabelset \setminus \alabelset'') \setminus S_1$ is disjoint to the arguments of ``local'' effector of operations of $(\alabelset'' \setminus \alabelset) \setminus S_1$.

    Then, let us ``move the local effectors of other operations out''. By $Prop'_1$, we can see that $\astate_4 = {\tt apply}( \astate_6, {\tt arg}(p_1) )$, where $\astate_6 = {\tt apply}( {\tt apply}(\ldots$ ${\tt apply} (\astate_0,{\tt arg}(p_2)),\ldots,{\tt arg}(p_s) )$. We already know that the argument of ``local'' effector of $p_1$ is different than arguments of ``local'' effectors in $\{ p_2,\ldots,p_s\} \cup \{q_1,\ldots,q_t\}$. Therefore, $P2(\astate_6, {\tt arg}(p_1))$ and $P2(\astate''_4, {\tt arg}(p_1))$ hold. Therefore, by $Prop'_2$ and $Prop_4$, we can see that ${\tt merge}(\astate_4, \astate''_4) = {\tt merge}(\astate''_4, \astate_4) = {\tt merge}(\astate''_4, {\tt apply}(\astate_6, {\tt arg}(p_1) ) ) = {\tt apply}( {\tt merge}(\astate''_4,\astate_6), {\tt arg}(p_1) )$. By doing this process for several times, we can prove that ${\tt merge}(\astate_4,\astate''_4) = {\tt merge}(\astate''_4,\astate_4) = {\tt apply}( {\tt apply}(\ldots{\tt apply} ($ ${\tt merge}(\astate''_4,\astate_0) ,{\tt arg}(p_s)),\ldots,{\tt arg}($ $p_1) )$. Similarly, we can prove that ${\tt merge}(\astate''_4,\astate_0) = {\tt merge}(\astate_0$, $\astate''_4) = {\tt apply}( {\tt apply}(\ldots{\tt apply} ($ ${\tt merge}(\astate_0,\astate_0) ,{\tt arg}(q_t)),\ldots,{\tt arg}(q_1) )$. Therefore, by $Prop'_1$, we can see that, $\astate'$ can be obtained from ${\tt merge}(\astate_0,\astate_0)$ by applying the ``local'' effectors of operations of $\alabelset \cup \alabelset''$ in any order. By $Prop'_4$, we can see that, $\astate'$ can be obtained from $\astate_0$ by applying the ``local'' effectors of operations of $\alabelset \cup \alabelset''$ in any order. Therefore, this lemma still holds for local configuration $(\alabelset',\astate')$ in the case of $C_2$.
\end{itemize}

This completes the proof of this lemma. $\qed$
\end {proof}

\noindent {\bf Proving Refinement}: Since our method has generator and ``local'' effector, similarly as \autoref{sec:distributed-lin}, when the state-based CRDT object admits execution-order linearizations, we prove the existence of a refinement mapping by proving $\mathsf{Refinement}$. Then, together with Lemma \ref{lemma:replica state and message can be obtained by doing merge to local effectors for the third case}, we can prove \crdtlin{}.

\section{Implementations of State-Based CRDT and Their Sequential Specifications}
\label{sec:implementation of state-based CRDT and their sequential specifications}

Let us introduce state-based CRDT implementations, their sequential specifications, and their ``local'' effectors and arguments of ``local'' effectors.

\subsection{State-Based Multi-Value Register, Its Sequential Specification and Its ``Local'' Effectors}
\label{subsec:state-based multi-value register, its sequential specification and its ``local'' effectors}

\noindent {\bf Implementation}: The state-based multi-value register implementation of \cite{ShapiroPBZ11} is given in Listing~\ref{lst:state-based multi-value register}. It implements an interface with operations: ${\tt write}(a)$ and ${\tt read}$. This implementation assumes that the number of replicas are fixed. A payload $S$ is a set of $(a,V)$ pairs, where $a$ is a value and $V$ is a vector called version vector. The size of $V$ is the number of replicas in distributed system. $\alabelshort[{\tt myRep}]{}$ is a function that returns current replica identifier. A partial order among version vectors are defined as: Given version vectors $V$ and $V'$, we say that $V > V'$, if for each replica $\arep$, we have $V[\arep] \geq V'[\arep]$, and there exists replica $\arep'$, such that $V[\arep'] > V'[\arep']$.

${\tt write}(a)$ generates a new version vector $V'$, and set the payload into $\{ (a,V') \}$. ${\tt read}$ returns the set of values of multi-value register.

\begin{figure}[t]
\begin{lstlisting}[frame=top,caption={Pseudo-code of state-based multi-value register},
captionpos=b,label={lst:state-based multi-value register}]
  payload Set S
  initial S = @|$\emptyset$|@

  write(a) :
    let g = myRep()
    let @|$\mathcal{V}$|@ = @|$\{ V \vert \exists x, (x,V) \in S \}$|@
    let @|$V'$|@ = @|$[ max_{V \in \mathcal{V}} V[j]]_{j \neq g}$|@
    let @|$V'[g]$|@ = @|$max_{V \in \mathcal{V}} V[g]$|@ + 1
    S = (a,V')


  read() :
    let @|$S_1$|@  = {a @|$\vert \exists$|@ V. (a,V) @|$\in$|@ S}
    return @|$S_1$|@

  compare(@|$S_1$|@, @|$S_2$|@): boolean b
    let b  = @|$\forall (a,V) \in S_1$|@, @|$\exists (a',V') \in S_2$|@, such that @|$V \leq V'$|@
    return b

  merge(@|$S_1$|@, @|$S_2$|@): @|$S_3$|@
    let @|$S'_1$|@ = @|$\{ (a,V) \vert (a,V) \in S_1, \forall (a',V') \in S_2, \neg (V < V') \}$|@
    let @|$S'_2$|@ = @|$\{ (a,V) \vert (a,V) \in S_2, \forall (a',V') \in S_1, \neg (V < V') \}$|@
    let @|$S_3 = S'_1 \cup S'_2$|@
    return @|$S_3$|@
\end{lstlisting}
\end{figure}

The query-update rewriting rewrites $\alabelshort[{\tt write}]{a}$ into $\alabelshort[{\tt write}]{a,V'}$, where $V'$ is the version vector generated by executing $\alabelshort[{\tt write}]{a}$.

\noindent {\bf Sequential Specification $\specMVReg$}: Each abstract state $\abstate$ is a set of tuples $(a,id)$, where $a$ is a data and $id$ is an identifier. Moreover, we assume that there is a partial order on identifier. The sequential specification $\specMVReg$ of multi-value register is defined by:

\[
  \begin{array}{rcl}
    \big(\ \abstate\ |\ \mathtt{id}\ \text{is not less or equal than any identifier of } \abstate\ \big)
             & \specarrow{ \alabelshort[{\tt write}]{a,id} }
    & \abstate \cup \{ (a,\mathtt{id}) \} \setminus \{ (a',id') \in \abstate, id' < id \} \\
    \abstate
    & \specarrow{\alabellong[\mathtt{read}]{}{ S }{}}
    & \abstate\
      \begin{array}{c}
        [\text{with}\ S = \{ a\ \vert\ \exists\ \mathtt{id}, (a,\mathtt{id}) \in \abstate \}]
      \end{array}
  \end{array}
\]

Method $\alabelshort[{\tt write}]{a,id}$ puts $\{ (a,id) \}$ into the abstract state and then removes from the abstract state all pairs with a less identifier. Method ${\tt read}$ returns the value of abstract state.

\noindent {\bf The ``Local'' Effector, Arguments, Partial Order, and Predicate $P1$}: The state-based multi-value register is uniquely-identified-effectors. Given operation $\alabel = {\tt write}(a)$, which generates version vector $V$, then, ${\tt arg}(\alabel) = (a,V)$. The partial order of arguments of ``local'' effector is defined as follows: $(a_1,V_1)$ is less than $(a_2,V_2)$, if $V_1 < V_2$. The {\tt apply} method is defined as follows: Given $\astate = S$ and ${\tt arg}(\alabel) = (a,V)$. ${\tt apply}(\astate, {\tt arg}(\alabel) ) = S \cup \{ (a,V) \} \setminus \{ (a',V') \vert (a',V') \in S, V' < V \}$. Predicate $P1$ is defined as follows: $P1(\astate,(a,V))$ holds, if for each $(a',V') \in \astate$, $V$ is not less than $V'$.

\noindent {\bf Properties of Version Vectors}: The follow lemma states that the version vector of argument of ``local'' effector of each {\tt write} operation is unique, and the order between version vectors is consistent with the visibility relation. 

\begin{lemma}
\label{lemma:version vector is similar as timestamp}
The version vector is unique for each {\tt write} operation, and the order of version vector is consistent with the visibility relation.
\end{lemma}

The follow lemma states that version vectors of concurrent {\tt write} operations are incomparable.

\begin{lemma}
\label{lemma:concurrent implies the version is incomparable}
Given two {\tt write} operation $\alabel_1$ and $\alabel_2$, assume $(a_1,V_1)$ and $(a_2,V_2)$ is the argument of ``local'' effector for $\alabel_1$ and $\alabel_2$, respectively, and assume that $\alabel_1$ and $\alabel_2$ are concurrent. Then, $\neg (V_1 < V_2 \vee V_2 < V_1)$.
\end{lemma}

\noindent Proof of Lemma \ref{lemma:version vector is similar as timestamp}:

\begin {proof}

Let us propose $fact1$, $fact2$, $fact3$ and $fact4$:

\begin{itemize}
\setlength{\itemsep}{0.5pt}
\item[-] $fact1$: Given two {\tt write} operation $\alabel_1$ and $\alabel_2$, assume $(a_1,V_1)$ and $(a_2,V_2)$ is the argument of ``local'' effector for $\alabel_1$ and $\alabel_2$, respectively, and assume that $(\alabel_1,\alabel_2) \in \avisord$. Then, $V_1 < V_2$.

\item[-] $fact2$: Given two {\tt write} operation $\alabel_1$ and $\alabel_2$, assume $(a_1,V_1)$ and $(a_2,V_2)$ is the argument of ``local'' effector for $\alabel_1$ and $\alabel_2$, respectively, and assume that $\alabel_1$ and $\alabel_2$ are concurrent. Then, $\neg (V_1 < V_2 \vee V_2 < V_1)$.

\item[-] $fact3$: Assume $(a,V)$ is the argument of ``local'' effector of a $\alabelshort[{\tt write}]{a}$ operation $\alabel$. Then, for each replica $\arep$, $V[\arep]$ = $\vert \{ \alabel' = \alabelshort[{\tt write}]{\_}, \alabel'$ originate in replica $\arep,  (\alabel',\alabel) \in \avisord \vee \alabel' = \alabel \} \vert$.

\item[-] $fact4$: Assume $(\alabelset,S)$ is the local configuration of a replica or the content of a message. Then, $S$ =  $\{ (a,V) \vert \exists \alabel, (a,V)$ is the argument of ``local'' effector of $\alabel, \alabel$ is maximal w.r.t $\avisord$ among {\tt write} operations in $\alabelset \}$.
\end{itemize}

Assume that above four facts are proved. Then, since the visibility relation is a partial order, by $fact3$ we can see that the version vector is unique for each {\tt write} operation, and by $fact1$ we can see that the order of version vector is consistent with the visibility relation.

Let us prove $fact1$. Given two {\tt write} operation $\alabel_1$ and $\alabel_2$, assume $(a_1,V_1)$ and $(a_2,V_2)$ is the argument of ``local'' effector for $\alabel_1$ and $\alabel_2$, respectively, and assume that $(\alabel_1,\alabel_2) \in \avisord$. Then, there are two possibilities.

\begin{itemize}
\setlength{\itemsep}{0.5pt}
\item[-] If $\alabel_1$ and $\alabel_2$ originate in a same replica, then $fact1$ obviously holds according to the implementation.

\item[-] Otherwise, there exists messages $(\alabelset_1,\astate_1),\ldots,(\alabelset_k,\astate_k)$ and replica $\arep_1,\ldots,\arep_{k+1}$, such that (1) $\alabel_1$ originates in replica $\arep_1$, $\alabel_2$ originates in replica $\arep_{k+1}$, (2) for each $1 \leq i \leq k$, message $(\alabelset_i,\astate_i)$ is generated by replica $\arep_i$, (3) the time point of originating $\alabel_1$ is earlier than the time point of generating message $(\alabelset_1,\astate_1)$, (4) for each $1 \leq i \leq k$, message $(\alabelset_i,\astate_i)$ is received by replica $\arep_{i+1}$, and such time point is earlier than the time point of generating message $(\alabelset_{i+1},\astate_{i+1})$, and (5) the time point of receiving message $(\alabelset_k,\astate_k)$ by replica $\arep_{k+1}$ is earlier than the time point of originating $\alabel_2$.

    According to the implementation, we can see that for each message $(\alabelset_i,\astate_i)$ with $1 \leq i \leq k$, there exists $(a'_i,id'_i) \in \astate_i$, such that $V_1 \leq id'_i$. According to the implementation, for each $(a',id') \in \astate_k$, we have $id'<V_2$. Therefore, we can see that $fact1$ still holds.
\end{itemize}

This completes the proof of $fact1$.

Let us prove that $fact2$, $fact3$ and $fact4$ are inductive invariant. We prove by induction on executions. Obvious they hold in $\aglobalstate_0$. Assume they hold along the execution $\aglobalstate_0 \xrightarrow{}^* \aglobalstate$ and there is a new transition $\aglobalstate \xrightarrow{} \aglobalstate'$. 

\begin{itemize}
\setlength{\itemsep}{0.5pt}
\item[-] For case when replica $\arep$ do $\alabelshort[{\tt write}]{a}$: let $(a,V')$ be the argument of ``local'' effector of $\alabel$. Let $Lc = (\alabelset,S)$ and $Lc' = (\alabelset',S')$ be the local configuration of replica $\arep$ of $\aglobalstate$ and $\aglobalstate'$, respectively. Obviously $S' = \{ (a,V') \}$ and $\alabelset' = \alabelset \cup \{ \alabel \}$. We need to prove that $fact4$ still holds for the local configuration $Lc'$, $fact3$ still holds for the argument of ``local'' effector of $\alabel$, and $fact2$ still holds in $\aglobalstate'$.

    Since $\alabel$ is greater than any operations of $\alabelset$ w.r.t the visibility relation, $fact4$ still holds for the local configuration $Lc'$.

    Let $\mathcal{V} = \{ V \vert (\_,V) \in S \}$ be the set of version vectors of $S$. By $fact4$ of local configuration $Lc$, $fact3$ of arguments of ``local'' effectors of $S$, and the transitivity of the visibility relation, we can see that, for each replica $\arep'$, $max_{V_1 \in \mathcal{V}} V_1[\arep']$ is the number of {\tt write} operations originates in replica $\arep'$ and is in $\alabelset$. It is obvious that, for each replica $\arep' \neq \arep$, $V'[\arep'] = max_{V_1 \in \mathcal{V}} V_1[\arep']$, and $V'[\arep] = max_{V_1 \in \mathcal{V}} V_1[\arep] +1$. Therefore, $fact3$ still holds for the argument of ``local'' effector of $\alabel$.

    We prove $fact2$ by contradiction. Since $fact2$ holds in $\aglobalstate$, it is easy to see that the counter-example of $fact2$ must contain operation $\alabel$. Assume there exists an operation $\alabel_2$ that is generated during $\aglobalstate_0 \xrightarrow{}^* \aglobalstate'$, such that the argument of ``local'' effector of $\alabel_2$ is $(a_2,V_2)$, $\alabel$ and $\alabel_2$ are concurrent, and $V' < V_2$. We can see that $V'[\arep] \leq V_2[\arep]$. By $fact3$ of the arguments of ``local'' effectors $(a,V')$ and $(a_2,V_2)$, and the transitivity of the visibility relation, we can see that $(\alabel,\alabel_2) \in \avisord$, which contradicts the assumption that $\alabel$ and $\alabel_2$ are concurrent.

    Similarly, assume there exists an operation $\alabel_2$ that is generated during $\aglobalstate_0 \xrightarrow{}^* \aglobalstate'$, such that the argument of ``local'' effector of $\alabel_2$ is $(a_2,V_2)$, $\alabel$ and $\alabel_2$ are concurrent, and $V_2 < V'$. Assume that $\alabel_2$ is originated in replica $\arep_2$. We can see that $V_2[\arep_2] \leq V'[\arep_2]$. By $fact3$ of the arguments of ``local'' effectors $(a,V')$ and $(a_2,V_2)$, and the transitivity of the visibility relation, we can see that $(\alabel_2,\alabel) \in \avisord$, which contradicts the assumption that $\alabel$ and $\alabel_2$ are concurrent. Therefore, $fact2$ still holds in $\aglobalstate'$.

\item[-] For case when replica $\arep$ generate a message: Let $Lc = (\alabelset,S)$ and $Lc' = (\alabelset',S')$ be the local configuration of replica $\arep$ of $\aglobalstate$ and $\aglobalstate'$, respectively. By the semantics we have $(\alabelset',S') = (\alabelset,S)$ and the content of the message is $(\alabelset,S)$. We only need to prove that $fact4$ still holds for the message $(\alabelset,S)$. This holds obviously since by induction assumption we already knows that $fact4$ holds for the local configuration $(\alabelset,S)$.

\item[-] For case when replica $\arep$ apply a message: Let $Lc = (\alabelset,S)$ and $Lc' = (\alabelset',S')$ be the local configuration of replica $\arep$ of $\aglobalstate$ and $\aglobalstate'$, respectively. Assume that the message content is $(\alabelset'',S'')$. By the semantics we have $\alabelset' = \alabelset \cup \alabelset''$ and $S' = \alabelshort[{\tt merge}]{S,S''}$. We need to prove that $fact4$ still holds for the local configuration $(\alabelset',S')$.

    Let us prove by contradiction. Assume that $fact4$ does not holds for the local configuration $(\alabelset',S')$.

    It is easy to see that, the arguments of ``local'' effectors of maximum of {\tt write} operation w.r.t. $\avisord$ in $\alabelset'$ must be in that of $\alabelset$ or $\alabelset''$. Then, there are two possibilities:

    \begin{itemize}
    \setlength{\itemsep}{0.5pt}
    \item[-] There exists a argument of ``local'' effector $(a_1,V_1) \in S$, such that $(a_1,V_1) \in S'$, and there exists $(a_2,V_2) \in S'$, such that $(a_1,V_1)$ and $(a_2,V_2)$ is the argument of ``local'' effector of some operation $\alabel_1$ and $\alabel_2$, respectively, and $(\alabel_1,\alabel_2) \in \avisord$. By $fact1$ and $fact2$, we can see that $V_1 < V_2$.

        By $fact1$, $fact2$, $fact3$ and $fact4$, we can see that the version vectors of $S$ are pair-wise incomparable, and the version vectors of $S''$ are pair-wise incomparable.
        Then, we can see that $(a_1,V_1) \in S \wedge (a_2,V_2) \in S''$. According to the {\tt merge} method, it is easy to see that $(a_1,V_1) \notin S'$, which contradicts that $(a_1,V_1) \in S'$.

        For the case when there exists a local effector $(a_1,V_1) \in S''$, such that $(a_1,V_1) \in S'$, and there exists $(a_2,V_2) \in S'$, such that $(a_1,V_1)$ and $(a_2,V_2)$ is the argument of ``local'' effector of some operation $\alabel_1$ and $\alabel_2$, respectively, and $(\alabel_1,\alabel_2) \in \avisord$. Similarly, we can see that $(a_1,V_1) \in S'' \wedge (a_2,V_2) \in S$. According to the {\tt merge} method, it is easy to see that $(a_1,V_1) \notin S'$, which contradicts that $(a_1,V_1) \in S'$.

    \item[-] There exists an operation $\alabel_1$ with its argument $(a_1,V_1)$ of ``local'' effector, such that $\alabel_1$ is maximal of {\tt write} operation w.r.t. $\avisord$ in $\alabelset'$, and $(a_1,V_1) \notin S'$. Then, there must exists an operation $\alabel_2$ with its argument $(a_2,V_2)$ of ``local'' effector, such that either $(a_1,V_1) \in S \wedge (a_2,V_2) \in S'' \wedge V_1 < V_2$, or $(a_1,V_1) \in S'' \wedge (a_2,V_2) \in S \wedge V_1 < V_2$. By $fact1$ and $fact2$, in each case, $\alabel_1$ is not maximal of {\tt write} operation w.r.t. $\avisord$ in $\alabelset'$, which is a contradiction.
    \end{itemize}

    Therefore, $fact4$ still holds for the local configuration $(\alabelset',S')$.
\end{itemize}

This completes the proof of this lemma. $\qed$
\end {proof}

\noindent Proof of Lemma \ref{lemma:concurrent implies the version is incomparable}: Obviously from the proof of Lemma \ref{lemma:version vector is similar as timestamp}.

\subsection{State-Based Last-Writer-Win-Element-Set (LWW-Element-Set), Its Sequential Specification and Its ``Local'' Effectors}
\label{subsec:state-based last-writer-win-element set (LWW-element-set), its sequential specification and its ``local'' effectors}

\noindent {\bf Implementation}: The state-based last-writer-win-element-set (LWW-element-set) of \cite{ShapiroPBZ11} is given in Listing~\ref{lst:state-based LWW-element-set}. It
implements a set interface with operations: ${\tt add}(a)$, ${\tt remove}(a)$ and ${\tt read}$. A payload $(A,R)$ contains a set $A$ that records pairs of inserted values and their timestamp, and a set $R$ that records the pairs of removed values and and their timestamp, and $R$ is used as \emph{tombstone}.

A value $b$ is ``in the set'', if $(b,\ats_b) \in A$ for some timestamp $\ats_b$, and there does not exists $(b,\ats'_b) \in R$, such that $\ats_b < \ats'_b$. ${\tt write}(a)$. ${\tt write}(a)$ generates a new timestamp $\ats$ and puts $(a,\ats)$ into set $A$. ${\tt remove}(a)$ generates a new timestamp $\ats$ and puts $(a,\ats)$ into set $R$. {\tt read} returns the set content.

\begin{figure}[t]
\begin{lstlisting}[frame=top,caption={Pseudo-code of state-based LWW-element-set},
captionpos=b,label={lst:state-based LWW-element-set}]
  payload Set A, Set R
  initial @|$\emptyset$|@, @|$\emptyset$|@

  add(a) :
    let @|$\ats$|@ = getTimestamp()
    A = A @|$\cup \{ (a,\ats) \}$|@

  remove(a) :
    let @|$\ats$|@ = getTimestamp()
    R = R @|$\cup \{ (a,\ats) \}$|@

  read() :
    let S = @|$\{ b \vert \exists (b,\ats_b) \in A, \ R$|@ does not contain @|$(b,\_)$|@, or @|$\forall (b,\ats'_b) \in R, \ats'_b < \ats_b \}$|@
    return S

  compare(@|$(A_1,R_1)$|@, @|$(A_2,R_2)$|@): boolean b
    let b  = @|$(A_1 \subseteq A_2) \wedge (R_1 \subseteq R_2)$|@
    return b

  merge(@|$(A_1,R_1)$|@, @|$(A_2,R_2)$|@): @|$(A_3,R_3)$|@
    let @|$A_3 = A_1 \cup A_2$|@
    let @|$R_3 = R_1 \cup R_2$|@
    return @|$(A_3,R_3)$|@
\end{lstlisting}
\end{figure}

\noindent {\bf Sequential Specification $\specLWWSet$}: Each abstract state $\abstate$ is a set $S$ of values. The sequential specification $\specLWWSet$ of set is defined by:
\[
  \begin{array}{rcl}
    S &
               \specarrow{\alabelshort[\mathtt{add}]{a}}
    & S \cup \{ a \} \\
    S &
               \specarrow{\alabelshort[\mathtt{remove}]{a}}
    & S \setminus \{a\} \\
    S
    & \specarrow{\alabellong[\mathtt{read}]{}{ S }{}}
    & S
  \end{array}
\]

Method $\alabelshort[{\tt add}]{a}$ puts $a$ into $S$. Method $\alabelshort[{\tt remove}]{a}$ removes $a$ from $S$. Method $\alabellong[{\tt read}]{}{S}{}$ returns the contents of the set.

\noindent {\bf The ``Local'' Effector, Arguments, Partial Order, and Predicate $P1$}: The state-based LWW-element-set is uniquely-identified-effectors. Given operation $\alabel = \alabellongind[{\tt add}]{a}{}{\ats_a}{}$, then, ${\tt arg}(\alabel) = (add,a,\ats_a)$; given operation $\alabel = \alabellongind[{\tt remove}]{a}{}{\ats_a}{}$, then, ${\tt arg}(\alabel) = (rem,a,\ats_a)$. The partial order of arguments of ``local'' effector is defined as follows: $(\_,\_,\ats)$ is less than $(\_,\_,\ats')$, if $\ats < \ats'$. The {\tt apply} method is defined as follows: Given $\astate = (A,R)$ and ${\tt arg}(\alabel) = (add,a,\ats_a)$, ${\tt apply}(\astate, {\tt arg}(\alabel) ) = (A \cup \{ (a,\ats_a) \},R)$. Given $\astate = (A,R)$ and ${\tt arg}(\alabel) = (rem,a,\ats_a)$, ${\tt apply}(\astate, {\tt arg}(\alabel) ) = (A,R \cup \{ (a,\ats_a) \})$. Predicate $P1$ is defined as follows: $P1((A,R),(\_,\_,\ats))$ holds, if for each $(\_,\ats') \in A \cup R$, $\ats$ is not less than $\ats'$.

\subsection{State-Based PN-Counter, Its Sequential Specification and Its ``Local'' Effectors}
\label{subsec:state-based PN-counter, its sequential specification and its ``local'' effectors}

\noindent {\bf Implementation}: The state-based PN-counter implementation of \cite{ShapiroPBZ11} is given in Listing~\ref{lst:state-based PN-counter}. It implements a counter interface with operations: ${\tt inc}()$, ${\tt dec}()$ and ${\tt read}$. This implementation assumes that the number of replicas are fixed. A payload $(P,N)$ contains a vector $P$ and a vector $N$. The size of $P$ and $N$ is the number of replicas in distributed system. $\alabelshort[{\tt myRep}]{}$ is a function that returns current replica identifier. Let us use $\vec 0$ to indicate the vector that maps each replica identifier to $0$, and the initial value of $P$ and $N$ is $\vec 0$.

The counter value is $\Sigma_{\arep} P[\arep] - \Sigma_{\arep} N[\arep]$. ${\tt inc}()$ that originates in replica $\arep$ increases $P[\arep]$ by $1$. ${\tt dec}()$ that originates in replica $\arep$ increases $N[\arep]$ by $1$. ${\tt read}$ returns the counter value.

\begin{figure}[t]
\begin{lstlisting}[frame=top,caption={Pseudo-code of state-based PN-counter},
captionpos=b,label={lst:state-based PN-counter}]
  payload Vector P, Vector N
  initial @|$\vec 0$|@, @|$\vec 0$|@

  inc() :
    let g = myRep()
    @|$P$|@ = @|$P[ g: P[g] + 1]$|@

  dec() :
    let g = myRep()
    @|$N$|@ = @|$N[ g: N[g] + 1]$|@

  read() :
    let c  = @|$\Sigma_{\arep} P[\arep]$|@ - @|$\Sigma_{\arep} N[\arep]$|@
    return c

  compare(@|$(P_1,N_1)$|@, @|$(P_2,N_2)$|@): boolean b
    let b  = @|$\forall \arep, (P_1[\arep] \leq P_2[\arep]) \wedge (N_1[\arep] \leq N_2[\arep])$|@
    return b

  merge(@|$(P_1,N_1)$|@, @|$(P_2,N_2)$|@): @|$(P_3,N_3)$|@
    let @|$P_3$|@ = @|$\forall \arep, P_3[\arep] = max(P_1[\arep], P_2[\arep])$|@
    let @|$N_3$|@ = @|$\forall \arep N_3[\arep] = max(N_1[\arep], N_2[\arep])$|@
    return @|$(P_3,N_3)$|@
\end{lstlisting}
\end{figure}

\noindent {\bf Sequential Specification}: PN-counter uses the sequential specification $\specCounter$.

\noindent {\bf The ``Local'' Effector, Arguments, and Predicate $P2$}: The state-based PN-counter is cumulative effectors. Given operation $\alabel = {\tt inc}()$ that originates in replica $\arep$, then, ${\tt arg}(\alabel) = (inc, \arep)$; given operation $\alabel = {\tt dec}()$ that originates in replica $\arep$, then, ${\tt arg}(\alabel) = (dec, \arep)$. The {\tt apply} method is defined as follows: Given $\astate = (P,N)$ and ${\tt arg}(\alabel) = (inc, \arep)$, ${\tt apply}(\astate, {\tt arg}(\alabel) ) = (P[\arep \leftarrow P[\arep]+1],N)$. Given $\astate = (P,N)$ and ${\tt arg}(\alabel) = (dec,\arep)$, ${\tt apply}(\astate, {\tt arg}(\alabel) ) = (P,N[\arep \leftarrow N[\arep]+1])$. Predicate $P2$ is defined as follows: Given $\astate = (P,N)$ and ${\tt arg}(\alabel) = (inc,\arep)$, $P2(\astate,{\tt arg}(\alabel))$ holds, if $P[\arep] = 0$; Given $\astate = (P,N)$ and ${\tt arg}(\alabel) = (dec,\arep)$, $P2(\astate,{\tt arg}(\alabel))$ holds, if $N[\arep] = 0$.

\subsection{State-Based Two-Phase Set (2P-Set), Its Sequential Specification and Its ``Local'' Effectors}
\label{subsec:state-based two-phase set (2P-set), its sequential specification and its ``local'' effectors}

\noindent {\bf Implementation}: The state-based 2P-set implementation of \cite{ShapiroPBZ11} is given in Listing~\ref{lst:state-based 2P-set}. It implements a set interface with operations: ${\tt add}(a)$, ${\tt remove}(a)$ and ${\tt read}$. A payload $(A,R)$ contains a set $A$ to record the inserted values, and a set $R$ to store the removed values and is used as \emph{tombstone}. Adding or removing a value twice has no effect, nor does adding an element that has already been removed. Therefore, we assume that the users guarantee that, in each execution, a value will not be added twice.

A value $b$ is ``in the set'', if $b \in A \setminus R$. ${\tt add}(a)$ puts $a$ into set $A$. ${\tt remove}(a)$ puts $a$ into set $R$. ${\tt read}$ returns the set content.

\begin{figure}[t]
\begin{lstlisting}[frame=top,caption={Pseudo-code of state-based 2P-set},
captionpos=b,label={lst:state-based 2P-set}]
  payload Set A, Set R
  initial @|$\emptyset$|@, @|$\emptyset$|@

  add(a) :
    A = A @|$\cup \ \{ a \}$|@

  remove(a) :
    precondition :  @|$a \in A \wedge a \notin R$|@
    R = R @|$\cup \ \{ a \}$|@

  read() :
    let s = A @|$\,\setminus\,$|@ R
    return s

  compare(@|$(A_1,R_1)$|@, @|$(A_2,R_2)$|@): boolean b
    let b  = @|$(A_1 \subseteq A_2) \wedge (R_1 \subseteq R_2)$|@
    return b

  merge(@|$(A_1,R_1)$|@, @|$(A_2,R_2)$|@): @|$(A_3,R_3)$|@
    let @|$A_3$|@ = @|$A_1 \cup A_2$|@
    let @|$R_3$|@ = @|$R_1 \cup R_2$|@
    return @|$(A_3,R_3)$|@
\end{lstlisting}
\end{figure}

\noindent {\bf Sequential Specification}: 2P-set uses the sequential specification $\specLWWSet$.

\noindent {\bf The ``Local'' Effector, Arguments, and Predicate $P2$}: The state-based LWW-element-set is idempotent effectors. Given operation $\alabel = {\tt add}(a)$, then, ${\tt arg}(\alabel) = (add,a)$; given operation $\alabel = {\tt remove}(a)$, then, ${\tt arg}(\alabel) = (rem,a)$. The {\tt apply} method is defined as follows: Given $\astate = (A,R)$ and ${\tt arg}(\alabel) = (add,a)$, ${\tt apply}(\astate, {\tt arg}(\alabel) ) = (A \cup \{ a \},R)$. Given $\astate = (A,R)$ and ${\tt arg}(\alabel) = (rem,a)$, ${\tt apply}(\astate, {\tt arg}(\alabel) ) = (A,R \cup \{ a \})$. Predicate $P2$ is defined as follows: Given $\astate = (A,R)$ and ${\tt arg}(\alabel) = (add,a)$, $P2(\astate,{\tt arg}(\alabel))$ holds, if $a \notin A$; Given $\astate = (A,R)$ and ${\tt arg}(\alabel) = (rem,a)$, $P2(\astate,{\tt arg}(\alabel))$ holds, if $a \notin R$.

\section{Structure of our Proof Code Files}
\label{sec:structure of our proof code file}

Proof of each operation-based CRDT $c$ is divided into two files. ${\tt shortname}(c)\_{\tt Ref}\_{\tt Boogie.bpl}$ proves that all the effectors and generators of $c$ preserves the refinement relation. Whereas, ${\tt shortname}($ $c)\_{\tt Com}\_{\tt Boogie.bpl}$ shows that the effectors of $c$ commute with respect to each other.

For the state based CRDTs, proof is divided into two files as well. ${\tt shortname}(c)\_{\tt Prop}\_{Boogie.bpl}$ proves that $c$ satisfies properties $Prop_1$ (rep., $Prop'_1$), $Prop_2$ (resp., $Prop'_2$), $Prop_3$ (resp., $Prop'_3$), $Prop_4$ and $Prop_5$ (or to $Prop_6$ if $c$ is of second case and satisfies $C1$). On the other hand ${\tt shortname}(c)\_{\tt Ref}\_$ ${\tt Boogie.bpl}$ again shows that all the effectors and generators of $c$ preserves the refinement relation.

${\tt shortname}(c)$ is a function that shortens the name of the CRDTs for using in the files. It is defined as:

\begin{itemize}
\setlength{\itemsep}{0.5pt}
\item[-] Two-Phase Set $\Rightarrow$ 2PSet,

\item[-] Counter $\Rightarrow$ Ctr,

\item[-] Last-Writer-Win Register $\Rightarrow$ LWWReg

\item[-] Last-Writer-Win-Element-Set $\Rightarrow$ LWWSet

\item[-] Multi-Value Register $\Rightarrow$ MVReg

\item[-] OR Set $\Rightarrow$ ORSet

\item[-] PN-Counter $\Rightarrow$ PNCounter

\item[-] RGA $\Rightarrow$ RGA

\item[-] Wooki $\Rightarrow$ Wookie
\end{itemize}

\end{document}